\documentclass[12pt]{article}
\usepackage{graphicx}
\usepackage{amsmath}
\usepackage{fullpage}
\usepackage{amsfonts}
\usepackage{amssymb}
\usepackage{dsfont}
\usepackage{tensor}
\usepackage{amsthm}
\usepackage{subcaption}
\usepackage{xcolor}
\newtheorem{Definition}{Definition}
\newtheorem{Lemma}{Lemma}
\newtheorem{Theorem}{Theorem}

\begin{document}
\title{Coherence in logical quantum channels}
\author{Joseph K. Iverson and John Preskill\\ \\ \textit{Institute for Quantum Information and Matter} \\ 
\textit{and}\\
\textit{Walter Burke Institute for Theoretical Physics} \\
\textit{California Institute of Technology, Pasadena CA 91125, USA}}
\date{23 April 2020}
\maketitle
\begin{abstract}
We study the effectiveness of quantum error correction against coherent noise. Coherent errors (for example, unitary noise) can interfere constructively, so that in some cases the average infidelity of a quantum circuit subjected to coherent errors may increase quadratically with the circuit size; in contrast, when errors are incoherent (for example, depolarizing noise), the average infidelity increases at worst linearly with circuit size. We consider the performance of quantum stabilizer codes against a noise model in which a unitary rotation is applied to each qubit, where the axes and angles of rotation are nearly the same for all qubits. 
In particular, we show that for the toric code subject to such independent coherent noise, and for minimal-weight decoding, the logical channel after error correction becomes increasingly incoherent as the length of the code increases, provided the noise strength decays inversely with the code distance. 
A similar conclusion holds for weakly correlated coherent noise. Our methods can also be used for analyzing the performance of other codes and fault-tolerant protocols against coherent noise. 
However, our result does not show that the coherence of the logical channel is suppressed in the more physically relevant case where the noise strength is held constant as the code block grows, and we recount the difficulties that prevented us from extending the result to that case.
Nevertheless our work supports the idea that fault-tolerant quantum computing schemes will work effectively against coherent noise, providing encouraging news for quantum hardware builders who worry about the damaging effects of control errors and coherent interactions with the environment. 
\end{abstract}
\pagebreak
\tableofcontents
\newpage

\section{Introduction}
Although there is no rigorous proof, much evidence supports the widely held belief that an ideal noiseless quantum computer would be able to solve problems that are intractable for classical digital computers. But in the real world, quantum computers are noisy. We therefore expect that quantum error correction will be needed to overcome the noise and reliably operate a large-scale quantum computer that can solve hard problems. Fortunately, the accuracy threshold theorem for quantum computation establishes that quantum computing is scalable, assuming that the noise is neither too strong nor too strongly correlated \cite{knillLaflammeZurek98, aharanovBen-Or99, gottesman97, aliferisGottesmanPreskill06, Raussendorf07}. 

Until we try it on a real device, though, we won't know for sure whether realistic noise is sufficiently benign for quantum error correction to work effectively. A general noise channel acting on $n$ qubits is extremely complex when $n$ is large, so it will not be practical to fully characterize the noise in a complex quantum device using any feasible experimental protocol. A commonly used metric for the performance of single-qubit and two-qubit quantum gates is the ``average infidelity'' $r=1-F$, where $F$ is the fidelity of the output from the gate relative to the output of an ideal gate, averaged uniformly over all possible input states. This quantity $r$ has the great virtue that it can be feasibly measured using randomized benchmarking  \cite{EmersonBenchmark2005,KnillBenchmark2008}, but as a characterization of the noise strength it has shortcomings. Assuming an uncorrelated noise model, threshold theorems guarantee scalability if a different metric, the diamond distance $D_\diamond$, is less than a critical value. Here $D_\diamond$ denotes the deviation of the noisy gate from the ideal gate as measured by the diamond norm. For an incoherent noise channel like a Pauli channel, the diamond distance $D_\diamond$ is equal to the average infidelity $r$; in contrast, for a highly coherent channel, $D_\diamond$ scales like the square root of $r$. If we know only $r$, and have no information about the coherence of the noise, we cannot estimate $D_\diamond$ accurately, and therefore cannot easily make sound predictions about how effectively any error-correcting code will combat the noise \cite{sanders2015bounding,kueng2016comparing,wallman2015estimating}. The situation is even worse for correlated noise models. 

Our purpose in this paper is to study further how well quantum error correction performs against coherent noise models. To make our analysis manageable, we will make some simplifying assumptions. For one, we will not actually consider quantum computation, but instead will focus on the easier task of operating a quantum memory. We envision encoding a quantum state in the memory using a quantum code; after the encoding step the memory is subjected to noise, and then the quantum state is decoded. As a further simplification, we will assume that the encoding and decoding are noiseless. Therefore, the performance of the code against the noise is captured by a \textit{logical channel}, the result of composing the encoding channel, noise channel, and decoding channel.

We will be interested in what happens to a quantum state which is stored in the memory for a long time, and undergoes many rounds of error correction --- that is, we want to characterize the effect of applying the logical channel many times in succession. For this purpose, we will need to understand the coherence properties of the logical channel. If the logical channel is incoherent, then the diamond distance of the decoded state from the ideal state grows linearly with the number of channel repetitions, while for an highly coherent logical channel, it can grow quadratically. Our main conclusion is that, even if the physical noise acting on the quantum memory is highly coherent, the coherence of the logical channel becomes strongly suppressed as the block length of the quantum error-correcting code increases, assuming that the noise is sufficiently weak and sufficiently weakly correlated. 

Although we can analyze the logical channel only in a simplified setting, and only for particular code families, we believe that the lessons learned apply more broadly. We expect, for example, that randomized benchmarking applied to \textit{logical} gates will accurately characterize logical noise even when the physical noise is highly coherent, at least for large code blocks. This also suggests that for concatenated coding schemes, in which the ``physical'' qubits of a higher-level code are themselves the logical qubits of a lower-level code, the average infidelity of the lower-level code should be a good predictor for the performance of the higher-level code. 

Our main conclusion is not unanticipated \cite{aliferisGottesmanPreskill06}, as the suppression of coherence in the logical channel has an intuitive explanation. To decode, one measures the error syndrome, and then applies a recovery operation conditioned on the syndrome. For a large code, many different syndromes are possible, and only the errors which are projected onto the same syndrome value can interfere constructively, while errors projected onto different syndrome values add stochastically. The stochastic average over many syndrome sectors suppresses coherence, leaving only small residual coherent effects arising from summing coherently over errors which are projected onto a given syndrome sector. That said, carefully analyzing the residual coherence in the logical channel involves daunting combinatorics. It turns out that further cancellations occur, resulting in even stronger suppression of logical coherence than might be naively expected.

This discussion about averaging over all syndrome sectors highlights an important issue. We will consider the logical channel obtained by averaging over error syndromes, and then study the coherence of the resulting channel. One could make a case for an alternative procedure: define a metric that characterizes coherence, evaluate that metric for the logical channel conditioned on each syndrome, and then average the value of the metric over syndromes by weighting each syndrome with its probability. To argue in favor of this alternative procedure one might note that the experimentalist who executes the error correction protocol could know the syndrome she measures in each run of the protocol, and might be interested in the properties of the logical channel conditioned on that knowledge \cite{Iyer17}. Our view is that properties of logical channels conditioned on the syndrome are potentially of interest for near-term experiments using relatively small codes, particularly because it might be feasible to postselect by retaining favorable syndromes and rejecting unfavorable ones. In future experiments using larger codes, though, syndrome histories will be quite complex, and it will be impractical to make useful inferences about the logical channel conditioned on syndrome information. For long computations using large codes, properties of the logical channel averaged over syndromes will most likely provide more usable guidance regarding the features of the protected quantum computation. 

We should also note that methods have been proposed to suppress the coherence of physical noise. One such method is randomized compiling, which, under certain assumptions, can transform any single-qubit noise channel into an incoherent depolarizing channel \cite{wallman2016noise}. The assumptions include a Markovian noise model, and gate independence of the noise for the ``easy" gates in the scheme. These assumptions may hold to a good approximation for some realistic cases, but they will not hold exactly. We may then ask how the residual coherence is affected by error correction, an issue that can be addressed using the methods in this paper. Other schemes for mitigating coherent noise have been proposed in \cite{Cai2020, chamberland2017hard, chamberland2018fault, debroy2018stabilizer}. These papers focus on the strength of the logical noise, whereas we study the character of the logical noise channel, specifically its degree of coherence.

Here we investigate the coherence of the logical channel in the case where the physical noise is fully coherent unitary noise. This problem has been previously studied \cite{greenbaum2017modeling, huang2018performance, beale2018coherence}, and we discuss this related work in Section \ref{subsec:related} below. Our work improves on these past results in that we consider a family of codes with an accuracy threshold (toric codes without boundaries) and prove bounds on the logical coherence which apply in the limit of a large code block. By specializing to a particular code family, we also find better bounds on the logical coherence for finite code length. Other authors have obtained numerical results for sufficiently small codes in the case where all physical qubits are rotated about a fixed axis \cite{Gutierrez16, bravyi2017correcting, suzuki2017efficient}, including analyses of logical channels conditioned on particular error syndromes \cite{Iyer17}. We focus instead on investigating asymptotic properties for large codes, using analytic methods. Some asymptotic statements about the performance of concatenated codes were proven in \cite{fern2006generalized}.

In our analysis we make extensive use of the chi-matrix formalism for describing quantum channels. The chi matrix arises when the action of a channel on an input density operator is expanded in terms of Pauli operators (tensor products of $2\times 2$ Pauli matrices) acting on the density operator from the left and from the right. A channel can be expressed as the sum of an ``incoherent part'' in which the Pauli operators on left and right are equal, and a ``coherent part'' in which the Pauli operators on left and right are distinct. Our main task will be to infer, in the case of stabilizer codes,  how the logical chi matrix which describes the logical channel after error correction is related to the physical chi matrix which describes the noise acting on physical qubits. 

Specifically, we study the logical channel for the toric code on an $L\times L$ lattice where $L$ is large, and where error correction is carried out using minimal-weight decoding. We estimate the coherent component of the logical chi matrix up to order $L+2\zeta$ in the rotation angle $\theta$, where $\zeta$ is any $L$-independent constant, and relate this coherent component to the incoherent component of the logical channel. Our main theorem states that the strength of the coherent part of the logical channel is bounded above by strength of the incoherent part times a factor of $1/\theta$. (Here $\theta$ is the rotation angle applied to each of the physical qubits --- our result also holds for rotation angles and axes that vary somewhat from qubit to qubit.)
From this statement, we may infer that when the logical channel is applied $m$ times in succession the average infidelity grows linearly with $m$. (There is a small contribution to the infidelity that grows quadratically with $m$, but this contribution is highly suppressed by a factor that scales as $L^{-L}$.) Stated differently, our result says that after $m$ applications of the logical channel, the accumulated distance from the identity channel, as measured by the diamond norm, grows linearly with $m$, apart from a correction which is negligible for large $L$. We emphasize that to reach this conclusion we assumed that the rotation angle $\theta$ scales with the block size as $1/L$. Therefore, unfortunately, we are not able to make a definitive statement about the coherence of the logical channel in the more physically relevant case where $L$ becomes large with $\theta$ fixed; the combinatoric task required exceeded our ability.

A related conclusion holds for a broad class of correlated noise models. We provide a detailed analysis of correlated noise for the simpler case of the quantum repetition code, under the assumption that the noise Hamiltonian commutes with the Pauli operator $X$ acting on each qubit, so that the repetition code provides effective protection against the noise model. In a model in which the rotations acting on pairs of qubits are strongly correlated, we find as expected that the correlations significantly enhance the probability of an uncorrectable logical error. However, the correlations enhance the coherent and incoherent parts of the logical chi matrix by comparable factors. Therefore, our conclusion that the coherence of the logical channel is heavily suppressed in the limit of large code length continues to apply despite the strong pairwise correlations in the noise. 

\subsection{Summary of the paper}\label{subsec:summary}
The rest of this paper is organized as follows. In Section \ref{subsec:overview} we present an overview of the proof of our main theorem, and in Section \ref{subsec:related} we compare our results to related work by previous authors. Section \ref{sec:channel-parameters} is a self-contained review of quantum channels, emphasizing metrics for characterizing coherence and relations among them. In particular, we prove a relationship between the chi matrix and the Pauli transfer matrix which had not been previously discussed to our knowledge. In Section \ref{sec:rep-code} we compute the logical channel for the repetition code assuming independent unitary noise, finding that the coherence of the logical channel becomes strongly suppressed as the code length increases.  Then in Section \ref{sec:rep-code-revisted} we analyze the repetition code again, this time using the chi-matrix formalism; we find that this analysis can be extended more easily to other stabilizer codes and other noise models. We consider the performance of the repetition code against two-body correlated noise in Section \ref{sec: Correlated unitary noise}, again concluding that the logical noise becomes incoherent in the limit of large code length. 

The heart of the paper is Section \ref{sec:toric-code}, where we build on lessons learned from the analysis of the repetition code to prove our main result, which asserts that, for an independent unitary noise model, the coherence of the logical channel is strongly suppressed by the toric code when the code block is large, assuming that the noise strength scales like $1/L$. The proof mainly consists of a combinatoric analysis which allows us to estimate the coherent and incoherent components of the logical chi matrix. We have divided the proof into a series of lemmas; figure \ref{fig: Proof guide} indicates how these lemmas fit together to build our main theorem, and Section \ref{subsec:overview} provides further guidance concerning the structure of the proof. 
Furthermore, our analysis of two-body correlated noise in the repetition code can be extended to the toric code assuming the noise is sufficiently weak for error correction to succeed with high probability; we therefore conclude that the coherence of the logical channel is highly suppressed even in the case of strongly correlated two-body noise. 

Section \ref{sec:conclusions} contains our conclusions. There we recount some of the obstacles that prevented us from extending our main theorem to the more physically relevant case where the noise strength is a constant independent of $L$. 

\begin{figure}
    \centering
    \includegraphics[height = 20cm]{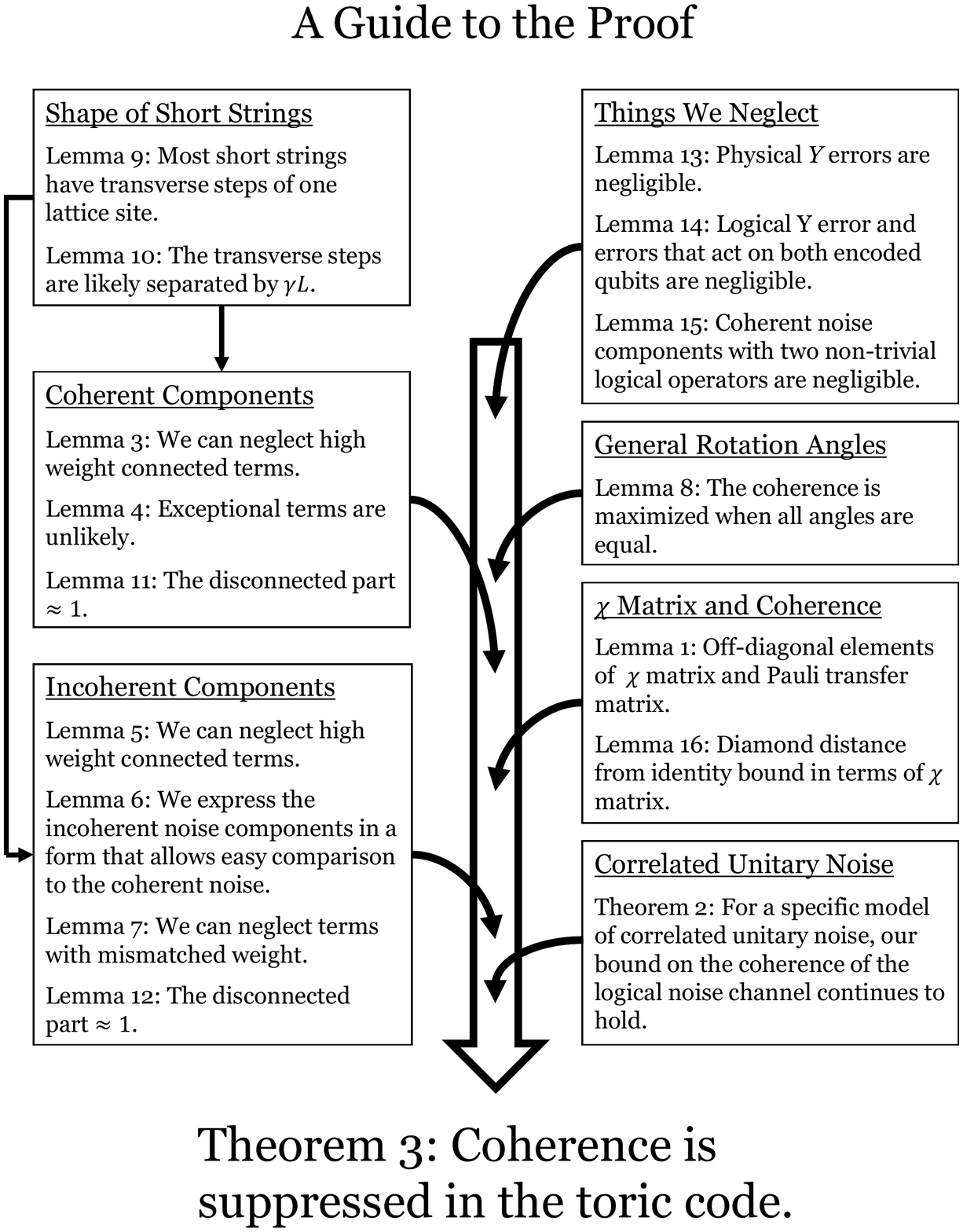}
    \caption{}
    \label{fig: Proof guide}
\end{figure}

\vskip .5cm

\subsection{Overview of the proof of Theorem \ref{theorem: Big Theorem}}\label{subsec:overview}

Here we provide some additional guidance regarding how the different parts of this paper fit together to build our main result, Theorem \ref{theorem: Big Theorem} in Section \ref{sec: Main Theorem}. The structure of our argument is also summarized in figure \ref{fig: Proof guide}.

As already noted, we study the logical channel acting on the code's protected qubits by deriving the chi matrix of this logical channel from the chi matrix of the  noise channel acting on the physical qubits. To interpret the meaning of the logical chi matrix, we find it convenient to relate the chi matrix to another formalism for describing quantum channels --- the Pauli transfer matrix. We explain some properties of the Pauli transfer matrix $N$ of a channel $\mathcal{N}$ in Section \ref{sec:channel-parameters}, relating $N$ to the diamond distance $D_\diamond(\mathcal{N})$ in equation (\ref{eq:diamond-bound-pauli-transfer}), and to the average infidelity $r_m$ of the $m$-times repeated channel $(\mathcal{N})^m$ in equations (\ref{eq:coherence-off-diagonal}) and (\ref{eq:coherence-angle-bound}). Using Lemma \ref{lemma:off-diagonal}, these expressions for the  diamond distance and the average infidelity in terms of the Pauli transfer matrix can be restated in terms of the chi matrix.

In Section \ref{sec:rep-code} we study the performance of the quantum repetition code against coherent noise, and prove Theorem \ref{Theorem: rep-code-epsilon-delta} using explicit computation of the logical channel combined with results derived in Section \ref{sec:channel-parameters}. This result shows that the logical channel is highly incoherent when the code block is large. An alternative proof of Theorem \ref{Theorem: rep-code-epsilon-delta}, making essential use of the chi matrix, is presented in Section \ref{sec:rep-code-revisted}, where we develop the key tools needed for the proof of Theorem \ref{theorem: Big Theorem}. We also prove Lemma \ref{lemma:minimization}, which is used to show that, for independent unitary noise acting on the physical qubits, the coherence of the logical channel for the repetition code is maximized when all qubits are rotated by the same angle. A similar idea can be adapted for analyzing the coherence of the logical channel for the toric code.

In Section \ref{sec: Correlated unitary noise}, we extend the analysis of the repetition code to the case of two-body correlated coherent noise, culminating in the proof of Theorem \ref{Theorem: Correlated noise}, showing that the coherence of the logical channel is heavily suppressed in this case as well. The proof is a computation of the logical channel for this case, achieved by a detailed combinatoric analysis. As expected, the noise correlations enhance the probability of a decoding error, but it turns out that both the coherent and incoherent parts of the logical channel are enhanced, so that the relationship between the two is not changed much compared to the case of uncorrelated coherent noise. The same reasoning used to prove Theorem \ref{Theorem: Correlated noise} can also be applied to the toric code to show that, in that case as well, two-body correlations in the noise do not enhance the coherence of the logical channel. 

Our analysis of the performance of the toric code against coherent noise, culminating in the proof of Theorem \ref{theorem: Big Theorem}, is in Section \ref{sec:toric-code}. To prove the theorem we compute first the coherent part of the logical channel, and then the incoherent part, after which we can make an inference about how the two are related. For this purpose, upper bounds on the logical noise strength would not suffice. Instead, we compute both the coherent and incoherent part of the logical channel up to an error which we show is small if the physical noise is sufficiently weak.

Our arguments in Section \ref{sec:toric-code} make use of observations, discussed in Section \ref{sec:rep-code-revisted}, which apply to any stabilizer code. We may assign a ``standard error'' $E_s$ to each error syndrome $s$, and define a decoder which returns the damaged state to the code space by applying $E_s^\dagger$ when the syndrome is measured to be $s$. This $E_s$ is a Pauli operator acting on the code block. Furthermore, each logical Pauli operator $\tilde{L}_a$ acting on the code may by convention be associated with a particular standard physical Pauli operator $L_a$ --- the choice of $L_a$ is not unique, and therefore must be fixed by convention, because we have the freedom to multiply $L_a$ by an element of the code's stabilizer group without changing its logical action. Once the standard error for each syndrome, and the physical Pauli operator corresponding to each logical Pauli operator, are determined, any physical Pauli operator acting on the code block has a unique decomposition of the form (up to a phase factor) $\sigma(s,a,x)= E_s L_a G_x$, where $E_s$ is a standard error, $L_a$ is a standard logical Pauli operator, and $G_x$ is an element of the code stabilizer.  

In the chi matrix formalism, the result $\mathcal{N}(\rho)$ of applying noisy channel $\mathcal{N}$ to density operator $\rho$ is expanded as a sum of terms of the form $\sigma(s,a,x)\, \rho \, \sigma(s',a',x')^\dagger$. As explained in Section \ref{sec:recovery-chi}, if $\rho$ is a logical density operator, then a term of this form is annihilated by the error recovery operation for $s\ne s'$, and for $s=s'$ is mapped to $L_a\rho L_{a'}^\dagger$, up to a phase. (That phase is important, and we will need to keep track of it carefully.) Recovery is successful if $L_a$ and $L_{a'}$ are both logical identity operators. The terms in the logical channel with $L_a = L_{a'}$ are said to be incoherent, and the terms with $L_a \ne L_{a'}$ are said to be coherent. 

The key point is that we have a conceptually simple algorithm for computing the chi matrix for the logical channel, and for identifying its coherent and incoherent parts. To find the coefficient of $L_a\rho L_{a'}^\dagger$ in the logical channel, we just need to sum up the coefficients of all terms in the physical chi matrix of the form $\sigma(s,a,x) \, \rho \, \sigma(s,a',x')^\dagger$, being mindful of phase factors, for all possible values of $s, x, x'$. Unfortunately, in general this algorithm is too complex to carry out in practice, but under suitable conditions we can estimate logical chi matrix with sufficient accuracy for our purposes.

For the case of the toric code, we can begin by noting some helpful simplifications. We choose standard errors defined by minimal-weight decoding. Because of the code's CSS structure, we can analyze the logical $X$ and logical $Z$ errors separately, and in fact a single analysis applies to errors of both types.  We don't need to worry about logical $Y$ errors or about logical errors acting nontrivially on more than one of the code's logical qubits (Lemma \ref{Lemma: Other logical maps} in Appendix \ref{app: X1X2 or Y1 terms}) because these are so highly suppressed; the same goes for coherent errors in which both $L_a$ and $L_{a'}$ are nontrivial (Lemma \ref{Lemma: More general coherent terms} in Appendix \ref{app: a and b nontrivial}).  We can assume that the coherent noise rotates physical qubits about an axis in the $X{-}Z$ plane (Lemma \ref{lemma: Y rotations} in Appendix \ref{app: Physical Y Errors}); otherwise the logical noise would be even less coherent. We are left with the task of estimating two nontrivial elements of the logical chi matrix --- the coherent term $\tilde Z_1\rho \tilde I$ term and the incoherent term $\tilde Z_1\rho \tilde Z_1$, where $\tilde Z_1$ denotes the logical $Z$ operator acting on one of the code's two encoded qubits. In the proof of Theorem \ref{theorem: Big Theorem}, we estimate both quantities using a series of approximations, and verify that these approximations are trustworthy when the physical noise is sufficiently weak. 

First consider the coherent part of the logical chi matrix. We need to sum up all the terms in the physical chi matrix which contribute to $\tilde Z_1\rho \tilde I$ after the action of the decoding map. Each such term has the form $E_s Z_1 G_x \rho \, G_y^\dagger E_s^\dagger$, where $E_s$ denotes a standard correctable Pauli error, $G_x$, $G_y$ are Pauli operators in the code stabilizer, and $Z_1$ is the standard physical Pauli operator whose logical action matches $\tilde Z_1$. For the purpose of our computation, we may assume that all the Pauli operators are of the $Z$ type --- that is, each applies $Z$ to a subset of the qubits and applies $I$ to the complementary set. For the purpose of enumerating all such contributions, it is convenient to note that the product $G_y^\dagger Z_1 G_x$ of the Pauli operators acting on the density operator from the right and from the left is a logical operator, one commuting with the code stabilizer. This logical operator can be decomposed into a connected path that winds once around on the periodically identified square lattice --- what we call a ``logical string'' --- and a collection of homologically trivial closed loops on the lattice --- what we call the ``disconnected part'' of the logical Pauli operator. 

We can therefore enumerate all the contributions to $\tilde Z_1\rho \tilde I$ by this procedure: (1) Consider all possible logical strings. (2) For each logical string, consider all possible ``partitions'' of that string into an uncorrectable error acting from the left and a correctable error acting from the right. (3) For each logical string and partition, consider all possible choices for the disconnected part. We compute $\tilde Z_1\rho \tilde I$ by summing all these contributions. Though we can't perform this sum exactly, we can approximate the sum and estimate the resulting errors. 

It is for the purpose of approximating this sum that we need the assumption that the rotation angle $\theta$ scales like $1/L$, where $L$ is the linear system size. Under this assumption, we show that we make a small error by truncating the sum to include only relatively short logical strings (Lemma \ref{Lemma: Coherent path counting} in Section \ref{sec: Path Counting}) which have a typical shape (Lemmas \ref{Lemma: shape of strings A} and \ref{Lemma: shape of strings B} in Appendix \ref{app: Shape of string}). Summing over all partitions of a fixed logical string is similar to the computation we performed for the repetition code, but with a few new subtleties. Specifically, there are some ``exceptional'' partitions such that the uncorrectable error acting from the left actually has lower weight than the correctable error acting from the right. Fortunately, we can show that we make a small error by ignoring this effect (Lemma \ref{Lemma: exceptional terms} in Section \ref{sec: sum over partitions}), simplifying the sum over partitions. 

For a fixed connected logical string and partition of that string, we need to sum over disconnected closed loops and partitions of those loops. Performing this sum is almost equivalent to adding up all possible error patterns weighted by their probabilities, which trivially sums to unity. The only complication is that, for some closed loops that closely approach the logical string, and for some special partitions, the additional loop can flip how the error is decoded. It turns out, though, that we make only a small error by ignoring this effect (Lemma \ref{lemma: added error exceptional terms} in Appendix \ref{app: Disconnected errors}). 

With all the above simplifications in hand, we can estimate the coherent part of the logical chi matrix. In particular, the sum over partitions for a fixed logical string can be evaluated much as in the proof of Theorem \ref{Theorem: rep-code-epsilon-delta} for the repetition code. It then remains to estimate the incoherent part and compare the two. 

In the incoherent part, $\tilde Z_1$ acts from both the left and the right; therefore, there are two logical strings to keep track of, one on each side. These two logical strings have segments in common, determined by the intersection of the string with the standard error, but are free to fluctuate independently away from those segments (figure \ref{fig: connected incoherent example} of Section \ref{section: Incoherent Sum}). To approximate the sum over contributions from these logical strings to the incoherent part of the logical chi matrix, we may truncate the sum as in the computation of the coherent part, limiting our attention to relatively short strings with a typical shape (Lemma \ref{Lemma: High Weigh Incoherent Terms} in Section \ref{section: Incoherent Sum} and Lemma \ref{Lemma: Incoherent string sum bound} in Section \ref{sec: Incoherent sum over strings}), and ignoring complications arising from the disconnected part of the error (Lemma \ref{lemma: incoherent disconnected part} in Appendix \ref{app: incoherent disconnected part}). Furthermore, we may also ignore contributions with ``mismatched weight,'' confining our attention to minimal-weight uncorrectable errors on the logical string acting from both the left and the right (Lemma \ref{Lemma: Mismatched weights} in Section \ref{sec: Mismatched weights}). With these approximations, the incoherent part of the logical chi matrix may be expressed in a form which can be conveniently compared with the coherent part. 

As for the repetition code, we can justify considering unitary noise such that all physical qubits are rotated by the same angle --- rotating different qubits by different angles only makes the logical channel less coherent (Lemma \ref{Lemma: general rotation angles} in Section \ref{sec: more general rotation angles}). For such a coherent noise model with uniform rotation angles, we compare the coherent and incoherent parts of the logical chi matrix, proving Theorem \ref{theorem: Big Theorem} (Section  \ref{sec: Main Theorem}). Using the findings from Section \ref{sec:channel-parameters}, these results can be translated into statements about the diamond distance of the logical channel and about the average infidelity of the $m$-times repeated logical channel. We also observe (Section \ref{sec: toric code with correlations}), that our analysis of the performance of the repetition code against two-body correlated coherent noise (Theorem \ref{Theorem: Correlated noise}) is applicable with few modifications to the toric code as well.

Our conclusion that the coherence of the logical channel is heavily suppressed applies in the limit of large code size $L$, and under the assumption that the physical qubits are rotated by an angle $\theta$ scaling like $1/L$. In Section \ref{sec:conclusions} we discuss the difficulties that have prevented us from extending the result to larger values of $\theta$.

\subsection{Related work} \label{subsec:related}

The performance of stabilizer codes against fully coherent unitary noise has been previously studied in  \cite{greenbaum2017modeling, huang2018performance, beale2018coherence}. 
Huang, Doherty, and Flammia \cite{huang2018performance} derived an inequality which relates the diamond distance $D_\diamond$ of the logical channel from the identity to the rotation angle $\theta$ for independent unitary noise, finding
\begin{equation}\label{eq:flammia}
    D_\diamond \le c_{n,k} |\sin\theta|^d;
\end{equation}
here $d$ is the code distance, $n$ is the code length, and $k$ is the number of encoded qubits. Their result applies to any stabilizer code, but $c_{n,k}$ grows exponentially with $n$ (it is bounded above by $2^{3n+k+1}$), so their result is not very informative for large codes. In contrast, we derive a bound relating the coherent and incoherent components of the logical channel which does not involve any exponentially large factors. We achieve this improved result by specializing to the toric code, and by assuming $\sin\theta < 1/L$. Furthermore, to obtain equation (\ref{eq:flammia}) the authors of \cite{huang2018performance} bounded a sum of contributions to the logical channel using the triangle inequality, hence obtaining a bound that would apply even if all the terms in the sum had a common phase. Instead, we sum the contributions with the appropriate phases; the resulting cancellations among terms yield a much smaller result than we would have obtained by merely invoking the triangle inequality. We are able to carry out this more detailed analysis because our assumption $\sin\theta < 1/L$ allows us to restrict our attention to short logical strings, for which approximating the sum becomes a manageable task. 

Beale, Wallman, Gutti\'errez, Brown, and Laflamme \cite{beale2018coherence} also studied the performance of stabilizer codes against independent unitary noise, and they concluded that the coherence of the logical channel is suppressed. For a fixed code length, they study the limit of small rotation angle $\theta$. If the logical channel is expanded in powers of $\theta$, then for sufficiently small $\theta$ the leading term in this expansion dominates, and they draw their conclusions by analyzing this leading term. In effect, they (like us) investigate the case in which the noise strength deceases as the code length increases, but their assumption about the noise strength is much stronger than ours. We (unlike them) include all corrections to the logical channel higher order in $\theta$ that are needed to accurately approximate the logical channels for $\sin\theta < 1/L$, albeit only for the special case of the toric code. 


Bravyi, Engelbrecht, K\"onig, and Peard \cite{bravyi2017correcting} have studied the performance of the toric code against independent unitary noise numerically, using a clever mapping from qubits to Majorana fermions, for code distance up to $d=37$, and they found that the coherence of the logical channel becomes negligible as the code length increases, provided that the rotation angle $\theta$ is smaller than a nonzero constant threshold value $\theta_0$. Their numerical method applies to a noise model in which all qubits are rotated about the $Z$ axis, which according to our analysis is the worst case that maximizes the coherence of the logical channel. The numerical results support a value of $\theta_0$ greater than $0.25$ and less than $0.32$, while for the largest code sizes they consider our analytic results apply only for $\theta$ less than about $0.027$. They characterize the coherence of the logical channel by sampling from the distribution governing the logical rotation angle $\theta_{\textrm{logical}}$  conditioned on the measured error syndrome, finding that this distribution becomes strongly peaked around $\theta_{\textrm{logical}}=0$ for large code length when $\theta_{\textrm{physical}}$ is smaller than $\theta_0$. They also consider, as we do, the logical channel averaged over syndromes, and show that the ``twirled'' logical channel has an error probability close to the error probability of the untwirled logical channel for large code length, a further indication of suppressed logical coherence. Their numerical findings appear to be at least notionally consistent with our analytic results, though it is difficult to make a quantitative comparison because our formulas are accurate only for asymptotically large $L$ and for $L\sin\theta$ sufficiently small compared to 1. 

\section{Channel parameters}\label{sec:channel-parameters}
\subsection{Pauli transfer matrix}
We will use the \textit{Pauli transfer matrix} representation to describe channels acting on $n$ qubits. For this purpose we expand the density operator $\rho$ in the Pauli operator basis $\{\sigma^i\}$:
\begin{equation}
\rho = \sum_{i = 0}^{d^2-1} \rho_i \sigma^i,
\end{equation}
where
\begin{equation}
\mathrm{Tr} \left( \sigma^i \sigma^j\right) = \frac{1}{d} \delta^{ij},
\end{equation}
and $\sigma^0 = (id)/d$. Here $d=2^n$ is the Hilbert-space dimension, and $id$ denotes the $d\times d$ identity matrix. Note that $\mathrm{Tr}(\rho) = \rho_0$. A linear map $\mathcal{N}$ acting on density operators defines a $d^2\times d^2$ matrix (the Pauli transfer matrix associated with $\mathcal{N}$) according to
\begin{equation}
\mathcal{N}(\rho) = \sum_{i,j} \left(N_{ij} \rho_j\right) \sigma^i.
\end{equation}
This matrix is real if $\mathcal{N}$ maps Hermitian operators to Hermitian operators. If the map $\mathcal{N}$ is trace preserving, then $\sum_i N_{0 i} \rho_i=\rho_0$; hence $N_{0i} = \delta_{0i}$. If the map $\mathcal{N}$ is unital (that is, $\mathcal{N}(id) = id$), then $\sum_i N_{ij} \delta_{j0} = \delta_{i0}$; hence $N_{i0} = \delta_{i0}$. Thus the matrix representing the map $\mathcal{N}$ may be expressed as 
\begin{equation}
N = \left(
\begin{array}{cccc}
1 & 0 & 0 & \cdots \\
\\
N_n & & N_u & \\
\\
\end{array} \right).
\end{equation}
We say that the $(d^2-1) \times (d^2-1)$ matrix $N_u$ is the unital part of $\mathcal{N}$ and that the length-$(d^2-1)$ vector $N_n$ is its nonunital part. Altogether the trace-preserving map $\mathcal{N}$ is specified by $d^2(d^2-1)$ parameters.

For a unitary map $\mathcal{N}(\rho) = U\rho U^\dagger$, we have $N_n=0$ and (for $i\ne 0$)
\begin{equation}
U\sigma^i U^\dagger = \sum_{j=1}^{d^2-1} (N_u)_{ji} \sigma^j,
\end{equation}
where
\begin{equation}
d \,\mathrm{Tr}\left(U\sigma^iU^\dagger U \sigma^kU^\dagger\right)=\delta^{ik} = d \sum_{j,l\ne 0}(N_u)_{ji}(N_u)_{lk} \mathrm{Tr}\left(\sigma^j\sigma^l\right) =\sum_{j\ne 0} (N_u)_{ji}(N_u)_{jk};
\end{equation}
hence $N_u$ is an orthogonal matrix. 

The matrix representing $\mathcal{N}$ is diagonal if and only if the map is a convex sum of Pauli operators
\begin{equation}
\mathcal{N}(\rho) = \sum_i p_i \sigma^i \rho \sigma^i,
\end{equation}
in which case the diagonal entries are
\begin{equation}
N_{jj} = \sum_i p_i \xi_{ij},
\end{equation}
where $\sigma^i\sigma^j = \xi_{ij}\sigma^j\sigma^i$; that is, $\xi_{ij}$ is the sign $\pm 1$ determined by whether the Pauli operators $\sigma^i$ and $\sigma^j$ commute or anticommute. 

\subsection{Average infidelity}
The \textit{fidelity} $F$ of a channel $\mathcal{N}$ acting on a pure state $|\psi\rangle$ is defined by
\begin{equation}
F = \langle \psi| \mathcal{N}(\rho)|\psi\rangle,
\end{equation}
and $1-F$ is called the \textit{infidelity}. The \textit{average infidelity} $r$ of $\mathcal{N}$ is 
\begin{equation}
r =1-\int dU \,\mathrm{Tr}\left[ U\rho U^\dagger \mathcal{N}(U\rho U^\dagger)\right],
\end{equation}
where the integral is with respect to the normalized invariant Haar measure on the unitary group, and $\rho$ is any pure state. Equivalently, $r$ is the infidelity of the averaged channel
\begin{equation}
\bar {\mathcal{N}}(\rho) = \int dU\, U^\dagger \mathcal{N}(U\rho U^\dagger) U.
\end{equation}
We may just as well define $r$ as the infidelity of $\mathcal{N}$ averaged over a unitary 2-design. Hence $r$ can be measured in randomized benchmarking experiments, in which $U$ is chosen by sampling uniformly from the Clifford group, which is a unitary 2-design. 

The $d\times d$ unitary matrix $U$ defines an orthogonal $(d^2-1)\times (d^2 -1)$ matrix $N_u=O$ according to
\begin{equation}
U\sigma^i U^\dagger = \sum_{j=1}^{d^2-1} O_{ji} \sigma^j, \quad U^\dagger\sigma^i U = \sum_{j=1}^{d^2-1} O^T_{ji} \sigma^j,
\end{equation}
where $O^T$ denotes the transpose of $O$;
therefore
\begin{equation}
U^\dagger \mathcal{N}(U \sigma^i U^\dagger) U = \sum_{j=1}^{d^2-1} \left(O^T N_u O\right)_{ji}\sigma^j; \quad
U^\dagger \mathcal{N}(U \sigma^0 U^\dagger) U = \sigma^0 + \sum_{i=1}^{d^2 -1}(O^T N_n)_i \sigma^i .
\end{equation}
The uniform average of $U$ over the unitary group becomes a uniform average of $O$ over the orthogonal group. The nonunital part of $\mathcal{N}$ averages to zero, and the average of the unital part can be evaluated using
\begin{equation}
\int dO\, O^T_{ij} O_{kl} = \frac{1}{d^2-1}\delta_{jk}\delta_{il},
\end{equation}
which yields
\begin{equation}
\left (\bar N_u\right)_{ij} = \frac{\mathrm{Tr} (N_u) }{d^2-1} \delta_{ij}.
\end{equation}
Hence, the averaged channel is a completely depolarizing Pauli channel of the form
\begin{equation}
\bar {\mathcal{N}}(\rho) = p \rho +(1 - p)\left( \frac{id}{d}\right),
\end{equation}
where 
\begin{equation}
p = \frac{1}{d^2-1} \mathrm{Tr}(N_u).
\end{equation}
Note that if this averaged channel is applied $m$ times in succession, we obtain
\begin{equation}
\bar {\mathcal{N}}^m(\rho) = p^m \rho + (1-p^m) \left( \frac{id}{d}\right);
\end{equation} 
thus $p$ is called the \textit{benchmarking parameter} because it determines the rate of exponential decay of fidelity in benchmarking experiments. The average infidelity $r$ is given by
\begin{equation}
r = 1 - \langle \psi | \bar {\mathcal{N}}(|\psi\rangle\langle \psi |) |\psi\rangle = 1- \left(p + \frac{1-p}{d} \right)= \frac{d-1}{d} (1-p)
= \frac{1}{d(d+1)} \mathrm{Tr}( I_{d^2-1} - N_u)
\end{equation}
for any pure state $|\psi\rangle$. Here $I_{d^2-1}$ denotes the $(d^2-1) \times (d^2 -1)$ identity matrix. Because $N_{00} =1$, we may also express the infidelity as
\begin{equation}
r = \frac{1}{d(d+1)}\mathrm{Tr} \left( I_{d^2} - N \right),
\end{equation}
where $I_{d^2}$ denotes the $d^2 \times d^2$ identity. 

\subsection{Examples}

\subsubsection{Depolarizing channel}
We have seen that if $\mathcal{N}_p$ is the depolarizing channel with benchmarking parameter $p$, then $(\mathcal{N}_p)^m = \mathcal{N}_{p^m}$. Using the relation $r = \frac{d-1}{d}(1- p)$, we can express the infidelity $r_m$ of $(\mathcal{N}_p)^m$ in terms of the infidelity $r$ of $\mathcal{N}_p$, finding
\begin{equation}
r_m = \frac{d-1}{d} (1 - p^m) = mr - \frac{d}{2(d-1)} m(m-1) r^2 + \mathcal{O}(m^3r^3).
\end{equation}
If $mr$ is small, the infidelity accumulates linearly with $m$, the number of times the channel is applied. A similar remark applies to more general Pauli channels. 

We say that a channel with this property is \textit{incoherent}. The interpretation is that (up to a constant factor), the infidelity $r$ may be regarded as a \textit{probability} of error. If the channel is applied $m$ times, where $mr$ is small, any one of the $m$ instances of the channel could be faulty, so that the total probability of error is $mr$ + higher-order terms. 

\subsubsection{Qubit rotation}

In contrast, consider the case of a unitary rotation of a single qubit about the $X$-axis 
\begin{equation}
U^X(\theta) = \exp\left( - i \frac{\theta}{2} \sigma^X\right)
\end{equation}
which rotates the Bloch sphere by $\theta$. For this channel the Pauli transfer matrix is
\begin{equation}\label{eq:qubit-rotation}
N(\theta) = \left(
\begin{array}{cccc}
1 & 0 & 0 & 0 \\
0 & 1 & 0 & 0 \\
0 & 0 & \cos\theta & \sin\theta \\
0 & 0 & -\sin\theta & \cos\theta \\
\end{array} \right);
\end{equation}
therefore, the infidelity is 
\begin{equation}
r = \frac{1}{6} \,\mathrm{Tr}\left(I - N_u\right) = \frac{1}{3} \left(1 - \cos\theta\right) = \frac{1}{6}\, \theta^2 - \frac{1}{72} \theta^4 + \mathcal{O}(\theta^6).
\end{equation}
Applying this channel $m$ times, we obtain $N(\theta)^m= N(m \theta)$, a rotation by an angle $m$ times larger. Therefore, 
\begin{equation} 
r_m = \frac{1}{3} ( 1 - \cos m\theta ) = m^2 r - \frac{1}{2} m^2 (m^2-1)r^2 +\mathcal{O}(m^6 r^3).
\end{equation}
Here, for $m^2 r$ small, the infidelity accumulates quadratically with $m$; it is the rotation angle, rather than the error probability, that increases linearly. We say that a channel like this one, for which the infidelity increases faster than linearly with $m$, is \textit{coherent}.

\subsubsection{Rotation/Dephasing channels}
The distinction between a coherent and incoherent channel is not always clearcut, and we will need measures that quantify the degree of coherence. As an example, consider the case where a qubit either dephases in the $X$-basis (with probability $q_D$) or is rotated by angle $\theta$ about the $X$-axis (with probability $q_R$):
\begin{equation}
\mathcal{N}(\rho) = (1- q_D - q_R) \rho + q_D \sigma^X \rho \sigma^X + q_R U^X(\theta) \rho U^X(\theta)^\dagger.
\end{equation}
The Pauli transfer matrix is 
\begin{equation}\label{eq:Pauli-transfer-N}
N= \left(
\begin{array}{cccc}
I & 0 \\
0 & M \\
\end{array} \right),
\end{equation}
where $I$ is the $2\times 2$ identity, and $M$ is the $2\times 2$ matrix
\begin{equation}\label{eq:M-dephase-rotate}
M= \left(
\begin{array}{cccc}
1- \epsilon & \delta \\
-\delta & 1 -\epsilon\\
\end{array} \right),
\end{equation}
with
\begin{equation}
\epsilon = 2 q_D +q_R(1-\cos\theta), \quad \delta = q_R \sin \theta.
\end{equation}
The infidelity is 
\begin{equation}\label{eq:rm-dephase-rotate}
r = \frac{1}{6} \, \mathrm{Tr}\left( I - M\right) = \frac{1}{3}\, \epsilon = \frac{2}{3}\, q_D + \frac{1}{3}\, q_R(1 - \cos\theta).
\end{equation}

The eigenvalues of $M$ are
\begin{equation}
\lambda_{\pm} = 1-\epsilon \pm i \delta, 
\end{equation}
and therefore the infidelity of $\mathcal{N}^m$ is
\begin{align}\label{eq:Nm-epsilon-delta-bound}
r_m = \frac{1}{6} \,\mathrm{Tr}\left( I - M^m\right) &= \frac{1}{6}\left[ 2 - (1 - \epsilon + i\delta)^m - (1 - \epsilon - i \delta)^m\right]\notag\\
&= \frac{1}{3}\, m\epsilon - \frac{1}{6} m (m-1) \left( \epsilon^2 - \delta^2 \right) + \mathcal{O}(m^3\epsilon^3,m^3 \epsilon\delta^2, m^4\delta^4).
\end{align}
Here the degree of coherence depends on the relative value of $\epsilon$ and $\delta$. In the case of a unitary rotation, we have $\epsilon = \mathcal{O}(\delta^2)$, which means that the term growing quadratically with $m$ can dominate. On the other hand, for $\epsilon \ge \delta$, there is no quadratically growing term at all.  

A generalization of this channel will be useful in Section \ref{sec:rep-code}. Instead of a single rotation by $\theta$ occurring with probability $q_R$, we may consider an ensemble of possible rotations, where a rotation by $\theta_a$ occurs with probability $q_a$. In that case $r_m$ is still given by equation (\ref{eq:Nm-epsilon-delta-bound}), but now
\begin{equation}
\epsilon = 2 q_D +\sum_a q_a(1-\cos\theta_a), \quad \delta = \sum_a q_a \sin \theta_a.
\end{equation}

\subsection{Unitarity and the coherence angle}
We have seen that $N_u$ is an orthogonal matrix if (and only if) the channel $\mathcal{N}$ is unitary. Hence a deviation from orthogonality of $N_u$ indicates a deviation from unitarity of $\mathcal{N}$. With that in mind, following \cite{wallman2015estimating} we define the \textit{unitarity}  $u(\mathcal{N})$ of the channel $\mathcal{N}$ as
\begin{equation}
u(\mathcal{N}) = \frac{1}{d^2-1} \, \mathrm{Tr}\left( N_u^T N_u\right),
\end{equation}
which is 1 for unitary channels and strictly less than 1 for nonunitary channels. For a fixed value of the infidelity $r$, the unitarity achieves its minimum for the depolarizing channel \cite{wallman2014randomized}, where 
\begin{equation}
u(\mathcal{N}) = \frac{1}{d^2-1} \left(\mathrm{Tr}\, N_u^2 \right)= p^2 = \left( 1- \frac{dr}{d-1} \right)^2.
\end{equation}

The unitarity $u$ and the benchmarking parameter $p$ together provide a useful characterization of the coherence of a channel. We will be primarily interested in the case where the infidelity $r$ is small, so that the diagonal elements $\{N_{ii}\}$ of the Pauli transfer matrix are close to one, and it makes sense to expand in the small quantity $1 - N_{ii}$. Writing 
\begin{equation}
\left(N_u\right)_{ii}^2 = \left(1 - \left(1-\left(N_u\right)_{ii}\right)\right)^2 = 1 - 2\left(1-\left(N_u\right)_{ii}\right) + \left(1-\left(N_u\right)_{ii}\right)^2,
\end{equation}
we see that
\begin{equation}\label{eq:u-ready-to-expand}
u(\mathcal{N}) =  \frac{1}{d^2-1} \sum_{i,j} \left(N_u\right)_{ij}^2 = 1- 2(1-p) +\frac{1}{d^2-1}\sum_{i,j | i \neq j} \left(N_u\right)_{ij}^2 + \frac{1}{d^2-1} \sum_i\left(1-\left(N_u\right)_{ii}\right)^2.
\end{equation}
Expanding the square root of $u$ we find
\begin{equation}
\sqrt{u(\mathcal{N})}= p +\frac{1}{2(d^2-1)} \sum_{i,j | i \neq j} \left(N_u\right)_{ij}^2 + \cdots ,
\end{equation}
where the ellipsis indicates terms that are fourth order in the off-diagonal entries $\left(N_u\right)_{ij}$ and terms that are quadratic order in $\left(1- \left(N_u\right)_{ii}\right)$.

The \textit{coherence angle} $\Theta$ is defined as 
\begin{equation}
\Theta = \arccos\left( p / \sqrt{u}\right),
\end{equation}
which for $p$ and $u$ close to one, can be expressed as 
\begin{equation}\label{eq:coherence-off-diagonal}
\Theta^2 = 2\left(1 - \frac{p}{\sqrt{u}}\right) + \cdots = \frac{1}{d^2-1}\sum_{i,j | i \neq j} \left(N_u\right)_{ij}^2 + \cdots .
\end{equation}
Apart from a normalization factor, and neglecting the higher-order terms, $\Theta^2$ is the sum of squares of all off-diagonal terms in $N_u$. It quantifies the coherence in the channel. 

For the qubit rotation channel in equation (\ref{eq:qubit-rotation}), the coherence angle is related to the rotation angle $\theta$ by
\begin{equation}
\Theta^2 \approx \frac{2} {3}\sin^2 \theta \approx \frac{2}{3}\theta^2.
\end{equation}
For the dephasing/rotation qubit channel in equation (\ref{eq:M-dephase-rotate}), our truncated power series expansion used to derive equation (\ref{eq:coherence-off-diagonal}) is justified if $\epsilon$ is negligible compared to $\delta$, in which case we find
\begin{equation}
\Theta^2   \approx \frac{2}{3} q_R^2 \theta^2.
\end{equation}
For the depolarizing channel, $u=p^2$ and hence $\Theta = 0$.

In \cite{dugas2016efficiently}, Carignan-Dugas \textit{et al.} derived a bound on $r_m$, the infidelity when a \textit{unital} channel $\mathcal{N}$ is applied $m$ times in succession, in terms of the infidelity $r$ and coherence angle $\Theta$ of $\mathcal{N}$:
\begin{equation}\label{eq:coherence-angle-bound}
r_m \le mr + \frac{d-1}{2d}m(m-1)\Theta^2 + \cdots,
\end{equation}
where the ellipsis indicates terms higher order in $r$ and $\Theta^2$. 
In this sense (for unital channels), the coherence angle controls the quadratic growth of $r_m$ as a function of $m$, when $r$ and $\Theta^2$ are small.

\subsection{Diamond distance}    

In some versions of the quantum accuracy threshold theorem, the strength of Markovian noise is characterized by the deviation of a noisy gate from the corresponding ideal gate in the \textit{diamond norm} \cite{kitaev2002classical}. This diamond norm deviation is useful for quantifying the damage inflicted when the noisy gate acts on qubits which are entangled with other qubits in a quantum computer. The diamond norm $|\mathcal{E}\|_\diamond$ of a linear map $\mathcal{E}$ is defined as the $L^1$ norm of the extended map $\mathcal{E}\otimes \mathcal{I}$:
\begin{equation}
    \|\mathcal{E}\|_\diamond=\max_\rho \| \mathcal{E}\otimes \mathcal{I}(\rho)\|_1.
\end{equation}
If $\mathcal{E}$ acts on Hilbert space $\mathcal{H}$ with dimension $d$, then $\mathcal{I}$ denotes the identity acting on another Hilbert space $\mathcal{H}'$ with dimension $d$; the maximum is over all density operators on $\mathcal{H}\otimes \mathcal{H}'$. A measure of noise strength for a noisy channel $\mathcal{N}$ is the diamond distance of $\mathcal{N}$ from the identity channel,
\begin{equation}
    D_\diamond(\mathcal{N}) := \| \mathcal{N} - \mathcal{I}\|_\diamond.
\end{equation}
If $\mathcal{N}$ is applied $m$ times in succession, we have
\begin{equation}
    D_\diamond(\mathcal{N}^m) \le m D_\diamond(\mathcal{N}).
\end{equation}
    
Upper and lower bounds on the diamond distance can be expressed in terms of the benchmarking parameter $p(\mathcal{N}) = 1-r(\mathcal{N})d/(d-1)$ and the unitarity $u(\mathcal{N})$ \cite{kueng2016comparing}:
\begin{equation}\label{eq:kueng-bound}
\frac{\sqrt{d^2-1}}{2d}f(p,u) \le D_\diamond \le \frac{d \sqrt{d^2-1}}{2} f(p,u),
\end{equation}
where 
\begin{equation}
f(p,u) = \sqrt{(1-2p +u)}.
\end{equation}
For the depolarizing channel, we have $u=p^2$ and $f = 1-p = rd/(d-1)$; the diamond distance scales linearly with the infidelity $r$. But for a unitary channel, we have $u = 1$ and $f= \sqrt{2(1-p)}$; then the diamond distance scales like $\sqrt{r}$.

From equation (\ref{eq:u-ready-to-expand}), we see that
\begin{equation}
\label{eq: f^2 bound diamond distance}
f(p,u)^2 = 1-2p+u =\frac{1}{d^2-1}\left(\sum_{i,j | i \neq j} \left(N_u\right)_{ij}^2 +  \sum_i\left(1-\left(N_u\right)_{ii}\right)^2\right),
\end{equation}
which together with equation (\ref{eq:kueng-bound}) provides upper and lower bounds on the diamond distance written in terms of of Pauli transfer matrix elements:
\begin{align}\label{eq:diamond-bound-pauli-transfer}
D_\diamond(\mathcal{N}) & \geq \frac{1}{d} \left(\sum_{i,j | i \neq j} \left(N_u\right)_{ij}^2 +  \sum_i\left(1-\left(N_u\right)_{ii}\right)^2\right)^{1/2} \nonumber
\\
D_\diamond(\mathcal{N}) & \le d \left(\sum_{i,j | i \neq j} \left(N_u\right)_{ij}^2 +  \sum_i\left(1-\left(N_u\right)_{ii}\right)^2\right)^{1/2}.
\end{align}
We will be mostly interested in the upper bound on the diamond distance for a logical channel with a fixed number of encoded qubits; therefore the unfavorable scaling of the upper bound with the dimension $d$ need not cause us great concern.

\subsection{Coherence in the chi-matrix representation}

The Pauli transfer matrix representation is useful for proving the preceding relationships between channel components, the growth of average infidelity, and the dependence of the diamond distance from identity on the average infidelity. When we analyze error correction, we will make use of a different representation of the noise channel. Any channel $\mathcal{N}$ has an expansion in terms of Pauli operators. Consider a completely positive map $\mathcal{N}$ with Kraus operators $\{K_\alpha\}$, and expand each $K_\alpha$ as
\begin{equation}
K_\alpha = \sum_{i=0}^{d^2 -1} c_{\alpha i} \sigma^i,
\end{equation}
where all Pauli operators $\{\sigma^i\}$ are chosen to be Hermitian, and the $\{c_{\alpha i}\}$ are complex numbers.
Then
\begin{equation}\label{eq:N-Pauli-expand}
\mathcal{N}(\rho) = \sum_\alpha K_\alpha \rho K_\alpha^\dagger=\sum_{i,j=0}^{d^2-1} \chi_{ij} \sigma^i \rho \sigma^j,
\end{equation}
where 
\begin{equation}
\chi_{ij}= \sum_\alpha c_{\alpha i} c_{\alpha j}^*= \chi_{ji}^*.
\end{equation}
This is called the chi-matrix representation of the channel.
The map $\mathcal{N}$ is trace preserving if
\begin{equation}
\sum_{ij} \chi_{ij} \sigma^j \sigma^i = id,
\end{equation}
and unital if
\begin{equation}
\sum_{i,j} \chi_{ij} \sigma^i \sigma^j = id.
\end{equation}
Note that $\sigma^i \sigma^k \sigma^j = \pm \sigma^k$ if and only if $i=j$; therefore, in the Pauli transfer matrix language, the terms in equation (\ref{eq:N-Pauli-expand}) with $i=j$ contribute to the diagonal entries in $N_{ab}$, while the terms with $i\ne j$ contribute to the off-diagonal entries. 

To be more concrete, consider the single-qubit rotation about the $X$-axis $U^X(\theta) = \exp(\left(-i\theta \sigma^X/2\right)$, for which
\begin{align}
    \rho\to U^X(\theta)\rho U^X(\theta)^\dagger & = \cos^2(\theta/2) \rho + i \cos(\theta/2)\sin(\theta/2)\rho\sigma^X\notag\\
    & \quad -i\cos(\theta/2)\sin(\theta/2)\sigma^X\rho  + \sin^2(\theta/2) \sigma^X\rho\sigma^X;
\end{align}
hence
\begin{equation}
    \begin{pmatrix}
    \chi_{II} & \chi_{IX} \\
    \chi_{XI}& \chi_{XX}
    \end{pmatrix}
    = \begin{pmatrix}
    \frac{1}{2}(1+\cos\theta) & \frac{i}{2}\sin\theta\\
    -\frac{i}{2} \sin\theta & \frac{1}{2}(1-\cos\theta)
    \end{pmatrix}.
\end{equation}
More generally, for the channel with Pauli transfer matrix
\begin{equation}
N= \left(
\begin{array}{cccc}
1 & 0 & 0 & 0 \\
0 & 1 & 0 & 0 \\
0 & 0 & 1-\epsilon & \delta \\
0 & 0 & -\delta & 1-\epsilon \\
\end{array} \right),
\end{equation}
as in equation (\ref{eq:M-dephase-rotate}), we have 
\begin{equation}
    \begin{pmatrix}
    \chi_{II} & \chi_{IX} \\
    \chi_{XI}& \chi_{XX}
    \end{pmatrix}
    = \begin{pmatrix}
    1 - \epsilon / 2 & i\delta/2\\
    -i\delta/2 & \epsilon/ 2
    \end{pmatrix}.
\end{equation}

There is a simple general relationship between the off-diagonal entries of the Pauli transfer matrix $N_{ab}$ and the chi matrix $\chi_{ij}$, namely

\begin{Lemma}
\label{lemma:off-diagonal}
The off-diagonal elements of the Pauli transfer matrix $N_{ab}$ and the chi matrix $\chi_{ij}$ are related by
\begin{equation}
\sum_{a,b|a\ne b} N_{ab}^2 = d^2 \sum_{i,j|i\ne j} |\chi_{i j}|^2,
\end{equation}
where $d=2^n$ is the Hilbert space dimension.
\end{Lemma}
\noindent Because of this identity, we may quantify the coherence of a channel using the off-diagonal entries in either $N_{ab}$ or $\chi_{ij}$. The case $d=2$ is explained explicitly in Appendix \ref{appendix:off-diagonal}.
\begin{proof}
\noindent To prove the claim, note that, for any Hermitian Pauli operators $\sigma^i$, $\sigma^j$, $\sigma^a$, we have 
\begin{equation}\label{eq:eta-ij-ab}
\sigma^i \sigma^a \sigma^j = \eta_{ij}^{ab} \sigma^b
\end{equation}
for some Hermitian Pauli operator $\sigma^b$ and some phase $\eta_{ij}^{ab}$. By taking Hermitian adjoints of both sides, we also have
\begin{equation}\label{eq:eta-ij-ab-*}
\sigma^j \sigma^a \sigma^i= \eta_{ij}^{ab*}\sigma^b.
\end{equation}
The phase is $\eta_{ij}^{ab}=\pm 1$ if $\sigma^i \sigma^a \sigma^j$ is Hermitian, and it is $\eta_{ij}^{ab}=\pm i$ if $\sigma^i \sigma^a \sigma^j$ is anti-Hermitian. Furthermore, for each fixed $i\ne j$, as $\sigma^a$ ranges over the $d^2$ Hermitian Pauli operators, $\sigma^i\sigma^a\sigma^j$ is Hermitian for $d^2/2$ choices of $\sigma^a$, and anti-Hermitian for the remaining $d^2/2$ choices. (If $\sigma^i$ and $\sigma^j$ commute, then $\sigma^i\sigma^a\sigma^j$ is Hermitian if and only if $\sigma^a$ commutes with $\sigma^j\sigma^i$. If $\sigma^i$ and $\sigma^j$ anticommute, then $\sigma^i\sigma^a\sigma^j$ is Hermitian if and only if $\sigma^a$ anticommutes with $\sigma^j\sigma^i$.)
Note that $b\ne a$ if $i\ne j$.

The entries in the Pauli transfer matrix are (for $a\ne b$).
\begin{equation}
N_{ab} =\sum_{i,j|i\ne j}\eta_{ij}^{ab}\chi_{ij}= \sum_{i,j|i<j}\left( \eta_{ij}^{ab}\chi_{ij} + \eta_{ij}^{ab*} \chi_{ji}\right),
\end{equation}
where the sum is restricted to $\{i,j\}$ such that $\sigma^i \sigma^a \sigma^j \propto\sigma^b$.
The summand is $\left(\pm 1 \right) \left(\chi_{ij}+\chi_{ji}\right)$ if $\sigma^i\sigma^a\sigma^j$ is Hermitian, and it is $\left(\pm i \right) \left(\chi_{ij}-\chi_{ji}\right)$ if $\sigma^i\sigma^a\sigma^j$ is anti-Hermitian. Suppose now that, for fixed $i,j$, we collect all the terms in $\sum_{a\ne b} N_{ab}^2$ which are quadratic in $\{\chi_{ij},\chi_{ji}\}$. Because $\sigma^i\sigma^a\sigma^j$ is Hermitian for half the choices of $\sigma^a$ and anti-Hermitian for half the choices, we have
\begin{equation}
\frac{d^2}{2} \left(\chi_{ij} + \chi_{ji}\right)^2 - \frac{d^2}{2}\left(\chi_{ij}-\chi_{ji}\right)^2 = 2 d^2 \chi_{ij}\chi_{ji} =  d^2 \left(|\chi_{ij}|^2+ |\chi_{ji}|^2\right),
\end{equation}
where we have used $\chi_{ij}=\chi_{ji}^*$, which is required by complete positivity.

To complete the proof of the claim, we must verify that all the multilinear terms of the form $\chi_{ij}\chi_{kl}$ (where $\{i,j\}$ and $\{k,l\}$ are disjoint) cancel in the sum $\sum_{a\ne b} N_{ab}^2$. Such a cross term of the form 
\begin{equation}
\eta_{ij}^{ab}\eta_{kl}^{ab} \chi_{ij}\chi_{kl}
\end{equation}
arises in $N_{ab}^2$ when we have
\begin{align}
&\sigma^i \sigma^a \sigma^j = \eta_{ij}^{ab} \sigma^b,\notag\\
&\sigma^k \sigma^a \sigma^l= \eta_{kl}^{ab}\sigma^b.
\end{align}
We will consider all such terms with $i,j,k,l$ fixed, as we vary $\sigma^a$ and $\sigma^b$ over the possible Hermitian Pauli operators. Multiplying both sides on the left by Hermitian Pauli operator $\sigma^c$ we obtain
\begin{align}
\left(\sigma^c \sigma^i\sigma^c\right) \left(\sigma^c \sigma^a\right) \sigma^j & = \eta_{ij}^{ab} \left(\sigma^c\sigma^b\right),\notag\\
\left(\sigma^c \sigma^k\sigma^c\right) \left(\sigma^c \sigma^a\right) \sigma^l & = \eta_{kl}^{ab} \left(\sigma^c\sigma^b\right).
\end{align}
Given a standard sign choice for the $d^2$ Hermitian Pauli operators, we may write
\begin{equation}
\sigma^c \sigma^a = \phi_{ca}^{a'} \sigma^{a'}, \quad \sigma^c \sigma^b = \phi_{cb}^{b'} \sigma^{b'};
\end{equation}
here {\it e.g.} $\phi_{ca}^{a'}$ is a phase, which is $\pm 1$ if $\sigma^a$ and $\sigma^c$ commute and $\pm i$ if $\sigma^a$ and $\sigma^c$ anticommute. We also have
\begin{equation}
\sigma^c \sigma^i\sigma^c = \xi_{ic} \sigma^i, \quad \sigma^c \sigma^k\sigma^c = \xi_{kc} \sigma^k;
\end{equation}
here $\xi_{ic}=\pm 1$ is a sign indicating whether $\sigma^c$ and $\sigma^i$ commute or anticommute. Therefore 
\begin{align}
\sigma^i \sigma^{a'} \sigma^j & = \left(\xi_{ic}\phi_{ca}^{a'*}\phi_{cb}^{b'}\eta_{ij}^{ab}\right) \sigma^{b'},\notag\\
\sigma^k \sigma^{a'} \sigma^l & = \left(\xi_{kc}\phi_{ca}^{a'*}\phi_{cb}^{b'}\eta_{kl}^{ab}\right)\sigma^{b'},
\end{align}
and the corresponding cross term arising from $N_{a'b'}^2$ is 
\begin{equation}
\xi_{ic}\xi_{kc}\left(\phi_{ca}^{a'*} \phi_{cb}^{b'}\right)^2\eta_{ij}^{ab}\eta_{kl}^{ab} \chi_{ij}\chi_{kl}.
\end{equation}

Now suppose that either $\sigma^c$ commutes with both $\sigma^a$ and $\sigma^b$, or anticommutes with both; in either case $\left(\phi_{ca}^{a'*} \phi_{cb}^{b'}\right)^2=1$. As we vary $\sigma^c$ over the $d^2/2$ Pauli operators with this property, the sign $\xi_{ic}\xi_{kc}$ has the value $+1$ for the $d^2/4$ choices of $\sigma^c$ such that $\sigma^c$ commutes with both $\sigma^i$ and $\sigma^k$ or anticommutes with both, while  $\xi_{ic}\xi_{kc}$ has the value $-1$ for the $d^2/4$ choices of $\sigma^c$ such that $\sigma^c$ commutes with one of  $\sigma^i$ and $\sigma^k$ and anticommutes with the other. Therefore, as we vary $a'$ and $b'$ over these $d^2/2$ possible choices for $\sigma^c$, with $i,j,k,l$ fixed, the cross terms cancel. 

Alternatively, suppose that $\sigma^c$ commutes with one of $\sigma^a$ and $\sigma^b$ and anticommutes with the other; then $\left(\phi_{ca}^{a'*} \phi_{cb}^{b'}\right)^2=-1$. Again, as we vary $a'$ and $b'$ over the $d^2/2$ possible choices for $\sigma^c$, with $i,j,k,l$ fixed, $\xi_{ic}\xi_{kc}=+1$ for half of the terms and $\xi_{ic}\xi_{kc}=-1$ for the other half; therefore the cross terms cancel. This completes the proof. 
\end{proof}

\section{Logical channel for the repetition code}
\label{sec:rep-code}

From now on we will use the streamlined notation for single-qubit Pauli operators:
\begin{equation}
I = \left(
\begin{array}{cc}
1 & 0  \\
0 & 1 \\
\end{array} \right),\quad 
X = \left(
\begin{array}{cc}
0& 1 \\
1 & 0 \\
\end{array} \right), \quad
Y = \left(
\begin{array}{cc}
0 & -i \\
i & 0 \\
\end{array} \right),\quad 
Z = \left(
\begin{array}{cc}
1 & 0  \\
0 & -1 \\
\end{array} \right).
\end{equation}
Consider the repetition code, which protects one logical qubit against bit flip ($X$) errors, but provides no protection against phase ($Z$) errors. Let us analyze how well this code protects against coherent errors, in which each physical qubit in the code block rotates about the $X$-axis. Similar calculations were carried out in \cite{greenbaum2017modeling, huang2018performance}. Understanding this example will prepare us for an analysis of more general stabilizer codes.

To be as concrete as possible, we will start with the simplest interesting case, the 3-qubit repetition code spanned by $|000\rangle$ and $|111\rangle$. Our goal is to determine the logical channel that results when rotation errors applied to the physical qubits are followed by error correction. We will assume for now that the same rotation is applied to each of the three qubits; this will be generalized later. 

Suppose that each physical qubit is subjected to the unitary rotation
\begin{equation}\label{eq:x-rotation-3times}
U^X(\theta) = cI -is X, \quad c = \cos (\theta/2), \quad s = \sin (\theta/2);
\end{equation}
thus the product unitary map applied to the three physical qubits is
\begin{align}\label{eq:U-expand}
U^X(\theta)^{\otimes 3} &= c^3 III -ic^2 s\left(X II + IXI + IIX\right)\notag\\
& \quad - cs^2\left(X XI + X I X + IXX\right) +is^3 XXX.
\end{align}
To perform error correction we measure the operators $ZZI$ and $I ZZ$ to obtain two syndrome bits. If the syndrome is trivial (both measurements yield +1) no further action is required. If the syndrome is nontrivial, $X$ is applied to one of the three qubits, returning the state to the code space. Thus the terms in the expansion in equation (\ref{eq:U-expand}) with weight 0 or 1 (where the weight is the number of $X$'s) are error corrected to the logical operator $ \bar I= III$, while terms with weight 2 or 3 are error corrected to the logical operator $\bar X = XXX$. We conclude that the logical channel $\mathcal{N}_L$ is a convex combination of two unitary transformations,
\begin{equation}
\mathcal{N}_L(\rho) = p_0 U^X(\theta_0)\rho U^X(\theta_0)^\dagger + p_1 U^X(\theta_1) \rho U^X(\theta_1)^\dagger, 
\end{equation}
where
\begin{align}
p_0 & = c^6 +s^6, & \quad\theta_0/2 & =  \arctan(-s^3/c^3) \nonumber\\
p_1 & = 3\left(c^4 s^2 + c^2 s^4\right), & \quad\theta_1/2 & = \arctan(s/c) =\theta/2.
\end{align}
A logical rotation by $\theta_0$ is applied when the syndrome is trivial (weight $0$), and a logical rotation by $\theta_1$ is applied when the syndrome is nontrivial (weight $1$).

The logical channel has the form specified in equation (\ref{eq:M-dephase-rotate}), where
\begin{equation}
\epsilon = p_0(1-\cos\theta_0) + p_1(1-\cos\theta_1), \quad
\delta = p_0 \sin\theta_0 + p_1 \sin\theta_1 .
\end{equation}
These expressions for $\epsilon$ and $\delta$ can be simplified using trigonometric identities. In terms of $s/c = t = \tan\theta/2$, we have 
\begin{align}
p_0 & = c^6(1+t^6), & 1 - \cos\theta_0 & = \frac{2t^6}{1+t^6}, & \sin\theta_0 & = \frac{-2t^3}{1+t^6},\notag \\
p_1 & = 3c^6t^2(1+t^2), & 1-\cos\theta_1 & = \frac{2t^2}{1+t^2}, & \sin\theta_1 & = \frac{2t}{1+t^2};
\end{align}
therefore we find
\begin{equation}\label{eq:epsilon-delta-rep}
\epsilon = 2 s^6 + 6 c^2 s^4, \quad \delta = -2s^3c^3 + 6 s^3c^3 = 4s^3c^3.
\end{equation}
Expanding to leading order for small $\theta$, we have
\begin{equation}\label{eq:eps-delta-rep-leading}
\epsilon \approx \frac{3}{8} \theta^4, \quad \delta\approx \frac{1}{2} \theta^3.
\end{equation}
Here, because $\epsilon$ is higher order in $\theta$ than $\delta$, equation (\ref{eq:coherence-off-diagonal}) applies, and therefore the coherence angle is 
\begin{equation}
\Theta^2 \approx 2 \delta^2/3\approx \theta^6 / 6.
\end{equation}
From equation (\ref{eq:Nm-epsilon-delta-bound}), we see that if this logical channel $\mathcal{N}_L$ is applied $m$ times, the infidelity becomes
\begin{equation}\label{eq:rm-bound-3-qubit}
r_m\approx \frac{1}{3} m \epsilon +\frac{1}{6} m(m-1)\delta^2 \approx \frac{1}{8} m \theta^4 + \frac{1}{24} m(m-1) \theta^6.
\end{equation}
Note that the term quadratic in $m$ actually matches the upper bound in equation (\ref{eq:coherence-angle-bound}). 
Equation (\ref{eq:rm-bound-3-qubit}) reveals that the coherence of the logical channel is somewhat suppressed, as it takes a number of repetitions $m = \mathcal{O}(\theta^{-2})$ for the quadratically growing contribution to $r$ to ``catch up'' with the dominant linear term. 

Now let's do a similar analysis for the length-$n$ repetition code (where $n$ is odd), which corrects up to $(n-1)/2$ bit-flip errors. In this case the logical channel is a convex combination of $(n+1)/2$ unitary rotations,
\begin{equation}
\mathcal{N}_L(\rho) = \sum_{w=0}^{(n-1)/2} p_w U^X(\theta_w) \rho U^X(\theta_w)^\dagger
\end{equation}
where $w$ ranging from $0$ to $(n-1)/2$ indicates the weight of a correctable $X$ error occurring in the expansion of $\left(c - isX\right)^{\otimes n}$. When the $(n{-}1)$-bit syndrome is measured, syndromes pointing to a weight-$w$ error occur with total probability
\begin{equation}
p_w= \binom{n}{w} \left[c^{2(n-w)}  s^{2w}+ c^{2w} s^{2(n-w)}\right]
= \binom{n}{w} c^{2n}  t^{2w}\left[ 1+ t^{2(n-2w)}\right],
\end{equation}
and the logical rotation angle conditioned on a weight-$w$ syndrome is
\begin{align}
&\theta_w/2 = (-1)^{(n-1-2w)/2}\arctan\left[ (s/c)^{n-2w}\right]\notag\\
&\implies 1-  \cos\theta_w = \frac{2t^{2(n-2w)}}{1 + t^{2(n-2w)}}, \quad
\sin\theta_w = (-1)^{(n-1)/2}(-1)^w\frac{2t^{n-2w}}{1 + t^{2(n-2w)}}.
\end{align}
Summing over the weight of the syndrome we find
\begin{align}\label{eq:eps-delta-all-orders}
\epsilon  &= \sum_{w=0}^{(n-1)/2} p_w\left(1-\cos\theta_w\right) = \sum_{w=0}^{(n-1)/2}\binom{n}{w} \left( c^2 \right)^w \left( s^2\right)^{n-w}, \notag\\
\delta &= \sum_{w=0}^{(n-1)/2} p_w\sin\theta_w = (-1)^{(n-1)/2}c^n s^n\sum_{w=0}^{(n-1)/2} (-1)^w \binom{n}{w} = 2 \binom{n-1}{\frac{n-1}{2}} c^n s^n.
\end{align}
In Appendix \ref{appendix:sum} we use Stirling's approximation to evaluate the sum in the expression for $\epsilon$. Applying Stirling's approximation to our expression for $\delta$ as well, we have proven
\begin{Theorem}\label{Theorem: rep-code-epsilon-delta}
Consider the length-$n$ repetition code which protects against bit flip ($X$) errors, subject to the independent unitary noise map $U=\left(\left(\cos\theta/2\right) I -i \left(\sin\theta/2\right) X\right)^{\otimes n} $, where $\sin^2\theta/2 < 1/2$. Let $\mathcal{N}_L(\rho)= \mathcal{R}\left(U\rho U^\dagger\right)$ be the logical map, where $\rho$ is a code state and $\mathcal{R}$ decodes using majority voting. Then $\mathcal{N}_L$ has Pauli transfer matrix $N$ of the form given in equations (\ref{eq:Pauli-transfer-N}) and (\ref{eq:M-dephase-rotate}), with $\epsilon$ and $\delta$ given by
\begin{align}
    \epsilon & = \sqrt{\frac{2}{\pi n}}\left(\frac{\sin^{n+1}\theta}{\cos\theta}\right) \left(1 + \mathcal{O}\left(\frac{1}{n}\right)\right),\notag\\
  \delta &=  \sqrt{\frac{2}{\pi n}} \sin^n\theta \left(1 + \mathcal{O}\left(\frac{1}{n}\right)\right)= \left(\frac{\cos\theta}{\sin\theta}\right)\epsilon \left(1 + \mathcal{O}\left(\frac{1}{n}\right)\right).  
\end{align}
\end{Theorem}
\noindent Therefore, using equation (\ref{eq:Nm-epsilon-delta-bound}) and approximations that are well justified (according to Theorem \ref{Theorem: rep-code-epsilon-delta}) when $n$ is large and $\sin^2\theta/2 < 1/2$, we can estimate the infidelity when the logical channel is applied $m$ times is succession, finding
\begin{equation}
r_m\approx \frac{1}{3} m \epsilon +\frac{1}{6} m(m-1)\delta^2 \approx  \frac{1}{3}\sqrt{\frac{2}{\pi n}}\left[  m  \left(\frac{\sin^{n+1}\theta}{\cos\theta}\right) + \frac{1}{\sqrt{2\pi n}} m(m-1) \sin^{2n}\theta\right].
\end{equation}
The scaling of the infidelity $r=\mathcal{O}(\theta^{n+1})$ arises because a bit flip error must have weight at least $w=(n+1)/2$ to cause a logical error. The scaling $\mathcal{O}(\theta^{2n})$ of the term quadratic in $m$ indicates that the coherence of the logical channel is suppressed when $\theta$ is small. It takes $m\approx \sqrt{2\pi n}/\theta^{n-1}$  successive applications of the logical channel $\mathcal{N}_L$ for the quadratic term in $r_m$ to become comparable to the linear term. This suppression arises because larger logical rotations occur with only smaller probability; for example a logical rotation by $\theta$ occurs  with probability $\mathcal{O}(\theta^{n-1})$. 

Keeping only the leading-order terms in equation (\ref{eq:eps-delta-all-orders}) we obtain
\begin{equation}\label{eq:epsilon-delta-length-n-methodA}
\epsilon  \approx 2 \binom{n}{\frac{n-1}{2}} \left(\frac{\theta}{2}\right)^{n+1},\quad 
\delta \approx 2 \binom{n-1}{\frac{n-1}{2}} \left(\frac{\theta}{2}\right)^n
\implies 
\delta \approx \frac{n+1}{n} \theta^{-1}\epsilon ,
\end{equation}
generalizing equation (\ref{eq:eps-delta-rep-leading}). 
We derived the relationship
\begin{equation}\label{eq:compare-epsilon-delta}
\epsilon \approx \frac{n}{n+1}\left(\theta\delta\right)
\end{equation}
using the identity
\begin{equation}\label{eq:binomial-identity}
\sum_{w=0}^{(n-1)/2}(-1)^{w} \binom{n}{w} = (-1)^{(n-1)/2} \binom{n-1}{\frac{n-1}{2}},
\end{equation}
which can be proved by induction. For drawing the conclusion that $\theta\delta/\epsilon$ is bounded above by an $n$-independent constant, the oscillating minus sign in this expression is important --- if not for the oscillating sign, the sum would be $2^{n-1}$, hence larger than equation (\ref{eq:binomial-identity}) by a factor which scales like $\sqrt{n}$. This would mean that average infidelity $r_m$ in equation (\ref{eq:Nm-epsilon-delta-bound}) would have a large quadratic component relative to the linear component as the code length $n$ becomes large. In other words, the logical noise channel would have significant coherence.

\section{Repetition code revisited}\label{sec:rep-code-revisted}

In this section we will compute the logical channel for the repetition code using a different method than in Section \ref{sec:rep-code}. This new method can be extended more easily to general stabilizer codes.

\subsection{Stabilizer formalism}

We now briefly review the structure of stabilizer codes, as this will be used in our analysis. Let $\{g_\alpha, \alpha = 1,2, \dots, n-k\}$ denote the $n-k$ stabilizer generators for an $[[n,k,d]]$ stabilizer code. These generators are mutually commuting Hermitian Pauli operators such that $g_\alpha^2 =I$.  The \textit{syndrome} $s(\sigma^i)$ of Pauli operator $\sigma^i$ is a length-$(n-k)$ binary vector such that $s(\sigma^i)_\alpha =s_{\alpha}^i$ where
\begin{equation}
g_\alpha \sigma^i = (-1)^{s_\alpha^i}\sigma^i g_\alpha.
\end{equation}
Note that the syndrome of a product of Pauli operators is additive: $s(\sigma^i\sigma^j)_\alpha=\sigma^i_\alpha + \sigma^j_\alpha$, where the addition is modulo 2. 

The code space is the simultaneous eigenstate with eigenvalue $1$ of all the stabilizer generators. If $|\bar\psi\rangle$ is a pure state in the code space, then
\begin{equation}
g_\alpha\left( \sigma^i|\bar \psi\rangle\right) = (-1)^{s_\alpha^i} \sigma^i g_\alpha|\bar\psi\rangle = (-1)^{s_\alpha^i} \sigma^i |\bar\psi\rangle.
\end{equation}
Therefore, the syndrome of $\sigma^i$ can be identified by measuring all of the stabilizer generators. Hence we may say that $s\left(\sigma^i\right)$ is the syndrome of the state $\sigma^i|\bar\psi\rangle$. A Pauli operator that commutes with the stabilizer generators preserves the code space and is said to be \textit{logical}.
We may define a complete set of orthogonal projectors $\{\Pi_s\}$ on the $n$-qubit Hilbert space, where $\Pi_s$ projects onto the subspace with syndrome $s$. Then
\begin{equation}
\Pi_s\Pi_t = \delta_{st}\Pi_s, \quad \sum_s \Pi_s = I.
\end{equation}
An encoded density operator $\bar \rho$ (one supported on the code space) has the property
\begin{equation}
\Pi_s \bar\rho \Pi_t = \delta_{s0}, \delta_{t0} \bar\rho,
\end{equation}
where $s=0$ denotes the trivial syndrome. 

To construct the error recovery map $\mathcal{R}$, we first perform an orthogonal measurement to identify the syndrome $s$. Then, for each syndrome $s$, a particular Pauli operator $E_s^\dagger$ is applied, which returns the measured state to the code space; therefore,
\begin{equation}
\mathcal{R}(\rho) = \sum_s E_s^\dagger \Pi_s \rho \Pi_s E_s.
\end{equation}
One says that $E_s$ is the \textit{standard error} associated with the syndrome $s$.
In the case of \textit{minimal-weight decoding}, $E_s$ is chosen to be a minimal-weight Pauli operator with syndrome $s$. By the weight $w(\sigma)$ of the $n$-qubit Pauli operator $\sigma$, we mean the number of qubits to which a nontrivial Pauli matrix $X$, $Y$, or $Z$ is applied, while $I$ is applied to the remaining $n-w$ qubits. 

By summing over all values of the syndromes $s$ to construct the error recovery channel we are averaging over all the possible outcomes of the syndrome measurement, with each syndrome weighted by its probability. We discussed in the introduction how to justify performing this average when computing the logical channel.

\subsection{Recovery in the chi-matrix representation}
\label{sec:recovery-chi}

For any such noise channel $\mathcal{N}$ acting on an encoded density operator $\bar \rho$, we would like to find the error corrected map $\mathcal{R}\circ \mathcal{N}(\bar \rho)$. Using the chi representation of the noise channel, it evidently suffices to compute
\begin{equation}
\mathcal{R}\left( \sigma^i \bar \sigma^k \sigma^j\right)
\end{equation}
for each pair of physical Pauli operators $\sigma^i$, $\sigma^j$ and each logical Pauli operator $\bar \sigma^k$. Because the syndrome is additive, we have
\begin{equation}
\Pi_s P_t \Pi_s = \delta_{t0} P_0 \Pi_s
\end{equation}
if $P_t$ is any physical Pauli operator with syndrome $t$,
and therefore
\begin{align}\label{eq:R=0-unless}
&\mathcal{R}\left(\sigma^i\bar\sigma^k\sigma^j\right) = \sum_s E_s^\dagger \Pi_s \sigma^i\bar\sigma^k\sigma^j \Pi_s E_s = 0\notag\\
&\quad \textrm{unless}\quad s\left(\sigma^i\bar\sigma^k\sigma^j\right)= s\left(\sigma^i\sigma^j\right) = 0.
\end{align}
That is, only the terms for which $\sigma^i$ and $\sigma^j$ have the same syndrome survive when the error recovery map is applied. This property will be crucial in our analysis of the logical channel.

Now let's understand the action of $\mathcal{R}$ in more detail. An $[[n,k,d]]$ stabilizer code has $4^k$ logical Pauli operators. The physical Pauli operator $L$ representing a logical Pauli operator is not unique, because $L$ and $LG$ act in the same way on the code space, where $G$ is any element of the stabilizer group. But let us by convention choose standard physical operators $\{L_a, a = 0, 1,2, \dots 4^k-1\}$ representing each of the logical Pauli operators. Since we have also assigned a standard error operator $E_s$ to each syndrome $s$, any Hermitian Pauli operator  has a unique decomposition of the form
\begin{equation}
\sigma(s,a,x) = \eta_{sax} E_s L_a G_x, \quad \eta_{sax}\in\{\pm1, \pm i\},
\end{equation}
where $G_x$ is an element of the stabilizer group, and $\eta_{sax}$ is a phase. Since there are $2^{n-k}$ stabilizer group elements (up to phases), $2^{n-k}$ distinct syndromes, and $2^{2k}$ logical Pauli operators, we see that this decomposition accounts for all $4^n$ physical Pauli operators. We conclude that if $\bar\rho$ is an encoded density operator, then
\begin{align}\label{eq:R-action}
\mathcal{R}\left(\sigma(s,a,x)\bar\rho \sigma(s',a',x') \right) &= \delta_{ss'} E_s^\dagger \left(\eta_{sax}E_s L_a G_x\right) \bar\rho\left( G_{x'}^\dagger L_{a'}^\dagger E_s^\dagger\eta_{s'a'x'}^*\right) E_s \notag\\
&= \delta_{ss'} \,\eta_{sax} \eta_{s'a'x'}^* \,L_a \bar \rho L_{a'}^\dagger,
\end{align}
where we have used the property that $\sigma(s',a',x')$ is Hermitian. 
In the logical channel, the terms with $L_a= L_{a'}$ are \textit{incoherent} -- they contribute to the on-diagonal elements of the logical Pauli transfer matrix. The terms with $L_a \ne L_{a'}$ are \textit{coherent} -- they contribute to the off-diagonal elements. 

When the noise channel $\mathcal{N}$ is weak, the dominant terms in the chi-matrix expansion in equation (\ref{eq:N-Pauli-expand}) are those such that $\sigma^i \sigma^j$ has minimal weight, and we have also seen that the only terms that survive when the recovery map is applied are those such that $\sigma^i \sigma^j$ is a logical operator (has trivial syndrome). Now let's suppose that the code distance is $d$ and that minimal-weight decoding is performed. This means that we choose $E_s$ such that $L_a = I$ (up to multiplication by an element of the stabilizer) whenever $\sigma(s,a,x)$ has weight no larger than $(d-1)/2$, assuming $d$ is odd. 

To get a contribution to the incoherent part of the logical channel, we will need both $\sigma^i$ and $\sigma^j$ to have weight at least $(d+1)/2$, so that the total weight must be at least $d+1$. In that case it is possible for both $\sigma^i$ and $\sigma^j$ to be error corrected to a nontrivial logical operator. But there are also weight-$d$ contributions to the coherent part of the logical channel, arising from the terms in which $w(\sigma^i)+w(\sigma^j) = d$, where $w(\sigma)$ denotes the weight of Pauli operator $\sigma$. In that case, one of the two Pauli operators has weight less than or equal to $(d-1)/2$, hence is error corrected to the identity, while the other has weight greater than or equal to $(d+1)/2$, hence is corrected to a nontrivial logical operator $L$. The resulting term in the logical channel is either $L\bar\rho$ or $\bar \rho L$ (up to a phase), depending on whether $\sigma^i$ or $\sigma^j$ has higher weight. 

If we choose the standard errors $\{E_s\}$ differently, then the action of the recovery operator may be modified. But it is evident from equation (\ref{eq:R-action}) that if we make the replacement $E_s\to E_s' = \phi_s E_s G_y$, where $G_y$ is an element of the stabilizer and $\phi_s$ is a phase, then $\mathcal{R}\left(\sigma\bar\rho\sigma'\right)$ is not changed. In particular, when we perform minimal-weight decoding, there may be more than one minimal-weight Pauli operator with syndrome $s$, so that the choice of $E_s$ is ambiguous. However, as long as any two minimal-weight Pauli operators $E_s$ and $E_s'$ with syndrome $s$ have the property that $E_s'^\dagger E_s$ is an element of the code stabilizer, then the logical channel will not depend on how the minimal-weight standard errors are chosen. This will certainly be the case if the code distance is $d$ and the standard errors have weight not larger than $(d-1)/2$, since then $E_s'^\dagger E_s$ has weight at most $d-1$ and cannot be a nontrivial logical operator.

\subsection{Analysis of repetition code using the chi-matrix formalism}
\label{subsec:rep-chi}
To illustrate this method, we return to the length-3 repetition code, where the noise channel is as in equation (\ref{eq:U-expand}). We write out the chi-matrix expansion of $\mathcal{N}(\rho)$ in equation (\ref{eq:N-Pauli-expand}), and then apply the recovery operator $\mathcal{R}$ to find the logical channel $\mathcal{N}_L = \mathcal{R}\circ \mathcal{N}$. The task of applying $\mathcal{R}$ is simplified by the observation that, if the state $\rho$ is supported on the code space, then $\mathcal{R}$ annihilates all terms in which $\sigma^i\sigma^j$ is not logical; that is, as indicated in equation (\ref{eq:R=0-unless}), $\sigma^i\sigma^j$ must commute with the stabilizer for the term to survive. We may write
\begin{equation}
\mathcal{N}(\rho) = \mathcal{N}_{\rm incoh}(\rho) +\mathcal{N}_{\rm coh}(\rho) + \mathcal{N}_{\rm null}(\rho),
\end{equation}
where $\mathcal{N}_{\rm null}$ is the sum of terms such that $\sigma^i\sigma^j$ is not logical (hence $\mathcal{R}\circ\mathcal{N}_{\rm null} = 0$ acting on encoded density operators), $\mathcal{N}_{\rm incoh}$ is the sum of terms such that $\sigma^i\sigma^j$ is the logical identity, and $\mathcal{N}_{\rm coh}$ is the sum of terms such that $\sigma^i\sigma^j$ is a nontrivial logical operator. Then $\mathcal{R}\circ\mathcal{N}_{\rm incoh}$ is the incoherent part of $\mathcal{N}_L$ and $\mathcal{R}\circ\mathcal{N}_{\rm coh}$ is its coherent part. Explicitly, 
\begin{align}\label{eq:pauli-expand-rep-incoh}
\mathcal{N}_{\rm incoh}(\rho) & = c^6 III \rho III + c^4s^2 \left( XII \rho XII + IXI \rho IXI + IIX \rho IIX \right) \nonumber
\\
&\quad + c^2 s^4\left( XXI \rho XXI + XIX \rho XIX + IXX \rho IXX \right) + s^6 XXX \rho XXX, 
\end{align}
and
\begin{align}\label{eq:pauli-expand-rep-coh}
\mathcal{N}_{\rm coh}(\rho) &= ic^3s^3 \left( XXX \rho III - III \rho XXX\right) + ic^3s^3 \left(XII \rho IXX + IXI \rho XIX + IIX\rho XXI  \right.
\nonumber
\\
&\quad   \left. - IXX \rho XII - XIX \rho IXI -XXI \rho IIX \right).
\end{align}
The code has two syndrome bits, given by the measured values of $ZZI$ and $IZZ$, and for a minimal-weight decoder we choose the standard errors to be
\begin{equation}
E_{00} = III, \quad E_{01}= IIX, \quad E_{10}= XII, \quad E_{11} = IXI, 
\end{equation}
while the nontrivial logical operator is $\bar X = XXX$. Each of the Pauli operators in equations (\ref{eq:pauli-expand-rep-incoh}) and (\ref{eq:pauli-expand-rep-coh}) can be expressed as a product of a standard error and a logical operator which is either $\bar I =III$ or $\bar X$, so the logical map becomes
\begin{align}
\mathcal{N}_{L,\rm incoh}(\rho) & = \mathcal{R}\circ \mathcal{N}_{\rm incoh}(\rho) = \left(c^6 +3c^4s^2\right)\rho +\left(3c^2s^4 + s^6\right)\bar X\rho \bar X,\notag\\
\mathcal{N}_{L,\rm coh}(\rho) & = \mathcal{R}\circ \mathcal{N}_{\rm coh}(\rho) = ic^3s^3\left([\bar X, \rho] - 3 [\bar X, \rho]\right).
\end{align}

To compare with our previous calculation of the logical channel, we note that
\begin{align}
\mathcal{N}_{L,\rm incoh}(\bar I) &= \left(c^2 + s^2\right)^3 \bar I = \bar I, \notag\\
 \mathcal{N}_{L,\rm incoh}(\bar X) &= \left(c^2 + s^2\right)^3 \bar X = \bar X, \notag\\
\mathcal{N}_{L,\rm incoh}(\bar Y) &= \left[\left(c^2 + s^2\right)^3 - 6 c^2 s^4 -  2 s^6\right] \bar Y= \left(1- 6 c^2 s^4 -  2 s^6\right)\bar Y, \notag \\
\mathcal{N}_{L,\rm incoh}(\bar Z) &=  \left[\left(c^2 + s^2\right)^3 - 6 c^2 s^4 -  2 s^6\right] \bar Z= \left(1- 6 c^2 s^4 -  2 s^6\right)\bar Z,
\end{align}
and
\begin{align}
\mathcal{N}_{L,\rm coh}(\bar I) & = \mathcal{N}_{L,\rm coh}(\bar X)= 0, \notag\\
 \mathcal{N}_{L,\rm coh}(\bar Y) &= -2ic^3s^3 [\bar X,\bar Y] = 4 c^3s^3 \bar Z, \notag\\
\mathcal{N}_{L,\rm coh}(\bar Z) &= -2ic^3s^3 [\bar X,\bar Z] = -4 c^3s^3 \bar Y.
\end{align}
In the notation of equation (\ref{eq:M-dephase-rotate}) we have found that the logical channel is parametrized by
\begin{equation}
\epsilon = 6 c^2 s^4 + 2 s^6, \quad \delta = 4 c^3s^3,
\end{equation}
in agreement with the result found in equation (\ref{eq:epsilon-delta-rep}).

Now consider the length-$n$ repetition code, for $n$ odd, where the noise is the product unitary transformation $U^X(\theta)^{\otimes n}$. The incoherent part $\mathcal{N}_{L, \rm incoh}$ of the logical channel arises from the diagonal terms $\{\sigma^i \rho \sigma^i\}$ in the chi-matrix expansion of $\mathcal{N}(\rho)$. Here $\sigma^i$ can be any one of the $2^n$ Pauli operators contained in $\{I, \sigma^X\}^{\otimes n}$. The code can correct $t = (n-1)/2$ $\sigma^X$ errors, so $\sigma^i$ is error corrected to $\bar I$ if its weight $w(\sigma^i)$ is $t$ or less, and is error corrected to $\bar X$ if its weight is $t+1$ or more. Therefore, if $\rho$ is an encoded density operator then
\begin{equation}
\mathcal{N}_{L,\rm incoh}(\rho) = \left( \sum_{w=0}^{t} \binom{n}{w} c^{2(n-w)} s^{2w} \right) \rho + \left( \sum_{w=t+1}^{n} \binom{n}{w} c^{2(n-w)} s^{2w}\right)\bar X \rho \bar X,
\end{equation}
where the binomial coefficient $\binom{n}{w}$ counts the number of weight-$w$ (or weight-($n{-}w$)) operators. Using
\begin{equation}
\sum_{w=0}^{n} \binom{n}{w} c^{2(n-w)} s^{2w} = \left(c^2+s^2\right)^n = 1,
\end{equation}
we see that $\mathcal{N}_{L,\rm incoh}(\bar I) = \bar I$ and $\mathcal{N}_{L,\rm incoh}(\bar X)= \bar X$, and furthermore
\begin{equation}
\mathcal{N}_{L,\rm incoh}(\bar Y) = \left(1 - 2\sum_{w=t+1}^{n} \binom{n}{w} c^{2(n-w)} s^{2w} \right)\bar Y;
\end{equation}
hence 
\begin{equation}
\epsilon = 2\sum_{w=(n+1)/2}^{n} \binom{n}{w} c^{2(n-w)} s^{2w} 
\end{equation}
in agreement with equation (\ref{eq:eps-delta-all-orders}).
To leading order in $s \approx \theta/2$ this becomes
\begin{equation}
\epsilon\approx 2 \binom{n}{\frac{n+1}{2}} \left(\frac{\theta}{2}\right)^{n+1},
\end{equation}
as in equation (\ref{eq:epsilon-delta-length-n-methodA}).

The coherent part $\mathcal{N}_{L,\rm coh}$ of the logical channel arises from the terms in the Pauli operator expansion of $\mathcal{N}(\rho)$ such that $\sigma^i\sigma^j = \bar X$. There are $2^n$ such terms --- $\sigma^i$ can be any operator among $\{I,X\}^{\otimes n}$, and $\sigma^j$ is then the complementary operator with $X$ and $I$ interchanged. If $\sigma^i$ has weight $\le t$, and so is error corrected to $\bar I$, then $\sigma^j$ has weight $\ge (t+1)$, and so is error corrected to $\bar X$. We obtain
\begin{align}
\mathcal{N}_{L,\rm coh}(\rho)& =\left(\sum_{w=0}^t \binom{n}{w} c^{w}(-is)^{n-w} c^{n-w} (is)^{w} \right) \bar X \rho+ \left(\sum_{w=0}^t \binom{n}{w} c^{n-w} (-is)^{w}c^{w}(is)^{n-w} \right)\rho \bar X\notag\\
&= (-i)^nc^n s^n \left(\sum_{w=0}^t (-1)^w \binom{n}{w} \right)[\bar X,\rho].
\end{align}
Therefore,
\begin{equation}
\mathcal{N}_{L,\rm coh}(\bar Y) = 2(-i)^{n-1}(cs)^n\left(\sum_{w=0}^t(-1)^w \binom{n}{w}  \right)\bar Z;
\end{equation}
hence 
\begin{equation}
\label{eq: Delta calculation chi matrix approach}
\delta =2(-i)^{n-1}\left(\sum_{w=0}^{(n-1)/2} (-1)^w \binom{n}{w} \right)(cs)^n\approx 2 \binom{n-1}{\frac{n-1}{2}}\left(\frac{\theta}{2}\right)^n,
\end{equation}
in agreement with equation (\ref{eq:epsilon-delta-length-n-methodA}).

\subsection{Inhomogeneous $X$-axis rotations}

Now let's consider the logical channel obtained by decoding the length-$n$ repetition code, in the case where the rotation angle varies from qubit to qubit. That is, the unitary noise channel is
\begin{equation}
U^X(\theta_1, \theta_2, \dots \theta_n) =\bigotimes_{\alpha =1}^n \left(c_\alpha - i s_\alpha \sigma^X\right),
\end{equation}
where $c_\alpha = \cos \theta_\alpha/2$ and $s_\alpha = \sin \theta_\alpha/2$.
As in our previous derivation for the case where all angles are equal, we can calculate the incoherent and coherent parts of the logical channel by expanding this tensor product and isolating the terms in $\mathcal{N}(\rho)$ of the form $\sigma^i\rho\sigma^j$ where $\sigma^i\sigma^j$ is either a trivial logical operator (for the incoherent part) or a nontrivial logical operator (for the coherent part). The only difference from the previous calculation is that, while previously all terms in the expansion of $U^X$ of with the same weight occurred with equal amplitudes, now operators of the same weight may have different amplitudes. 

Still, the derivation goes through in much the same way as before. Let $S$ denote a subset of the $n$ qubits,  let $|S|$ denote the size of $S$, and let $\bar S$ denote the subset complementary to $S$. Extending our previous argument to the case of unequal angles yields
\begin{align}
\epsilon &= 2 \sum_{S, |S|\ge t+1}\prod_{\alpha\in S}\prod_{\bar \alpha \in \bar S} c_{\bar \alpha}^2 s_\alpha^2\notag,\\
\delta  &=(-2i) \sum_{S, |S|\le t}\prod_{\alpha\in S}\prod_{\bar \alpha \in \bar S} c_{\bar\alpha}(-is_\alpha) c_\alpha (is_{\bar \alpha})\notag\\
&= (-2i)\prod_{\alpha = 1}^n (ic_\alpha s_\alpha)\sum_{S, |S|\le t} (-1)^{|S|}.
\end{align}
Note that the sum in the expression for $\delta$ does not depend on the angles. To leading order in the small $\{s_\alpha\}$, we find
\begin{align}
\epsilon &= 2\sum_{S, |S|= (n+1)/2}\prod_{\alpha\in S}s_\alpha^2+ \cdots\notag,\\
\delta &= 2 \binom{n-1}{\frac{n-1}{2}}\prod_{\alpha=1}^n s_\alpha,
\end{align}
where we have used the identity
\begin{equation}
\sum_{S, |S|\le \frac{n-1}{2}} (-1)^{|S|} = \sum_{w=0}^{(n-1)/2}(-1)^{w} \binom{n}{w} =(-1)^{(n-1)/2} \binom{n-1}{\frac{n-1}{2}}.
\end{equation}
As before we find $\epsilon = \mathcal{O}(s^{n+1})$ and $\delta = \mathcal{O}(s^n)$. Furthermore, the expression for $\delta$ is very simple --- the same as our previous formula, but with $s^n$ replaced by $\prod_\alpha s_\alpha$. 

The formula for $\epsilon$ depends in a more complicated way on the set of angles $\{\theta_\alpha\}$. But we can show that for fixed $\delta$, $\epsilon$ is minimized when all the $s_\alpha$ are equal. Therefore, we have a lower bound on $\epsilon$, namely
\begin{equation}
\epsilon \ge 2 \binom{n}{\frac{n+1}{2}} s^{n+1} + \cdots
\end{equation}
where the ellipsis indicates terms higher order in $s$, and we have defined
\begin{equation}
s^n = \prod_{\alpha=1}^n s_\alpha.
\end{equation}
Correspondingly, using
\begin{equation}
\binom{n}{\frac{n+1}{2} }= \frac{2n}{n+1} \binom{n-1}{\frac{n-1}{2}},
\end{equation}
we have the upper bound on $\delta$:
\begin{equation}\label{eq:delta-bound-inhomo-rot}
\delta \le \frac{n+1}{2n}\left(\frac{\epsilon}{s}\right)+\cdots.
\end{equation}
Therefore, for inhomogeneous as well as homogeneous rotations, we conclude that the coherent part of the logical channel is suppressed. In fact, the case where all rotation angles are equal is the worst case, where equation (\ref{eq:delta-bound-inhomo-rot}) is saturated. 

Now let's prove that $\epsilon$ is minimized (for fixed $\delta$), when all $\{s_\alpha\}$ are equal. 

\begin{Lemma}\label{lemma:minimization}
Consider minimizing the function 
\begin{equation}
f_m(x_1, x_2, \dots, x_n) = \sum_{S,|S|=m} \prod_{\alpha\in S}x_\alpha
\end{equation}
subject to the constraint $\prod_{\alpha=1}^n x_\alpha= c > 0$, where all $x_\alpha$ are nonnegative.  Here $S$ denotes a subset of the $n$ variables, and $|S|$ is the size of $S$. The minimum occurs for $x_1=x_2=\dots = x_n= c^{1/n}$.
\end{Lemma}
\begin{proof}Note that $f_m$ is a symmetric function, invariant under permutations of its $n$ arguments, and can be decomposed as
\begin{align}
f_m(x_1, x_2, \dots, x_n) & = f_m(x_3, \dots x_n) + x_1 f_{m-1}(x_3, \dots x_n) \notag\\
& \quad + x_2 f_{m-1}(x_3, \dots x_n) + x_1 x_2 f_{m-2}(x_3, \dots x_n) .
\end{align}
 Using the constraint we write 
\begin{equation}
x_1 = \frac{c}{x_2 x_3 \dots x_n}, 
\end{equation}
and regard $f_m$ as a function of the $n{-}1$ independent variables $x_2, x_3, \dots, x_n$; then
\begin{equation}
\frac{\partial}{\partial x_2} \left(x_1x_2\right)= 0, \quad \frac{\partial}{\partial x_2} \left(x_1\right)= \frac{-x_1}{x_2}.
\end{equation}
Therefore, setting the gradient of $f_m$ equal to zero we find
\begin{equation}
\frac{\partial}{\partial x_2} f_m(x_1, x_2, \dots, x_n) = \left( 1 - \frac{x_1}{x_2}\right) f_{m-1}(x_3, \dots, x_n)=0.
\end{equation}
The constraint requires  that all $x_\alpha$ are positive; therefore $f_{m-1}(x_3, \dots x_n)$ is positive and we find that $x_1=x_2$. From the symmetry of $f_m$, we conclude that $x_1 = x_\alpha= c^{1/n}$ for $\alpha = 2,3, \dots, n$, when $f_m$ is stationary. This is the unique stationary point of $f_m(x_1, x_2, \dots x_n)$ when all $x_\alpha$ are positive; furthermore $f_m$ is smooth and bounded below. Therefore it must be the minimum of $f_m$. 
\end{proof}

\section{Correlated unitary noise}
\label{sec: Correlated unitary noise}

Now let's consider unitary noise acting on $n$ qubits which does not factorize into a product of single-qubit unitaries. Since we still wish to consider noise that can be corrected by the repetition code, assume that the $n$-qubit unitary $U$ has an expansion in terms of $X$-type Pauli operators:
\begin{equation}\label{eq:U-expand-psi}
U = \sum_S \psi(S) X(S),
\end{equation}
where $S$ denotes a subset of the $n$ qubits and $X(S) = \otimes_{\alpha\in S} X_\alpha $ is the $X$-type operator supported on $S$. ($X_\alpha$ means $X$ acting on the $\alpha$th qubit, and it is implicit that $I$ acts on qubit $\alpha$ for $\alpha\not\in S$.) Unitarity of $U$ implies
\begin{equation}
\sum_S |\psi(S)|^2 = 1,
\end{equation}
and
\begin{equation}
\sum_S \psi(S)^*\psi(S+S') = 0,
\end{equation}
where $S'$ is a nonempty set and $S+S'= S\cup S' \setminus S \cap S'$ is the disjoint union of $S$ and $S'$. To make the analysis of the noise more tractable, let's also suppose the noise is invariant under permutations of the $n$ qubits. In that case, $\psi(S) = \psi(|S|)$; that is, the amplitude $\psi$ depends only on the weight $w=|S|$ of the error operator $X(S)$. A tensor product of $n$ identical unitary $X$ rotations, $U= \left( cI - is X\right)^{\otimes n}$, is the special case where 
\begin{equation}
\psi(w) = c^{n}\left(\frac{-is}{c}\right)^{w},
\end{equation}
an exponential function of the weight $w$.

The symmetric unitary transformation may also expressed as $U= e^{-iH}$ where $H$ is a symmetric $n$-qubit Hamiltonian of the form
\begin{equation}
H = \sum_{w=0}^n h_w\left(\sum_{S,|S|=w} X(S)\right).
\end{equation}
We are assuming that there is no geometric locality constraint on the interactions among the qubits --- the strength of a weight-$w$ term in the Hamiltonian depends only on the weight, not on which set $S$ of $w$ qubits are interacting. Since $h_w$ is the coefficient of a sum of $\binom{n}{w}$ terms, it is implicit that $h_w$ decays as a function of $w$. It is natural to assume that $\binom{n}{w} h_w= \mathcal{O}(n)$, as only in that case do we expect (for $h_w$ sufficiently small) the probability of a logical error to drop rapidly as $n$ gets large. For example, if $h_2= \mathcal{O}(1)$, then each qubit has $\mathcal{O}(1)$ coupling strength with $n{-}1$ other qubits, so the strength of the noise acting on each qubit grows linearly in $n$, and error correction fails for $n$ sufficiently large. We will elaborate on this point in the discussion below of two-body correlated noise. In a more realistic noise model, the higher-weight terms in the Hamiltonian would have $\mathcal{O}(1)$ strength (independent of system size), but would decay sufficiently rapidly as the qubits separate that the effective single-qubit noise strength is also $\mathcal{O}(1)$ \cite{aharonov2006fault,preskill2013sufficient}.

The structure of the noise correlations is determined by how $h_w$ falls off as the weight $w$ increases. In particular, if $n^{w-1}h_w = \mathcal{O}\left(h_1^w\right)$, then $\psi(w)$ in equation (\ref{eq:U-expand-psi}) is a sum of $\mathcal{O}\left(h_1^w\right)$ terms; in that case the parameters of the logical channel will be $\epsilon = \mathcal{O}(h_1^{n{+}1})$ and $\delta = \mathcal{O}(h_1^{n})$, so the coherent and incoherent parts of the logical channel qualitatively resemble what we found for uncorrelated noise. On the other hand, in the extreme case where $h_n\ne 0$ and $h_w = 0$ for $0\le w \le n-1$, the code provides no protection against logical errors and there is no suppression of coherence. Instead we find $\delta = \mathcal{O}(h_n)$ and $\epsilon = \mathcal{O}(h_n^2)$ so that $\epsilon=\mathcal{O}(\delta^2)$ just as in equation (\ref{eq:qubit-rotation}).

To be concrete, consider the 3-qubit repetition code and noise Hamiltonian
\begin{equation}
H = h_1 \left(X_1+X_2 +X_3\right) + h_2\left(X_1X_2 + X_2X_3 + X_3X_1\right) + h_3\left(X_1X_2X_3\right).
\end{equation}The unitary noise has the expansion
\begin{align}
U= e^{-iH} &= \left( 1 + \cdots\right) I + \left(-ih_1+ \cdots \right) \left(X_1+X_2 +X_3\right) \notag\\
&\quad + \left(-ih_2 -h_1^2 + \cdots\right)\left(X_1X_2 + X_2X_3 + X_3X_1\right) \notag\\
&\quad + \left(i h_1^3 -3h_1h_2  -ih_3 + \cdots\right)X_1X_2X_3,
\end{align}
where only the leading terms are shown in the coefficient of each Pauli operator. Repeating the analysis of the logical channel as in Section \ref{subsec:rep-chi}, but now using this modified unitary noise operator, we find
\begin{align}
\mathcal{N}_{L,\rm incoh}(\rho)&= \rho +\left(3 h_1^4 + 3h_2^2 + h_3^2 + \cdots\right)\bar X\rho \bar X,\notag\\
\mathcal{N}_{L,\rm coh}(\rho)&= \left(ih_1^3 - i h_3\right)[\bar X, \rho] -3 h_1h_2 \left(\bar X\rho + \rho\bar X\right)\notag\\
&\quad - 3ih_1^3 [\bar X, \rho]  +3 h_1h_2 \left(\bar X\rho + \rho\bar X\right) +\cdots \notag\\
& = -i(2h_1^3 + h_3)[\bar X,\rho],
\end{align}
(showing only the leading terms), which yields
\begin{equation}\label{eq:3-qubit-correlated-eps-delta}
\tilde \chi_{\scriptscriptstyle XX} = \epsilon/2 = 3 h_1^4 + 3 h_2^2 +  h_3^2 + \cdots, \quad  \tilde \chi_{\scriptscriptstyle XI}= -i\delta/2 = -i\left(2 h_1^3 + h_3 + \cdots\right),
\end{equation}
where $\tilde \chi$ denotes the logical chi matrix after error correction. 
(We don't find any contribution to the coherent part of the logical channel depending only on $h_2$, because the $h_2$ term in the Hamiltonian has even $X$ parity, while the logical operator $\bar X$ has odd parity.)
Now whether coherence is suppressed hinges on the strength of the $h_3$ term in the Hamiltonian. In particular, if $h_3$ is large compared to $h_1^2$ and $h_2$, then highly correlated noise dominates, and coherence of the logical channel is unsuppressed. 

As another instructive example, consider the length-$n$ repetition code, where the Hamiltonian contains only single-qubit and two-qubit terms. We will compute the coherent and incoherent parts of the logical channel following the same reasoning as in Section \ref{subsec:rep-chi}. Again, we'll need to sum over all the possible values of the syndrome weight, which we'll now denote by $k$. For each value of $k$, we'll find a contribution to the chi matrix for the error-corrected logical channel, with logical operators acting on the encoded density operator $\rho$ from the left and from the right. Each such operator can be obtained in many ways as a product of one-body and two-body terms in the Hamiltonian, and we'll have to do some combinatorics to sum up those contributions. By computing the logical chi matrix, and comparing its coherent and incoherent parts, we can prove the following:

\begin{Theorem}
\label{Theorem: Correlated noise}
    Consider the bit flip code with $n$ qubits, and let the noise model be given by the $n$-qubit unitary map
    \begin{equation}
        U = \exp (-i H) ,\quad {\rm where}\quad \quad H = \sum_{i} h_1 X_i + \sum_{i< j}h_2 X_i X_j
    \end{equation}
    After error correction, the logical noise channel satisfies the following bound relating the coherent and incoherent components: 
    \begin{equation}\label{eq:lemma_chi}
        \tilde{\chi}_{\scriptscriptstyle XX} \geq \frac{2n}{n+1} (\tan h_1) |\tilde{\chi}_{\scriptscriptstyle XI}|,
    \end{equation}
    where $\tilde \chi$ denotes the logical chi matrix. 
    Equation (\ref{eq:lemma_chi}) holds for any odd $n$, and for any $h_1$, but we have made the approximation $n h_2 \ll 1$, neglecting a multiplicative $\left(1 + \mathcal{O}(n h_2)\right)$  correction on the right-hand side.
\end{Theorem}

Theorem~\ref{Theorem: Correlated noise} implies that, even for this correlated unitary noise model, the coherence of the logical noise channel is heavily suppressed for large $n$. In fact, the ratio of the coherent to incoherent components of the logical noise channel is similar to what we found for the  uncorrelated case, where $h_1\approx \theta/2$; compare to equation (\ref{eq:compare-epsilon-delta}). 

\begin{proof}

The proof of Theorem \ref{Theorem: Correlated noise} is contained in the next few subsections. We'll compute first the coherent component of the logical channel, then the incoherent component, and finally we'll compare the two to obtain equation (\ref{eq:lemma_chi}).

The unitary operator $U= e^{-iH}$ can be expressed as
\begin{equation}
    U = \prod_i(c_1 -i s_1X_i)\prod_{i< j }(c_2 -is_2 X_iX_j) = c_1^n c_2^{n(n-1)/2}\prod_i(1 -i t_1X_i)\prod_{i< j }(1 -it_2 X_iX_j),
\end{equation}
where $s_1 = \sin h_1$, $c_1 = \cos h_1$, $t_1 = \tan h_1$, and likewise for $h_2$. In our computations, we will suppress the prefactor $c_1^n c_2^{n(n-1)/2}$, which is implicit in all formulas, and we will expand $U$ in a \textit{collisionless approximation}. That is, we will neglect terms in the expansion in which operators such as $X_i$ and $X_iX_j$ or $X_k X_i$ and $X_i X_j$ act on a qubit in common. The terms we are neglecting are systematically suppressed by powers of $n h_2$ compared to the terms we are keeping. More precisely, these corrections can be absorbed into a multiplicative renormalization of $h_1$ and $h_2$ by a factor $\left(1 + \mathcal{O}(n h_2)\right)$.

\subsection{Coherent component} 
Let us look first at the coherent component $\tilde \chi_{\scriptscriptstyle XI}$ of the logical chi matrix. For each syndrome of weight $k$, the physical error contributing to this logical component consists of an uncorrectable $X$ error of weight $n-k$ on the left of $\rho$ and a correctable $X$ error of weight $k$ on the right, where $k$ ranges from $0$ to $(n-1)/2$. The operators on the left and right are supported on disjoint sets of qubits. When we write these operators as products of one-body and two-body terms we will need to count the number of ways of dividing a set of $2 p$ $X$ errors into distinct combinations of $p$ two body terms. We denote this number by $\kappa_p$ where
\begin{equation}\label{eq:kappa-define}
    \kappa_p = \frac{(2p)!}{2^p p!}.
\end{equation}

Let us count the terms with $k_L$ factors of $t_2$ on the left and $k_R$ factors of $t_2$ on the right. In addition, there will be some number $w$ of factors of $t_1$ on the right and $n-2k_L-2k_R-w$ factors of $t_1$ on the left to fill out the coherent term. First we choose the $2k_L$ qubits on the left where the $t_2$ terms act; these qubits can be chosen in $\binom{n}{2k_L}$ ways. Once these $2k_L$ qubits have been chosen, there are $\kappa_{k_L}$ ways to divide up the qubits into pairs where the two-body terms act. Next, we  choose the $2k_R$ qubits on the right where the $t_2$ terms act. Because the operators on the left and right are supported on disjoint sets of qubits, these $2k_R$ qubits can be chosen in $\binom{n-2k_L}{2k_R}$ ways. Once these $2k_R$ qubits have been chosen, there are $\kappa_{k_R}$ ways to divide up the qubits into pairs where the two-body terms act. Of the remaining $n-2k_L - 2k_R$ qubits where no two-body terms act, we choose $w$ qubits on the left where the one-body terms acts; these can be chosen in $\binom{n-2k_L-2k_R}{w}$ ways. 
As usual, this contribution to the logical channel has a phase, which is determined by including a factor of $-i$ for each term in the Hamiltonian which acts from the left, and a factor of $i$ for each term in the Hamiltonian which acts from the right. By combining all these factors, we find a contribution to $\tilde \chi_{\scriptscriptstyle XI}$
\begin{equation}\label{eq:fixed-w-k}
    \left(i \right)^{n-w-2k_R-k_L} (-i)^{w+k_R} t_1^{n-2k_L-2k_R} t_2^{k_L+k_R} \binom{n}{2k_L} \binom{n-2k_L}{2k_R} \binom{n-2k_L-2k_R}{w} \kappa_{k_L} \kappa_{k_R} .
\end{equation}

Next we sum over $w$, taking care to note the $w$-dependent phase in equation (\ref{eq:fixed-w-k}). Fortunately, this sum can be evaluated explicitly using an identity satisfied by binomial coefficients, just as we saw in Section \ref{sec:rep-code}. The sum ranges from $w=0$ to $w=(n-1)/2-k_R$, so we have
\begin{equation}\label{eq:sum-w-binomial}
    \sum_{w=0}^{(n-1)/2-2k_R} (-1)^w \binom{n-2k_L-2k_R}{w} = (-1)^{(n-1)/2-2k_R} \binom{n-2k_L-2k_R-1}{(n-1)/2-2k_R}
    .
\end{equation}

To complete the evaluation of $\tilde \chi_{\scriptscriptstyle XI}$, it remains to sum over $k_L$ and $k_R$ in 
\begin{equation}\label{eq:sum-chi-XI}
   \tilde  \chi_{\scriptscriptstyle XI} = \sum_{k_L, k_R} \Omega(k_L,k_R) t_2^{k_L+k_R} t_1^{n-2k_L-2k_R},
\end{equation}
where from equations (\ref{eq:fixed-w-k}) and (\ref{eq:sum-w-binomial}) we have 
\begin{equation}
    \Omega(k_L,k_R) =  \left(i \right)^{n-k_L-k_R} (-1)^{k_R} (-1)^{(n-1)/2} \binom{n}{2k_L} \binom{n-2k_L}{2k_R} \binom{n-2k_L-2k_R-1}{(n-1)/2-2k_R} \kappa_{k_L} \kappa_{k_R} .
\end{equation}
In the sum in equation (\ref{eq:sum-chi-XI}), $2k_R$ can be any nonnegative integer less than or equal to $(n-1)/2$, and $2(k_R+k_L)$ can be any nonnegative integer less than or equal to $n-1$.

Our goal is to compare this coherent component with the incoherent component, which can also be expressed as a sum. Instead of performing an unrestricted sum over $k_L$ and $k_R$, we will consider the sum over $k_L$ where $k_L+k_R=q$ is fixed. This collects all the terms in  $\tilde \chi_{\scriptscriptstyle XI}$ of order $t_2^q$. Then we will follow a similar path to compute the incoherent component $\tilde \chi_{\scriptscriptstyle XX}$ to order $t_2^q$, so that we can compare the coherent and incoherent components in each order.

Let us isolate the parts of $\Omega(k_L,q-k_L)$ that depend on $q$ only (not on $k_R$), and let us introduce the shorthand $m=(n-1)/2$, finding
\begin{equation}
\label{eq: Omega}
    \Omega(q-k_R,k_R) = 
    \frac{(i)^{n-q} (-1)^{m} (m+1) }{ (n-2q) 2^q} \binom{n}{m} \frac{(m!)^2}{(2m-2q)!q!} 
    \times
    (-1)^{k_R} \binom{2m-2q}{m-2k} \binom{q}{k}
    ,
\end{equation}
where we have used equation (\ref{eq:kappa-define}).
Now we need to sum $k_R$ from $k_R=0$ to $k_R=q$, and then sum $q$ from $q=0$ to $q=(n-1)/2$.

We observe that, due to the oscillating sign $(-1)^{k_R}$, the sum over $k_R$ vanishes when $q$ is odd. This cancellation occurs because if we replace $k_R$ by $q- k_R$, the summand remains the same except for a change in phase $(-1)^q$. What's happening is that for each term contributing to $\tilde \chi_{\scriptscriptstyle XI}$ with $l$ factors of $i t_2$ on the right and $q-l$ factors of $-i t_2$ on the left, there is a corresponding term with $q-l$ factors of $i t_2$ on the right and $l$ factors of $-i t_2$ on the left. These two terms have equal magnitude but opposite sign, if $q$ is odd.
Similar cancellations occur in the computation of the incoherent component $\tilde \chi_{\scriptscriptstyle XX}$.

\subsection{Incoherent component} 
Now we can use similar reasoning to compute the incoherent component $\tilde \chi_{\scriptscriptstyle XX}$ of the logical channel. In this case, though, we will not perform a sum over all syndromes; instead we will keep only the contribution of lowest order in $t_1$ and $t_2$, arising from the syndrome of highest weight. This will suffice for deriving the lower bound in equation (\ref{eq:lemma_chi}), because the contributions to $\tilde \chi_{\scriptscriptstyle XX}$ higher order in $t_1$ and $t_2$ are nonnegative. Furthermore, keeping only the lowest-order term is a good approximation when $t_1$ and $t_2$ are sufficiently small. 

For $n$ odd, this leading-order contribution arises from terms with $X$ acting $(n+1)/2$ times from both the left and the right. In a term with $k_L$ factors of $t_2$ on the left and $k_R$ factors of $t_2$ on the right, there will also be $(n+1)/2-2k_L$ factors of $t_1$ on the left, and $(n+1)/2-2k_R$ factors of $t_1$ on the right. Summing over $k_L$ and $k_R$, and arguing as in our discussion of the coherent contribution, we find 
\begin{equation}
    \tilde{\chi}_{\scriptscriptstyle XX} = \sum_{k_L, k_R} \Delta(k_L,k_R) t_2^{k_L+k_R} t_1^{n+1-2k_L-2k_R} + \cdots.
\end{equation}
Here 
\begin{equation}
    \Delta(k_L,k_R) = (i)^{m+1-k_L} (-i)^{m+1-k_R} \binom{n}{m} \binom{m+1}{2k_L} \binom{m+1}{2k_R} \kappa_{k_L} \kappa_{k_R},
\end{equation}
we have defined $m=(n-1)/2$, and the ellipsis indicates nonnegative higher-order corrections. 
We can again introduce $q=k_L+k_R$ and isolate the portion of $\Delta(q-k_R, k_R)$ that depends only on $q$:
\begin{equation} \label{eq: Delta}
    \Delta(q-k_R,k_R)  =  \frac{(i)^{q}}{2^q} \binom{n}{m} \frac{ ((m+1)!)^2}{(2m-2q+2)! q!} 
    \times (-1)^{k_R} \binom{2m-2q+2}{m-2k+1} \binom{q}{k};
\end{equation}
here $k_R$ is to be summed from $0$ to $q$, followed by a sum over $q$ from $0$ to $(n+1)/2$.
As for the coherent component, the sum over $k_R$ with $q$ fixed vanishes when $q$ is odd, due to the oscillating minus sign $(-1)^{k_R}$.

\subsection{Comparing the coherent and incoherent components} 
Now we are ready to compare $\tilde{\chi}_{\scriptscriptstyle XI}$ and $\tilde{\chi}_{\scriptscriptstyle XX}$. In both cases there is a sum over $k_R$ to perform for each even value of $q$, and by inspecting (\ref{eq: Omega}) and (\ref{eq: Delta}) we see that the $k_R$-dependent factors in $\Omega(q,k_R)$ and $\Delta(q, k_R)$ are nearly the same; the factor in $\Delta$ is obtained from the factor in $\Omega$ if we replace $m$ by $m+1$.
Because this factor grows rapidly with $m$, we see that the factor in $\Delta$ is larger than the factor in $\Omega$ for each value of $q$ and $k_R$, but that by itself does not suffice for comparing  $\tilde{\chi}_{\scriptscriptstyle XI}$ and $\tilde{\chi}_{\scriptscriptstyle XX}$, due to the alternating sign $(-1)^{k_R}$ in the sum over $k_R$.

To compare the coherent and incoherent logical noise components properly we must perform the sum over $k_R$. We will make use of the generalized hypergeometric function $\tensor[_3]{F}{_2}$. This function is defined 
\begin{equation}
\label{eq: 3F2 definition}
    \tensor[_3]{F}{_2} \left[ \begin{matrix} a, & b, & c\\ & d,& e\end{matrix}  \, ; \, \, z \right] = \sum_{k=0}^\infty \frac{(a)_k (b)_k (c)_k z^k}{(d)_k (e)_k k!},
\end{equation}
where $(a)_k$ denotes the Pochhammer function or the rising factorial
\begin{equation}
    (a)_k = a (a+1) (a+2) (a+3) \dots (a+k-1).
\end{equation}
If $a$ is a negative integer, then 
\begin{equation}
    \frac{(a)_k}{k!} = (-1)^k \binom{-a}{k}.
\end{equation}
and the sum over $k$ in equation (\ref{eq: 3F2 definition}) terminates --- instead of $0$ to $\infty$, the sum runs from $0$ to $-a$. The same is true if $b$ or $c$ is a negative integer. 

Using this definition of $\tensor[_3]{F}{_2}$, we can write the sum over $k_R$ of $\Omega$ or $\Delta$ in terms of $\tensor[_3]{F}{_2}$. We will have to distinguish the two cases $2q<m$ and $2q \geq m$, although we will see at the end that the final expressions will coincide for the two cases. Take the second term in equation (\ref{eq: Omega}). Supposing that $2q < m$, we can write
\begin{align}
\label{eq: Hypergeometric form of Omega}
    \sum_{k_R} (-1)^{k_R} \binom{2m-2q}{m-2k} \binom{q}{k} &= \binom{2m-2q}{m} \sum_{k_R} \binom{q}{k} \frac{m (m-1) \dots (m-2k_R+1)}{(m-2q+2k_R) \dots (m-2q+1)} \nonumber
    \\
    &= \tensor[_3]{F}{_2} \left[ \begin{matrix} -q, & \frac{1-m}{2}, & \frac{-m}{2}\\ & \frac{m+1}{2}-q, & \frac{m}{2}-q+1 \end{matrix}  \, ; \, \, 1 \right] \binom{2m-2q}{m}.
\end{align}
Then we can apply Dixon's identity for the hypergeometric function $\tensor[_3]{F}{_2}$. This reads
\begin{equation}
    \tensor[_3]{F}{_2} \left[ \begin{matrix} a, & b, & -c\\ & 1+a-b,& 1+a+c \end{matrix}  \, ; \, \, 1 \right] = \frac{\Gamma(1+ \frac{a}{2}) \Gamma(1+\frac{a}{2}-b-c) \Gamma(1+a-b) \Gamma(1+a-c)}{\Gamma(1+a) \Gamma(1+a-b-c) \Gamma(1+\frac{a}{2}-b) \Gamma(1+\frac{a}{2}-c)};
\end{equation}
\textit{c.f.} equation (2.3.3.5) in \cite{Slater66}.
Applying this formula to equation (\ref{eq: Hypergeometric form of Omega}) we get
\begin{align}
\label{eq: Apply Dixon's theorem}
   & \sum_{k_R} (-1)^{k_R} \binom{2m-2q}{m-2k} \binom{q}{k}\nonumber\\ 
   &\quad = \frac{(-q/2)!}{(-q)!} \frac{\Gamma(m-q/2-1/2) \Gamma(m/2-q-1/2) \Gamma(m/2-q)}{\Gamma(m-q-1/2) \Gamma(m/2-q/2-1/2) \Gamma(m/2-q/2)}  
    \times \binom{2m-2q}{m}.
\end{align}

We need to do something about the first factor on the right hand side $(-q/2)!/(-q)!$ because the gamma function has poles at each negative integer. However, this ratio can still be defined: 
\begin{equation}
    \frac{(-q/2)!}{(-q)!} = \lim_{q/2 \rightarrow \mathrm{Integer}} \frac{\Gamma(-q/2+1)}{\Gamma(-q+1)} = (-q)_{q/2} = (-1)^{q/2} \frac{q!}{(q/2)!}.
\end{equation}
We can substitute this into equation (\ref{eq: Apply Dixon's theorem}) and we find that we can simplify the expression 
\begin{equation}
\label{eq: Simplify Dixon's result}
    \sum_{k_R} (-1)^{k_R} \binom{2m-2q}{m-2k_R} \binom{q}{k_R} =\frac{(-1)^{q/2} (2m-q)! q!}{(m-q/2)!(q/2)! m!}.
\end{equation}
Up until now we have assumed $2q <m$. If we instead assume $2q \leq m$ we find that the intermediate steps look different, but we arrive at the same final answer as in equation (\ref{eq: Simplify Dixon's result}).

Now we can compute the sum of equation (\ref{eq: Omega}) as $k_R$ goes from $0$ to $q$ using what we found in equation (\ref{eq: Simplify Dixon's result}). We can also apply our result to perform the sum over $k$ for equation (\ref{eq: Delta}). This gives:
\begin{align}\label{eq:sum-Omega-Delta}
  \Omega(q)\equiv  \sum_{k_R} \Omega(q-k_R,k_R) & =  \frac{(2m-q)!(m+1)!}{\left(m-q/2\right)! (q/2)! (2m+1-2q)! 2^q}\binom{n}{m}, \notag
    \\
\Delta(q)\equiv    \sum_{k_R} \Delta(q-k_R,k_R) & =  \frac{(2m+2-q)!(m+1)!}{\left(m-q/2+1\right)! (q/2)! (2m+2-2q)! 2^q}\binom{n}{m}.
\end{align}    
The ratio of these quantities is
\begin{equation}\label{eq:correlated-Hamiltonian}
\frac{\Omega(q)}{\Delta(q)}
  =  \frac{(2m+2-2q)(m-q/2+1)}{(2m-q+2)(2m-q+1)}=\frac{n+1 - 2q}{2n -2q} \leq  \frac{n+1}{2n}.
\end{equation}
Now we can sum over $q$; because all terms are nonnegative and the bound holds for every $q$, we conclude
\begin{equation}
    \tilde{\chi}_{\scriptscriptstyle XX} > \frac{2n}{n+1} t_1 \tilde{\chi}_{\scriptscriptstyle XI},
\end{equation}
thus proving Theorem \ref{Theorem: Correlated noise}. 
\end{proof}

\subsection{Summary}

By setting $q=0$, we can check that the result in equation (\ref{eq:sum-Omega-Delta}) matches what we found in Section \ref{sec:rep-code} for the uncorrelated case. It is also instructive to consider the expansion of $\tilde\chi_{\scriptscriptstyle XI}$ in powers of $t_2$, under the assumption $q\ll m$. From equation (\ref{eq:sum-Omega-Delta}) we see that
\begin{equation}\label{eq:Omega-q-correction}
    \Omega(q) = \left(\frac{m^{q/2}(2m)^{2q}}{2^q (q/2)!(2m)^q}+\cdots\right)\Omega(0)
    =\left(\frac{m^{3q/2}}{ (q/2)!}+\cdots\right)\Omega(0),
\end{equation}
where the ellipsis indicates $\mathcal{O}(q/m)$ corrections.

Restoring the factors of $t_1$ and $t_2$ from equation (\ref{eq:sum-chi-XI}), we see this expansion in $t_2$ generates a multiplicative correction to $\tilde \chi_{\scriptscriptstyle XI}$ which exponentiates:
\begin{equation}
    \sum_{q \,{\rm even}}
    \frac{1}{(q/2)!}\left(\frac{m^3t_2^2}{t_1^4}+\cdots\right)^{q/2}\approx \exp\left(m^3 t_2^2/ t_1^4\right)
    .
\end{equation}
Since the sum over $q$ is dominated by terms with $m^3 t_2^2/ t_1^4 \sim q$, this exponential series should be a good approximation for $m^3 t_2^2 t_1^4 \ll m$, or $m t_2 \ll t_1^2$, since in that case neglecting the terms higher order in $q/m$ can be justified. Under this condition, the two-body terms in the Hamiltonian in equation (\ref{eq:correlated-Hamiltonian}) make a small contribution to the total energy, suppressed by $\mathcal{O}(t_1)$ compared to the one-body terms. Recall that we also needed $m t_2 \ll 1$ to justify the collisionless approximation used in the proof of Theorem \ref{Theorem: Correlated noise}; this condition is subsumed by $m t_2 \ll t_1^2$ if $t_1 = \mathcal{O}(1)$.

We see that there is a regime
\begin{equation}
    \frac{1}{m}\gg \frac{t_2}{t_1^2}\gg \frac{1}{m^{3/2}}
\end{equation}
in which our approximations are reliable, yet the multiplicative corrections to $\tilde \chi_{\scriptscriptstyle XI}$ are large. That large corrections occur, even when the two-body terms make a small contribution to the total energy, is not a surprise; we have found as expected that the noise correlations can substantially enhance the probability of a logical error. The important point established by Theorem~\ref{Theorem: Correlated noise} (at least for the simple noise model we have analyzed) is that even when the correlated noise produces large corrections to the logical channel, the corrections occur in both the coherent part and the incoherent part of the channel, so that our conclusion that the coherence is strongly suppressed for large $n$ continues to apply. 

It is not immediately obvious why the leading power of $m$ in equation (\ref{eq:Omega-q-correction}) should be $m^{3q/2}$, because higher powers of $m$ occur in $\Omega(q-k_R, k_R)$ and $\Delta(q-k_R, k_R)$ for each fixed $k_R$ and $q$. It turns out that these higher powers of $m$ all cancel when we do the sum over $k_R$. In Appendix \ref{app: Correlated noise} we explain why these cancellations occur, providing a useful check on our results. 

\section{The toric code against coherent noise}\label{sec:toric-code}
We now analyze the logical channel for the two-dimensional toric code on an $L\times L$ square lattice, where $L$ is odd. We'll consider uncorrelated unitary noise acting on the $2L^2$ qubits, and suppose that error correction is performed using minimal-weight decoding. Our goal is to show that, when the noise is sufficiently weak, the coherence of the logical noise channel is highly suppressed for large $L$.

Our analysis will draw heavily on the tools we developed in our study of the repetition code. Before proceeding further, we will review some notation. 

\subsection{The Toric Code}

We will consider the 2D toric code, which is defined on a square lattice with qubits placed on edges. We choose a square patch of lattice with side length $L$ and identify opposite edges. (The toric code can be constructed on a lattice with boundaries, but for simplicity we choose periodic boundary conditions.) The stabilizer group for the toric code is generated by the $X$ and $Z$ generators shown in figure \ref{fig: Toric code stabilizers}. The logical operators of the toric code are topologically non-trivial loops that wrap around the torus. Figure \ref{fig: Toric code logical operators} shows two logical operators.

\begin{figure}
    \centering
    \includegraphics[width=7cm]{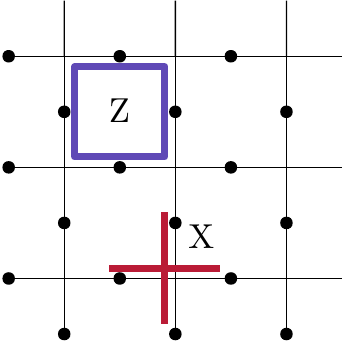}
    \caption{In blue is a $Z$-type stabilizer generator for the toric code. There are $Z$ generators on every plaquette in the lattice. In red is an $X$-type stabilizer generator. There are $X$ generators at every vertex of the lattice.}
    \label{fig: Toric code stabilizers}
\end{figure}

\begin{figure}
    \centering
    \includegraphics[width=7cm]{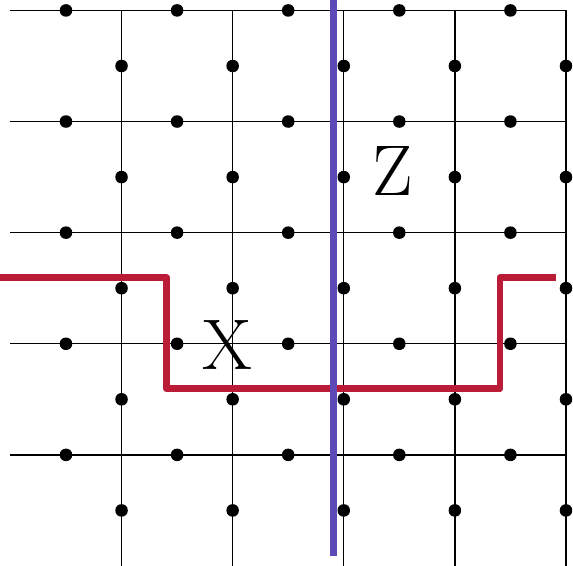}
    \caption{This figure illustrates two of the logical Pauli operators in the toric code. $X$ and $Z$ logical operators are shown for one of the two encoded qubits. Each logical operator is a topologically non-trivial loop that wraps around the torus. The logical operators for the other encoded qubit are the similar, but rotated by 90 degrees.}
    \label{fig: Toric code logical operators}
\end{figure}

The toric code is parameterized by the linear dimensions of the lattice; when the side length is $L$, the code distance (the minimum weight of a nontrivial logical operator) is $L$. We will also sometimes refer to $L$ as the code ``size.'' The number of physical qubits in the code block is $2L^2$, and there are two encoded logical qubits. To analyze the logical channel, we must choose a decoding procedure. Decoding the toric code is a well-studied problem and many good algorithms are known \cite{dennis2002topological, Duclos-Cianci2010, Bravyi2014}. We will choose minimal-weight decoding, in which the applied recovery operation has the lowest possible weight consistent with the measured error syndrome. This recovery operation can be computed efficiently on a classical computer \cite{edmonds1965paths}, and corrects the error with a success probability that is exponentially close to 1 when $L$ is large and the noise is both sufficiently weak and sufficiently weakly correlated.

\subsection{Notation}

We will use the chi matrix to describe the physical noise channel $\mathcal{N}$ acting on the $2L^2$ qubits in the code block:
\begin{equation}
    \label{eq: Chi matrix definition}
    \mathcal{N}( \rho) = \sum_{i,j} \chi_{ij} \sigma^i \rho \, \sigma^j,
\end{equation}
where $\{\sigma^i\}$ is a basis of Pauli operators.
\begin{Definition}
    When we speak of a ``noise term'' we will mean a component of the chi matrix for the physical noise channel acting on the qubits in the code block. We will find it convenient to use the notation $(\sigma^i \rho \, \sigma^j)$ for the number $\chi_{ij}$, the coefficient of $\sigma^i \rho \, \sigma^j$ in the chi-matrix expansion in equation (\ref{eq: Chi matrix definition}).
\end{Definition}

We may choose the index that labels a Pauli operator to be $(s,a,x)$, where $\sigma({s,a,x})= E_s L_a G_x$; here $s$ denotes the error syndrome, $E_s$ is the standard error associated with the syndrome $s$, $L_a$ is a standard choice for the physical Pauli operator that acts  as the logical Pauli operator $\tilde L_a$, and $G_x$ is an element of the code stabilizer. To compute the logical chi matrix, we sum over the syndrome $s$ and the stabilizer elements, observing that the standard error $E_s$ is removed by the recovery procedure. Hence we find that a term in the logical chi matrix can be expressed in our notation as
\begin{equation}
    \label{eq: General logical component sum}
    \tilde{\chi}_{\scriptscriptstyle{a b}} \equiv (\tilde{L}_a \tilde{\rho} \tilde{L}_b^\dagger ) = \sum_{s,x,y} (E_s L_a G_x \rho \, G_y^\dagger L_b^\dagger E_s^\dagger).
\end{equation}
We say that the diagonal components of the logical chi matrix $\tilde \chi_{ab}$ with $a=b$ are ``incoherent'' noise terms. and that the off-diagonal terms with $a\ne b$ are ``coherent.''

\subsection{Coherent and Incoherent Logical Components}

We are going to analyze the coherent and incoherent sums separately at first. Using path counting, and assuming the noise is sufficiently weak, we will prove that in both cases the logical chi matrix is dominated by ``short logical strings'' (logical Pauli operators of relatively low weight), those with  length $\leq L+2\zeta$ for a constant $\zeta$. Then, by summing up the contributions due to these short logical strings, we will derive an inequality relating the coherent and incoherent components of the logical channel. 

Our argument will use equation (\ref{eq: General logical component sum}), where we have expressed the the logical chi matrix as a sum of terms in the physical chi matrix. In the next several sections we will analyze the sums contributing to coherent and incoherent components of $\tilde\chi_{ab}$. We will make a series of approximations to simplify the sums by neglecting certain terms. In the end we will demonstrate that the two sums are related by a constant factor. 

\subsection{The Coherent Sum}
\label{section: The Coherent Sum}

First, consider the coherent sum. The coherent components of the logical noise channel are sums of terms from the physical noise channel. We want to upper bound the magnitude of these coherent logical components. Before we go any further, we will make some simplifications. For one, we will neglect certain coherent logical noise components. We focus on the components of the logical noise $\tilde{\chi}_{\scriptscriptstyle{a b}}$, where exactly one of the operators $L_a$ and $L_b$ is identity and the other is either an $X$ or a $Z$ error on one of the two encoded qubits. These components of the logical noise channel can be expressed as a sum over physical noise terms:
\begin{equation}
    \tilde{\chi}_{\scriptscriptstyle{a I}} = \sum_{s,x,y} (E_s L_a G_x \rho \, G_y^\dagger E_s^\dagger),
    \label{eq: Coherent Term}
\end{equation}
where $L_a$ is either an $X$ or $Z$ logical error on one of the two encoded qubits of the toric code. In Appendix \ref{app: a and b nontrivial} we prove that we can neglect the coherent terms with non-trivial logical operators on both sides of $\rho$, and in Appendix \ref{app: X1X2 or Y1 terms} we prove that we can neglect $Y$ logical operators and operators that act non-trivially on both encoded qubits. The proof comes down to showing that terms with a non-trivial error on both sides of $\rho$, that act on both encoded qubits, or that apply a $Y$ to one of the logical qubits, have high weight relative the terms we keep. A further simplification concerns the structure of the noise model. Our result applies to a noise model in which the single-qubit unitary operator acting on each qubit has an axis of rotation and angle of rotation that varies somewhat from qubit to qubit. However, we will prove that the most coherent logical channel is one in which the same unitary operator is applied to each qubit, so we may confine our attention to that case for the purpose of deriving a bound on the relative strength of the coherent and incoherent parts of the logical channel. 

We will make use of another way of writing the coherent sum. Each coherent term in the form of equation (\ref{eq: Coherent Term}) can be unambiguously associated with a logical string. The product of the Pauli operators acting on the left-hand and right-hand sides is the logical operator $L_a G_x G_y$, which in general consists of a connected logical string wrapping around the code block, accompanied by some number of closed loops. To be concrete, if $L_a$ is a logical $X$ error, then the logical string contains only physical $X$ errors, the closed loops are either loops of $X$ errors which are disjoint from the logical string, or closed loops of $Z$ errors which may or may not intersect with the logical string or with the closed loops of $X$ errors (the intersections are the $Y$ errors).

\begin{Definition}
\label{Def: Connected coherent}
    For a given noise term $(E_s L_a G_x \rho \, G_y^\dagger E_s^\dagger)$, we can extract a connected logical string by removing the topologically trivial loops from $L_a G_x G_y$. Call this logical string $\mathcal{L}$. We define the ``connected part" of the noise term as the restriction to the qubits in $\mathcal{L}$. The connected part of $(E_s L_a G_x \rho \, G_y^\dagger E_s^\dagger)$ is a noise term given by
    \begin{equation}
        ((E_s L_a G_x) |_{\mathcal{L}} \, \rho  \, (G_y^\dagger E_s^\dagger)|_{\mathcal{L}}) ,
    \end{equation}
    where the symbol $|_{\mathcal{L}}$ denotes the restriction of an operator to the support of $\mathcal{L}$.
\end{Definition}

\begin{Definition}
\label{Def: Disconnected Coherent}
    For a noise term $(E_s L_a G_x \rho \, G_y^\dagger E_s^\dagger)$ the ``disconnected part" is the part of the noise term not in the connected part. Once again, we can define a continuous logical string $\mathcal{L}$ by removing all topologically trivial closed loops from $L_a G_x G_y$. The disconnected part of $(E_s L_a G_x \rho \, G_y^\dagger E_s^\dagger)$ is given by 
    \begin{equation}
        ((E_s L_a G_x) |_{\mathcal{L}^C} \, \rho  \, (G_y^\dagger E_s^\dagger)|_{\mathcal{L}^C}) ,
    \end{equation}
    where the symbol $|_{\mathcal{L}^C}$ denotes the restriction of an operator to the qubits in the complement of the support of $\mathcal{L}$.
\end{Definition}

Furthermore, we will be able to assume that all of the physical single-qubit errors in the connected part are $X$- or $Z$-type. For example, in the case of a logical $X$-type error, we may neglect terms in which a closed loop of $Z$ errors intersects with the logical string. To justify this assumption, we show in Appendix \ref{app: Physical Y Errors} that allowing $Y$ errors along the logical string will only make the logical noise channel less coherent. 

A coherent term contributing to the logical chi matrix element $\tilde \chi_{Z_1 I}$, which includes disconnected errors, is illustrated in figure \ref{fig: simple added error}. The disconnected part includes identity on the qubits without errors in addition to the the disconnected errors. 
\begin{figure}
    \centering
    \includegraphics[height = 7cm]{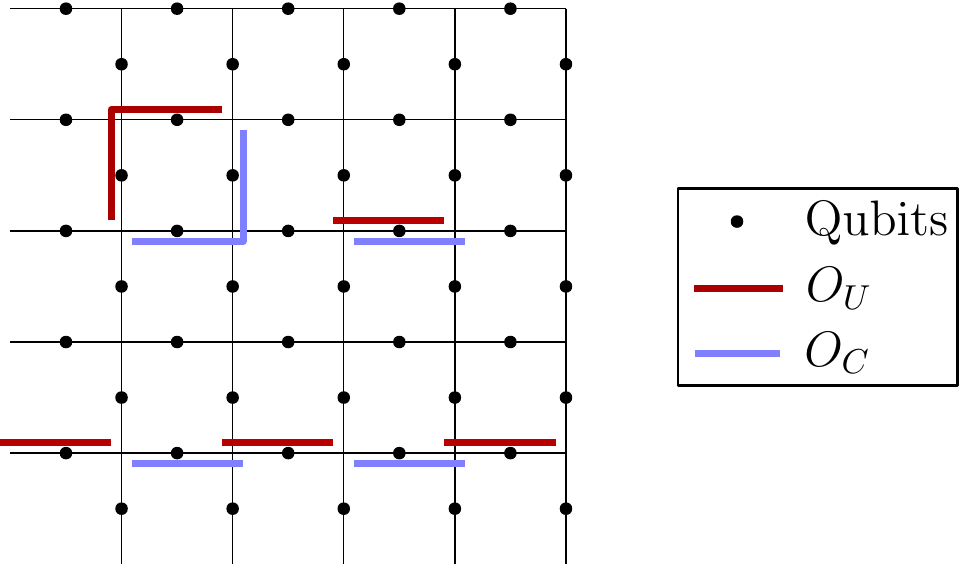}
    \caption{In this coherent term, the uncorrectable error $O_U$ (acting on the density operator from the left) is in red, while the correctable error $O_C$ (acting from the right) is in blue. Only $Z$ errors are shown. The connected logical string consists of the five qubits near the bottom that are split between red and blue. In addition, there are disconnected errors in the form of the closed loop containing two red edges and two blue edges and the pair of cancelling single-qubit errors acting on the left (in red) and right (in blue).
    }
    \label{fig: simple added error}
\end{figure}
$Z$ errors acting on the density operator from the left are shown in red, and $Z$ errors acting from the right are shown in blue. Because the errors acting from the left and right have the same syndrome $s$, the product of the left and right logical operators is logical. The connected logical string crosses the code block near the bottom of the figure. Associated with the syndrome $s$ is the corresponding standard error $E_s$, the $Z$ error of minimal weight with that syndrome. (If the minimal-weight error is not unique, we arbitrarily choose $E_s$ to be one of the errors of minimal weight by convention.) To evaluate the logical chi matrix element  $\tilde \chi_{Z_1 I}$ as in equation (\ref{eq: General logical component sum}), we need to sum over the syndrome $s$ and the stabilizer elements $G_x$ and $G_y$. To facilitate estimating the sum, it will be helpful to organize it in an appropriate way. 

To this end, we introduce the following definition:
\begin{Definition} 
\label{Def: Partition}
    For a logical string $\mathcal{L}$ with no topologically trivial closed loops, the word ``partition" denotes a connected noise term $(O_1 \rho \, O_2)$ such that $O_1 O_2 = \mathcal{L}$ and $O_1$ and $O_2$ are disjoint. In other words, each partition is a way of dividing the single-qubit errors in $\mathcal{L}$ into two subsets, $O_1$ and $O_2$. By definition $O_1$ and $O_2$ share the same syndrome. Because the code size $L$ is odd, exactly one of $O_1$ and $O_2$ will be corrected to a logical operator with the same action as $\mathcal{L}$ and the other will be corrected to the logical identity.
\end{Definition}
For each fixed logical string, the sum over all partitions of the logical string will produce the full set of connected terms derived from that logical string. The sum over partitions, for a fixed logical string, is directly analogous to the sum over syndromes we encountered in our analysis of the repetition code in Section \ref{subsec:rep-chi}. In the case of the toric code, we compute the coherent part $\tilde \chi_{\scriptscriptstyle Z_1 I}$ of the logical channel by summing over all possible logical strings, and for each choice of logical string we sum over all partitions of the logical string. In addition, for each chosen logical string, we sum over the possible disconnected pieces, the additional closed loops of $Z$ errors which are disjoint from the logical string. 

Schematically, the coherent component of the logical chi matrix is
\begin{equation}
    \label{eq: coherent sum string form}
 \tilde\chi_{\scriptscriptstyle Z_1 I} = \sum_{\mathrm{strings}} \sum_{\mathrm{partitions}} \left( \text{Connected Part} \right) \left( \text{Disconnected  Sum} \right).
\end{equation}
This form will allow us to approximate the coherent sum. Assuming that the noise is sufficiently weak, we will prove that we can truncate the sum over logical strings, including only short strings.  Furthermore, most of the short logical strings have a particular shape. To complete the argument, we will show that the disconnected sum is approximately the same for each short logical string and for each partition of the logical string.

\subsection{Counting of Logical Strings}
\label{sec: Path Counting}

We want to find an upper bound on the magnitude of the coherent component of the logical noise channel. We have already put the sum over physical noise terms into a convenient form by factoring out the disconnected piece of each term. Next, we will simplify the sum by restricting the set of connected pieces we need to consider; we will neglect the long logical strings in favor of those strings with length no larger than than $L+2\zeta$, where $\zeta$ is an $L$-independent constant. To justify this truncation we will require a strong assumption on how the physical noise strength scales with $L$; namely, the single-qubit rotation angles must scale as $1/L$.

In equation (\ref{eq: coherent sum string form}), we wrote the contribution of a given logical string to the coherent logical noise as a product of a connected and disconnected part as described in Definitions \ref{Def: Connected coherent} and \ref{Def: Disconnected Coherent}. The connected part summed over partitions as defined in Definition \ref{Def: Partition}. The sum over partitions contains $2^{w-1}$ terms for a weight-$w$ logical string (one containing $w$ lattice edges). Suppose that the unitary rotation $U^Z(\theta) = \exp\left(-i \frac{\theta}{2}Z\right)$ is applied to each physical qubit in the toric code block. We can upper bound the sum using the number of terms times the magnitude of each term. Then, the contribution of each logical string is upper bounded by $2^{w-1} ( |\sin (\theta/2)| \cos(\theta/2))^w$ times the factor from the disconnected part. We will prove in Section \ref{sec: disconnected pieces} and Appendix \ref{app: Disconnected errors} that the disconnected piece is $1$ plus a higher weight correction that we can neglect for short logical strings.

There is a regime where we can upper bound the number of logical strings as a function of the string's length. Asymptotically, the number $c_w$ of self-avoiding random walks with length $w$ was proven in \cite{Hammersley1960} to satisfy
\begin{equation}
\label{eq: path counting formula}
    c_w = \mu^{w + o(w)} ,
\end{equation}
where $\mu \approx 2.64$ for the 2D square lattice. We can start a walk from a fixed point along one edge of the toric code. Logical strings will be the self-avoiding walks that wrap around the torus and end at the starting point. We can use equation (\ref{eq: path counting formula}) to show that the contribution to the coherent logical noise from logical strings of length $\ell$ is exponentially decaying with $\ell$ as long as $|\theta| < \arcsin{1/\mu} \approx 0.39$. This statement apples only for logical strings with length much greater than the minimum of $L$, the code distance. We do not have a precise estimate indicating at what length above $L$ the number of logical strings begins to scale like equation (\ref{eq: path counting formula}). This means we do not know at what string length $\ell$ the contribution will begin to decay exponentially, and therefore we do not know where to truncate the sum if we wish to use equation (\ref{eq: path counting formula}) to bound the terms we are neglecting. In any case, in our subsequent analysis we will truncate the sum over the string length $\ell$ at $L+2\zeta$ for some constant $\zeta$. In this regime the asymptotic estimate in equation (\ref{eq: path counting formula}) is not helpful and we will not make use of it. Instead, we will assume that $|\theta|$ is sufficiently small that we can use the following lemma to bound the terms we neglect.

\begin{Lemma}
    Suppose that $|\sin \theta| < 1 / L$. In equation (\ref{eq: coherent sum string form}), we wrote $\tilde\chi_{\scriptscriptstyle{Z_1I}}$ as a sum over logical strings. If we truncate the sum to include only logical strings of length $w \leq L+2\zeta$, then magnitude of the difference between the truncated sum and the complete sum is
    \begin{equation}
        \leq \alpha L^{2\zeta+1} |\sin \theta|^{L+2\zeta} ,
    \end{equation}
    where $\alpha = (1-L |\sin \theta|)^{-1}$.
    \label{Lemma: Coherent path counting}
\end{Lemma}

\begin{proof}

We begin by fixing a point along one edge of the code block, which can be chosen in $L$ ways. We will count the number of logical strings that wrap around the torus and pass through that fixed point on the edge. Let $\ell$ denote the logical string length. At minimum length $\ell = L$, there is only one logical string. At length $\ell = L+2$, if the logical string runs left to right across the code, then the string features one step up and one step down. There are $L(L-1)$ such logical strings. At longer string lengths, there are many steps up and down.
We can upper bound the number of such logical strings by supposing we choose any of the $L$ positions to place each of the steps up and steps down. We divide by $((\ell-L)/2)!$ to capture the fact that the $(\ell-L)/2$ steps up are all identical, and the same for the steps down. This encompasses all possible combinations of steps up and down including cases where the several steps up are placed at the same point creating a step up of more than one. It does not encompass strings that backtrack, but in Lemma \ref{Lemma: shape of strings A} we show that among strings of length $L+2 \zeta$, those that feature backtracking are suppressed by $\mathcal{O}(1/L)$. Also, 
The number of strings grows most quickly near minimum and eventually approaches the asymptotic value, where the number of strings grows like $\mu^L$. In the asymptotic regime, the number of strings grows much slower than $L^{\ell-L}$. We conclude that
\begin{equation}
\label{eq: Logical strings with length l}
    \text{number of logical strings of length $\ell$} \leq L L^{\ell-L}.
\end{equation}

In equation (\ref{eq: coherent sum string form}), for each logical string in the sum the contribution to the logical noise is a sum over partitions of the connected part times a disconnected part. 
We will discuss the sum over partitions in detail in Section \ref{sec: sum over partitions}, but for now it is enough that we know that the sum over partitions contains $2^{\ell-1}$ terms for each connected logical string of length $\ell$. These terms have different phases and in general the sum can be complicated. We can obtain a simple bound by multiplying the number of terms by the magnitude of each term, in other words treating all the phases as if they are the same. For each weight $w$ string,
\begin{equation}
    \sum_{\mathrm{partitions}} (\text{Connected Part}) \leq 2^\ell (|\sin \theta/2| \cos \theta /2)^\ell.
\end{equation}

We still have to handle the disconnected piece. In Section \ref{sec: disconnected pieces} we will argue that the disconnected sum decreases as the length of the logical string increases. Furthermore, the disconnected part equals $1$ up to corrections which are small for logical strings with length $\leq L+2 \zeta$ for a constant $\zeta$. This means that we can upper bound the coherent logical noise component $\tilde\chi_{\scriptscriptstyle{Z_1 I}}$ by
\begin{equation}
    \label{eq: path counting}
    \sum_{\ell=L}^{\ell_{\rm max} } L^{\ell-L+1} |\sin \theta|^\ell ,
\end{equation}
where $\ell_{\rm max}$ is the longest $Z_1$ logical string supported on the code. If $|\sin \theta |< 1/L$, the contribution from logical strings of length $\ell$ decreases exponentially with $\ell$.

If we truncate the sum over logical strings to those with weight $w\le L+2 \zeta$, the error we make is equal to the total contribution of strings with weight $w>L+2 \zeta$. The contribution at weight $w$ is exponentially decreasing with $w$, so we can bound the sum over the long logical strings using
\begin{equation}
    \label{eq: sum over tail}
    \sum_{\ell=c}^{\infty} \beta^\ell = \frac{\beta^c}{1-\beta} = \alpha \beta^c \quad 0 < \beta < 1 ,
\end{equation}
where $\alpha = \frac{1}{1-\beta}$.
We conclude that the absolute error we make by truncating the series is
\begin{equation}
\label{eq: Path counting coherent tail bound}
    \leq \alpha L^{2\zeta+1} |\sin \theta|^{L+2\zeta}
\end{equation}
where $\alpha = (1-L |\sin \theta|)^{-1}$.
Therefore, the error due to truncation is exponentially small in both $L$ and $\zeta$.

\end{proof}

In Lemma \ref{Lemma: Coherent path counting}, we proved an upper bound on the absolute magnitude of the error due to truncation in the coherent sum. However, so far we have not described any lower bound on the terms that we have kept, arising from the logical strings with length $\le L + 2\zeta$.
Therefore, we have not yet justified that the error we have neglected is small relative to the coherent noise contributions that we kept. However, we will prove in Section \ref{sec: Mismatched weights} that the \emph{incoherent} logical noise component is at least $L \binom{L}{\frac{L+1}{2}} (\frac{\sin \theta}{2} )^{L+1}$; compared to this incoherent component the contribution in equation (\ref{eq: Path counting coherent tail bound}) to the coherent component due to strings of length $> L + 2\zeta$ is suppressed by a factor $(L\sin\theta)^{2\zeta}$. This means that the error we make in truncating the sum in Lemma \ref{Lemma: Coherent path counting} is negligible compared to the incoherent component, an observation which will be helpful for showing that the coherence of the logical channel is suppressed.  For now, we will restrict our attention to connected logical strings with length $\leq L+2\zeta$ for a constant $\zeta$. We will refer to these as ``short logical strings."
\begin{Definition}
    A ``short logical string" is a nontrivial logical Pauli operator with no topologically trivial closed loops and length $\leq L+2\zeta$, where $L$ is the code size and $\zeta$ is our chosen cutoff constant.
\end{Definition}

\subsection{Sum Over Partitions}
\label{sec: sum over partitions}

In the previous section we restricted our attention to short logical strings, which have length $\leq L+2 \zeta$ where $L$ is the code size and $\zeta$ is a constant. We can go further by characterizing the shape of a logical string, and arguing that logical strings with shape meeting certain criteria give a dominant contribution to the logical channel. 

\begin{Definition}
\label{Def: typical shape}
    Among short logical strings, we will speak of those with ``typical shape." This means two things. First, supposing that the logical string in question runs left to right across the code block, then the steps up and down along the string are by one lattice spacing at a time. Furthermore, the string contains no backtracking steps that moving from right to left. Second, the individual steps up and down are separated from each other by at least $\gamma \sqrt{L}$, where $\gamma$ is a small constant we may choose. This constant $\gamma$ will appear in the error term in many of our subsequent estimates.
\end{Definition}
In Lemmas \ref{Lemma: shape of strings A} and \ref{Lemma: shape of strings B} in Appendix \ref{app: Shape of string} we prove that most short strings have a typical shape. Among short strings with length $\leq L+2 \zeta$, the fraction of atypical strings relative to the total number of logical strings of the same length is 
\begin{equation}\label{eq:atypical-fraction}
    \frac{\textrm{Atypical strings}}{\textrm{Total strings}} = \frac{8 \gamma \zeta^2}{\sqrt{L}} +\mathcal{O}\left( \frac{1}{L} \right)
\end{equation}
Figure \ref{fig: weight 15 string A} illustrates a string with typical shape for some small $\gamma$. Short logical strings with typical shape are simple, which makes our analysis easier, particularly when we discuss the sum over partitions.

Let's revisit the sum over partitions for a fixed connected logical string. That is, for a given logical string contributing to $\tilde\chi_{\scriptscriptstyle{Z_1I}}$, we wish to enumerate all the ways to divide the $Z$ errors along the string into an uncorrectable error acting on the density operator from the left and a correctable error acting from the right. This sum over partitions of a fixed logical string is analogous to the sum we encountered when we summed over syndromes in our analysis of the repetition code. In the case of the repetition code of length $n$, there is just one length-$n$ ``logical string'' to consider, and summing over syndromes is equivalent to summing over all ways of choosing a (correctable) error acting on the right that has weight at most $(n-1)/2$ (where $n$ is odd). 

In the toric code, although the sum over partitions is similar to the sum over syndromes in the repetition code, there is a complication. 
\begin{Definition}
    \label{Def: Exceptional term}
    An ``exceptional term" is a partition of a connected logical string $\mathcal{L}$ such that the uncorrectable error has lower weight than the correctable error.
\end{Definition}
In some cases, depending on the geometry of the logical string, we will have some number of exceptional terms. These exceptional terms complicate our analysis of the logical channel. Fortunately, because we need only consider contributions to the the logical channel arising from short logical strings when the noise is weak enough, we will be able to fully characterize the exceptional terms and show they are negligible.

How exceptional terms can occur is illustrated in figure \ref{fig: weight 7 uncorrectable A}. Here, for the toric code with $L=9$, we consider the logical string of length 15 shown in figure \ref{fig: weight 15 string A}, and we have chosen a partition such that the uncorrectable error shown in red has weight 7, while the correctable error shown in blue has weight 8. Note that the minimal-weight standard error associated with the error syndrome on the logical string has weight 6 --- it follows nearly the same path as the correctable error, but achieves a lower weight than the correctable error by taking a ``shortcut'' across the blue notch on the logical string. Another example of an exceptional term for this same logical string is shown in figure \ref{fig: weight 6 uncorrectable A}, where this time the weight of the uncorrectable error is 6, and the minimal-weight error has weight 5. Again, the minimal-weight error takes a shortcut, avoiding the excursions up and down followed by the correctable error.

For all these examples, the correctable error contains the qubits along the logical string that make the furthest excursions up and down. This turns out to be a universal rule, at least among the typical short logical strings ---  for exceptional terms, the uncorrectable error has no support on the outermost steps along the string. 
In the next lemma we count the number of exceptional terms and find that relative to the total number of partitions of a typical short logical string, these exceptional terms are exponentially unlikely in $L$.

\begin{figure}
    \centering
    \includegraphics[height = 8cm]{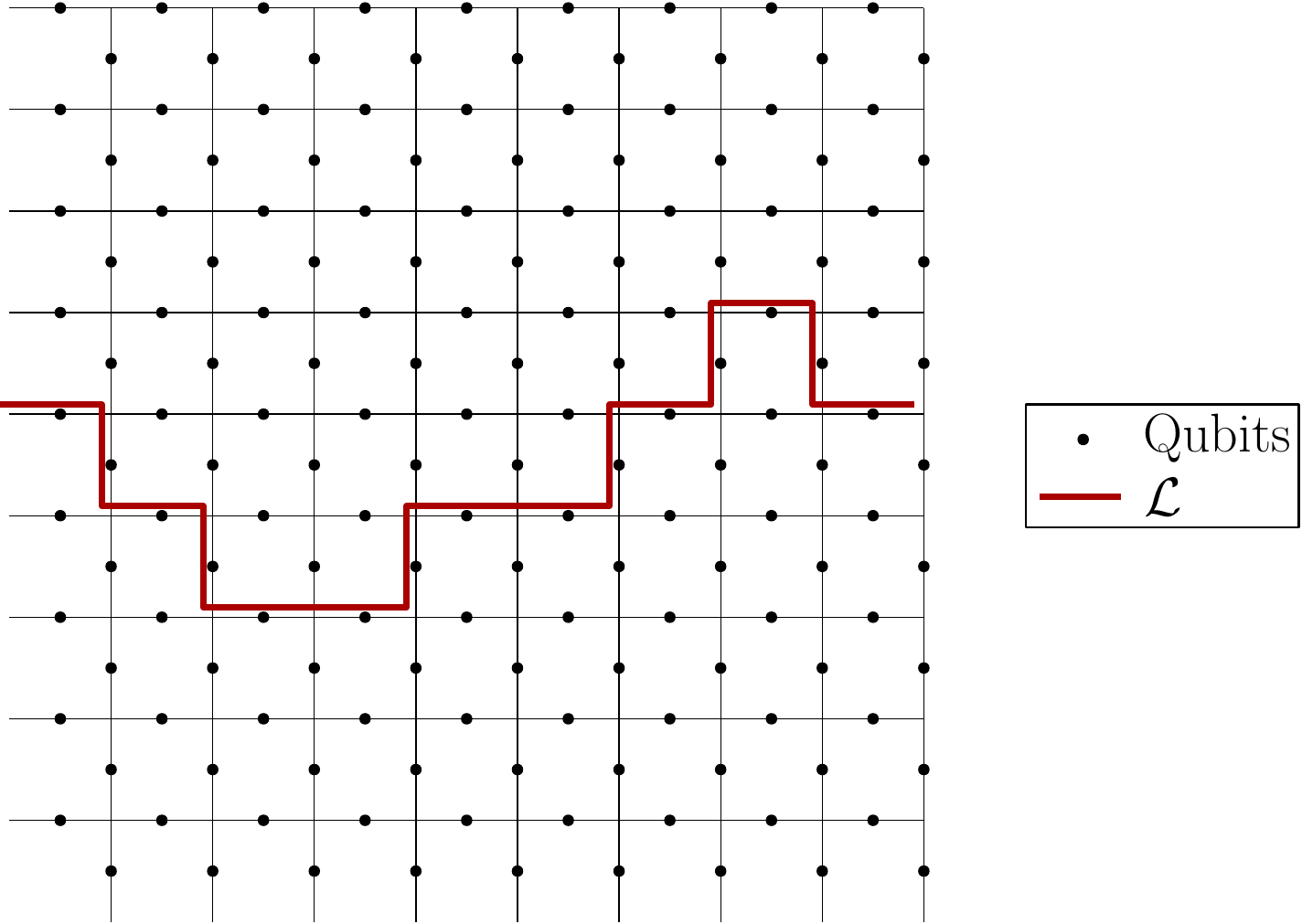}
    \caption{Here we have one logical string $\mathcal{L}$ of length 15 in an $L = 9$ toric code. Imagine the code size growing while $\zeta$ remains fixed. The likely strings will be those where the steps up and down are widely separated.}
    \label{fig: weight 15 string A}
\end{figure}

\begin{figure}
    \centering
    \includegraphics[height = 8cm]{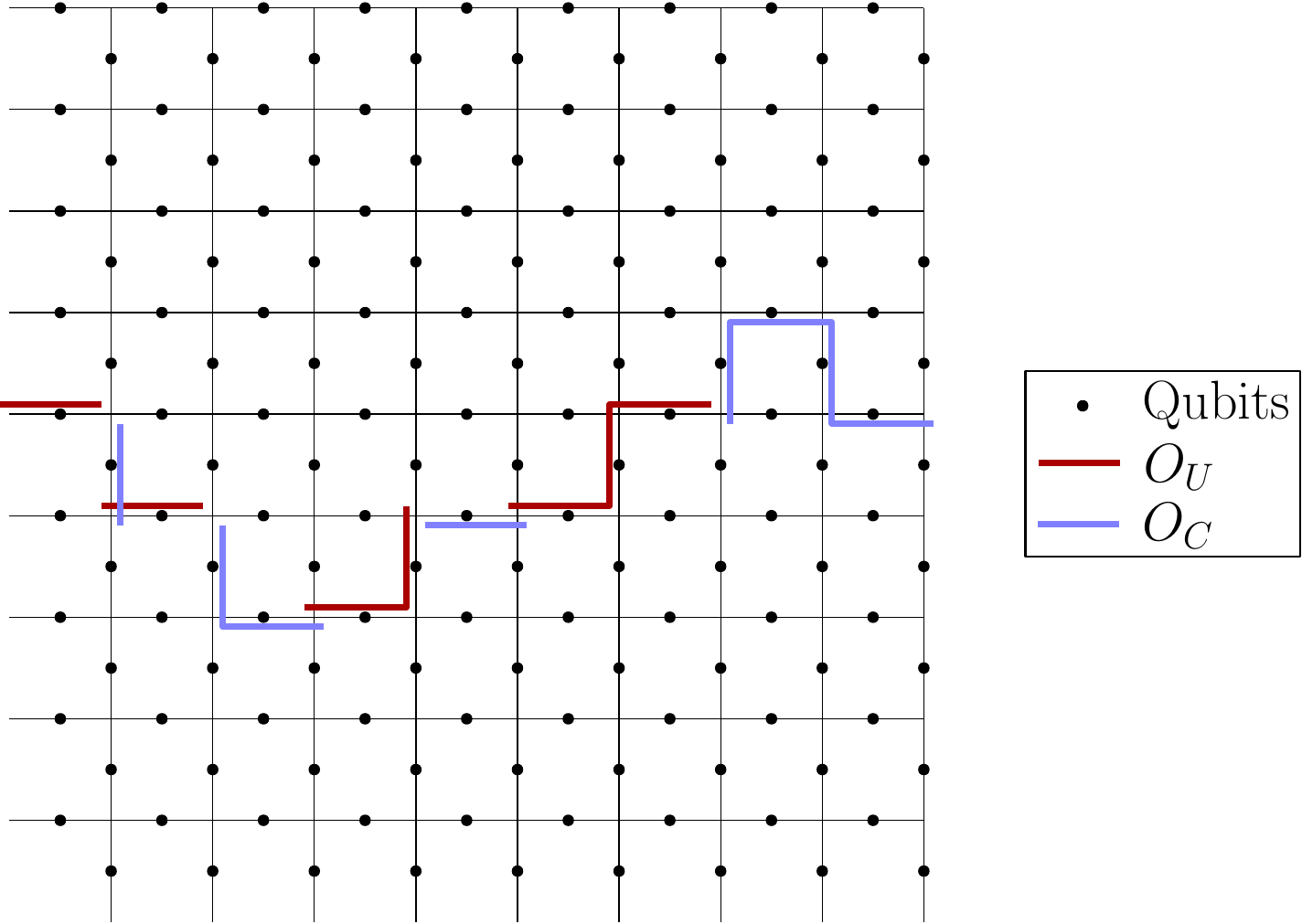}
    \caption{Now we choose a subset of 7 of the errors in the logical string in figure \ref{fig: weight 15 string A}. The uncorrectable error $O_U$ is in red and the correctable error $O_C$ is in blue. All three errors along the ``cap" in the top right appear on the correctable side. For this reason, the correctable error has weight 8, which is higher than the uncorrectable error with weight 7. We call this a weight-7 exceptional term.}
    \label{fig: weight 7 uncorrectable A}
\end{figure}

\begin{figure}
    \centering
    \includegraphics[height = 8cm]{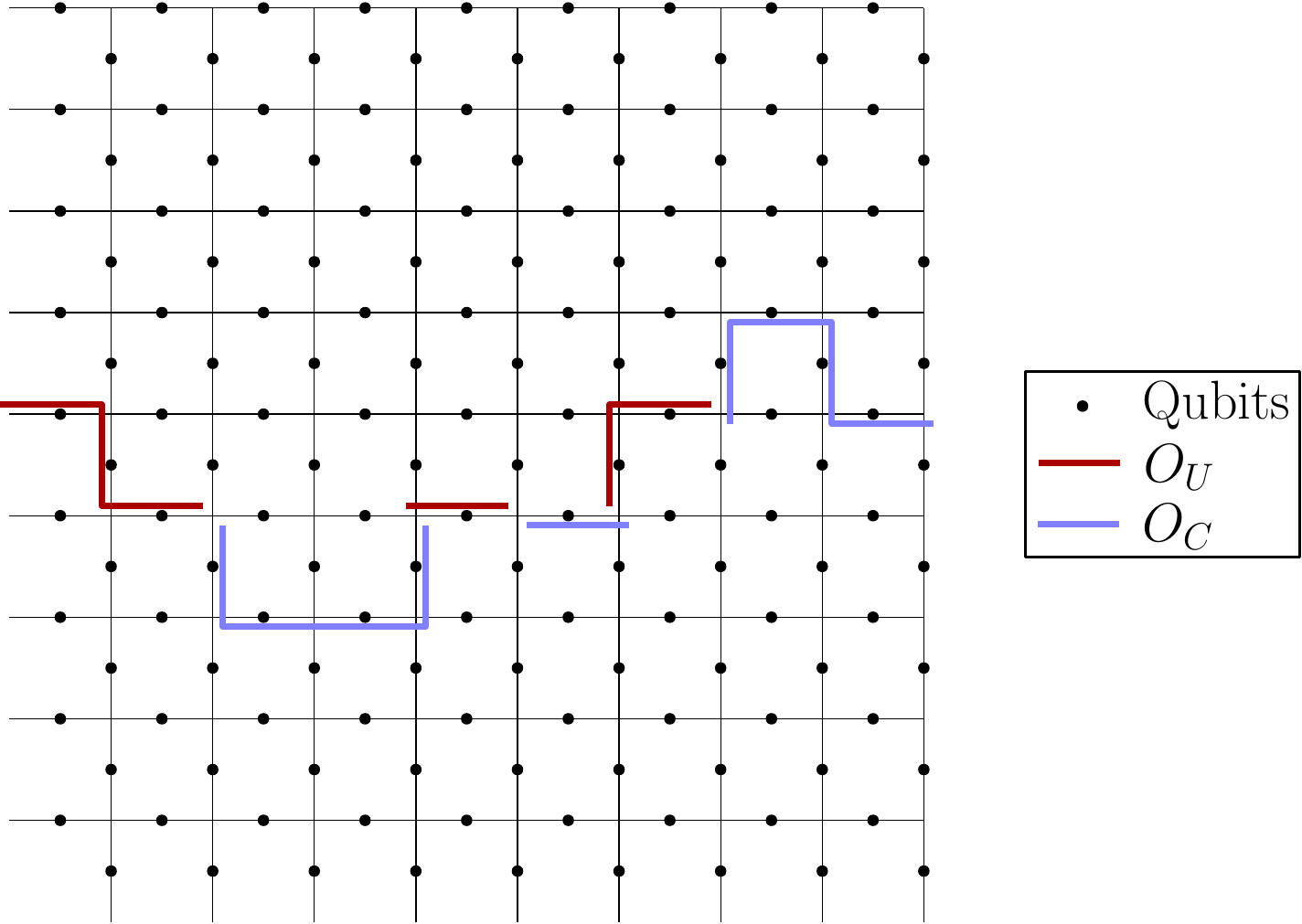}
    \caption{Again, one possible partition of the logical string in figure \ref{fig: weight 15 string A} is illustrated. The uncorrectable error $O_U$ is in red, and the correctable error $O_C$ is in blue. The correctable error includes all the errors along both the cap in the top right and the bottommost cap of the logical string. For this reason, the correctable error has weight 9, while the uncorrectable error has weight 6. Therefore, we call this partition a weight-6 exceptional term.}
    \label{fig: weight 6 uncorrectable A}
\end{figure}

\begin{Lemma}
Fix a logical string of length $\ell \leq L+2\zeta$, where $\zeta$ is a specified $L$-independent constant, with a typical shape according to Definition \ref{Def: Exceptional term}. This means that if the string runs left to right across the code block, it has steps up and down by one lattice spacing at a time and the steps are separated by at least $\gamma \sqrt{L}$ for some constant $\gamma$. To keep the fraction of atypical strings small in equation (\ref{eq:atypical-fraction}), we will choose $\gamma$ to be a sufficiently small constant. Now consider all the ways of partitioning this typical logical string into a correctable error and an uncorrectable error. Then, the fraction of exceptional partitions relative to all partitions of this string is bounded by 
\begin{equation}
    \frac{\mathrm{Exceptional}}{\mathrm{Total}} <  (\zeta+1) \exp (\gamma^2) \left(2  \right)^{-\gamma \sqrt{L}} \left(1+\mathcal{O}(1/L)\right).
\end{equation}
Exceptional terms are exponentially rare for typical short logical strings and large $L$.
\label{Lemma: exceptional terms}
\end{Lemma}

\begin{proof}
Take a logical string of length $\ell \leq L+2\zeta$ with typical shape. Each step is separated from the others by at least $\gamma \sqrt{L}$ for some $\gamma$. Now, consider taking a subset of $\frac{\ell-1}{2}$ of the qubits in the logical string. We would expect such a subset to be correctable. If not, this partition is exceptional.

Choose a partition of a connected logical string and let $O_U$ be the uncorrectable error and $O_C$ be the correctable error. $O_U$ and $O_C$ share a syndrome by definition. Denote that syndrome by $s$. The decoding algorithm, which in our case is minimal-weight decoding, applies some correction to this syndrome to return it to the code space. Call this correction $E_s$. $E_s$ is by definition a correctable error in the code, and therefore, because we are using minimal-weight decoding, $E_s$ must have lower weight than $O_U$. The fact that we chose our code size $L$ to be odd ruled out the case where the two might be equal. Now if the partition we are considering happens to exceptional, this means by definition that $O_C$ has higher weight than $O_U$, and we have
\begin{equation}
\label{eq: exceptional term condition}
    \text{Exceptional Term $(O_U \rho \, O_C)$} \implies \qquad |O_C| > |O_U| > |E_s| .
\end{equation}
We will use this condition to bound the number of exceptional terms for a given connected logical string.

What does it mean for $O_C$ to have higher weight than $E_s$? For connected logical strings of typical shape as in Definition \ref{Def: typical shape}, this happens only if on some subset or subsets of the logical string, the correctable error $O_C$ contains errors on qubits arranged in a ``cap." By this we mean a configuration of errors where the errors form three edges of a rectangle. The minimal-weight decoder will choose the fourth edge of the rectangle as part of the correction $E_s$. This is illustrated in figure \ref{fig: weight 7 uncorrectable A} and figure \ref{fig: weight 6 uncorrectable A}. If the connected logical string has length greater than $L$, then it has steps up and down if it crosses the code block left to right. In every exceptional term, the correctable error $O_C$ will contain the outermost qubits around some of the steps, forming a cap.

Now that we have a simple necessary condition for an exceptional term, we will bound the number of exceptional terms for each short logical string with a typical shape according to Definition \ref{Def: typical shape}. Start with a logical string of length $\ell$. Consider first the partitions into $\frac{\ell+1}{2}$ and $\frac{\ell-1}{2}$. Of course, those partitions for which the weight-$\frac{\ell+1}{2}$ error is correctable will be exceptional. Every exceptional term like this will have the property that the correctable error contains some number of ``caps" where all of the qubits around three sides of a rectangle are part of the correctable error. To bound the number of exceptional terms we will count the number of partitions with this property. 

Each partition of a weight-$\ell$ connected logical string into weight-$\frac{\ell+1}{2}$ and $\frac{\ell-1}{2}$ errors is formed by choosing $\frac{\ell-1}{2}$ out of the $\ell$ errors in the logical string. This is what we mean by a partition. We want to count the number of ways of choosing these errors such that the correctable error (of weight $\frac{\ell+1}{2}$ because we are counting exceptional terms) contains all the errors along a ``cap". This means that the subset of $\frac{\ell-1}{2}$ errors contains no errors along one or more of the ``caps." A typical short logical string running left to right across the code consists of horizontal segments separated by single steps up and down. The outermost of these steps form ``caps." The number of such ``caps" depends on the particular pattern of steps in the logical string. However, we can bound the number of exceptional terms by counting the number of ways of choosing no qubits along one of the horizontal segments of length $\gamma \sqrt{L}$. This is because every ``cap" consists of an outermost horizontal segment combined with the up and down steps on either side. This counting gives
\begin{equation}
    \text{Ways of choosing no qubits along a horizontal segment of length $\gamma \sqrt{L}$} = \binom{(\ell-\gamma \sqrt{L})}{\frac{\ell-1}{2}} .
\end{equation}
We want the number of ways of choosing no qubits along at least one of the horizontal segments. There are $\leq 2 \zeta$ steps up and down along the logical string. Therefore, there are $\leq 2 \zeta$ horizontal segments. We can use a union bound to write
\begin{equation}
    \text{Number of weight-$\left( \frac{\ell+1}{2}, \frac{\ell-1}{2} \right)$ exceptional terms} \leq 2\zeta \binom{\ell-\gamma \sqrt{L}}{\frac{\ell-1}{2}} .
\end{equation}
This is relative to the total number of $\left( \frac{\ell+1}{2}, \frac{\ell-1}{2} \right)$ partitions for our logical string of length $\ell$, which is
\begin{equation}
    \text{Total number of $\left( \frac{\ell+1}{2}, \frac{\ell-1}{2} \right)$ partitions} = \binom{\ell}{\frac{\ell-1}{2}} .
\end{equation}
We can expand the ratio of exceptional terms to the total using Stirling's approximation. This gives
\begin{align}
    2\zeta\binom{\ell-\gamma \sqrt{L}}{\frac{\ell-1}{2}}/\binom{\ell}{\frac{\ell-1}{2}} & = \frac{2\zeta \left( \ell - \gamma \sqrt{L} \right)! \frac{\ell+1}{2}! }{\ell! \left( \frac{\ell+1}{2} - \gamma \sqrt{L}\right)!} \nonumber
    \\
    & \approx 2\zeta \sqrt{\frac{(\ell+1)(\ell-\gamma \sqrt{L})}{2\ell\left(\frac{\ell+1}{2}-\gamma \sqrt{L} \right)}} \frac{(\ell-\gamma \sqrt{L})^{\ell-\gamma \sqrt{L}} \left(  \frac{\ell+1}{2} \right)^{\frac{\ell+1}{2}}}{\ell^{\ell} \left( \frac{\ell+1}{2} - \gamma \sqrt{L} \right)^{\frac{\ell+1}{2}- \gamma \sqrt{L}}} .
\end{align}
This approximation holds up to corrections $\mathcal{O}(1/\ell)$. We can rewrite this as
\begin{equation}
    = 2\zeta \sqrt{\frac{(\ell+1)(\ell-\gamma \sqrt{L})}{\ell(\ell+1-2 \gamma \sqrt{L})}} \left( 1- \frac{\gamma \sqrt{L}}{\ell} \right)^\ell \left(1- \frac{2 \gamma \sqrt{L}}{\ell+1} \right)^{-\frac{\ell+1}{2}} \left( \frac{\ell-\gamma  \sqrt{L}}{\frac{\ell+1}{2}-\gamma \sqrt{L}} \right)^{- \gamma \sqrt{L}} .
\end{equation}
Next we square the $(1-\frac{\gamma \sqrt{L}}{\ell})$ term in order to combine terms:
\begin{equation}
    = 2\zeta \sqrt{\frac{(\ell-\gamma \sqrt{L})}{\ell}} \left(\frac{1- \frac{2 \gamma \sqrt{L}}{\ell} + \left( \frac{\gamma \sqrt{L}}{\ell} \right)^2}{1- \frac{2 \gamma \sqrt{L}}{\ell+1}}  \right)^{\ell/2}  \left( \frac{\ell-\gamma \sqrt{L}}{\frac{\ell+1}{2}-\gamma \sqrt{L}} \right)^{- \gamma \sqrt{L}} .
\end{equation}
We upper bound the term inside the radical and also the term raised the power $\ell/2$:
\begin{equation}
    < 2\zeta \left(1 + \frac{\gamma^2}{\ell - 2 \gamma \sqrt{\ell}}  \right)^{\ell/2} \left( \frac{\ell-\gamma \sqrt{L}}{\frac{\ell+1}{2}-\gamma \sqrt{L}} \right)^{- \gamma \sqrt{L}} .
\end{equation}
The second of the three terms term is exponentially decaying to $\exp(\gamma^2/2)$. As long as $\ell\geq 4$, we can bound it by
\begin{equation}
    \left( 1+ \frac{\gamma^2}{\ell-2 \gamma \sqrt{\ell}} \right)^{\ell/2} < \exp \gamma^2 .
\end{equation}
Now, we bound $\frac{\ell-\gamma L}{(\ell+1)/2-\gamma L}>2$ and assemble one term raised to the power $L$ and another to the power $(\ell-L)/2$:
\begin{equation}
\label{eq: Exceptional terms error bound}
    < 2\zeta \exp(\gamma^2) 2^{- \gamma \sqrt{L}} .
\end{equation}
We chose some small value for $\gamma$ in Lemma \ref{Lemma: shape of strings B}, and then the number of exceptional terms with a weight-$\frac{\ell-1}{2}$ logical error on one side and a weight-$\frac{\ell+1}{2}$ correctable error on the other is exponentially small in $L$.

For the chosen connected logical string of weight $\ell$, we have calculated the fraction of exceptional terms among the partitions into $\frac{\ell-1}{2}$ and $\frac{\ell+1}{2}$. We will also have exceptional terms among the partitions into other weights, possibly all the way down to partitions into weight $\frac{L+1}{2}$ and $\ell-\frac{L+1}{2}$. Above, we applied the condition in equation (\ref{eq: exceptional term condition}) that for every exceptional term the correctable error must have higher weight than the minimal-weight correction. If we apply this same method to bound the number of exceptional terms among partitions into $\frac{\ell-3}{2}$ and $\frac{\ell+3}{2}$, we find that the correctable error must be at least $4$ longer than the minimal-weight correction. This means we want to count the number of configurations where at least two of the ``caps" are contained in the correctable error. This is clearly far fewer than the number of configurations where one ``cap" is contained. Therefore, the ratio of exceptional terms to total partitions is bounded by the ratio we found for partitions into $\frac{\ell-1}{2}$ and $\frac{\ell+1}{2}$.

In the end we see that number of weight-$\frac{\ell-1}{2}$ exceptional terms is exponentially small in $L$ for fixed $\zeta$ and further that the weight-$\frac{\ell-3}{2}$ exceptional terms are exponentially small in $L$ relative to the higher-weight exceptional terms, and so on. Then for large $L$, exceptional terms are negligible.

\end{proof}

Lemma \ref{Lemma: exceptional terms} allows us to approximate the sum over partitions for a typical, short logical string $\mathcal{L}$. Neglecting exceptional terms, the sum over partitions resembles the calculation of what we called $\delta$ in the repetition code in equations (\ref{eq:eps-delta-all-orders}) and (\ref{eq: Delta calculation chi matrix approach}). Let $\mathcal{L}$ have length $\ell$. Each partition contributes $\left( \frac{\sin \theta}{2} \right)^\ell$ with a phase. The sum over partitions is given by
\begin{align}
\label{eq: Sum over partitions}
    \sum_{\textrm{partitions}} (\textrm{Connected Part}) & = \left( \sum_{j=0}^{\frac{\ell-1}{2}} i^\ell (-1)^j \binom{\ell}{j} \left( \frac{\sin \theta}{2}\right)^\ell \right) (1 + \epsilon) \nonumber
    \\
    & = \left( i \binom{\ell -1 }{\frac{\ell-1}{2}} \left( \frac{\sin \theta}{2}\right)^\ell \right) (1 + \epsilon) ,
\end{align}
where $\epsilon$ is the error from exceptional terms, which is upper bounded
\begin{equation}
    |\epsilon| < 4\zeta \exp(\gamma^2) 2^{-\gamma \sqrt{L}} .
\end{equation}
This is two times the the expression in equation (\ref{eq: Exceptional terms error bound}), because each exceptional term contributes to the sum over partitions with the opposite sign relative to a non-exceptional term. 

\subsection{The Disconnected Part}
\label{sec: disconnected pieces}

In the preceding subsections, we analyzed the coherent component of the logical noise channel, expressed as a sum over many physical noise terms. So far we have only considered the connected logical string associated with each coherent term. In this subsection, we will analyze the disconnected errors in more detail, and describe more rigorously how they affect the evaluation of the coherent terms in the logical channel. In Section \ref{section: The Coherent Sum} we described how to decompose a contribution to $\tilde\chi_{\scriptscriptstyle{Z_1 I}}$ into a connected piece and some number of disconnected pieces. The left and right hand side of each coherent term can be expanded as the product of the errors contained in the connected logical string and the errors outside of it; schematically,
\begin{equation}
    \label{eq: Disconnected part factors out}
    (\mathrm{Conn}_L \mathrm{Disc}_L\, \rho \, \mathrm{Conn}_R \mathrm{Disc}_R) = (\mathrm{Conn}_L \, \rho \, \mathrm{Conn}_R) \, \mathrm{Disconnected}.
\end{equation}
The factor ``Disconnected" means the contribution to the coherent term from disconnected components that appeared in equation (\ref{eq: coherent sum string form}). The product of the two (disjoint) factors $ \mathrm{Conn}_L$ and  $\mathrm{Conn}_R$ yields the connected logical string, with no additional disjoint loops included. The connected factor includes $\sin \theta/2 \cos \theta/2$ for each qubit along the connected logical string. The disconnected factor includes $(\cos \theta/2)^2$ on every qubit not in the connected logical string in addition to a sum over all possible disconnected errors.

\begin{figure}
    \centering
    \includegraphics[height = 8cm]{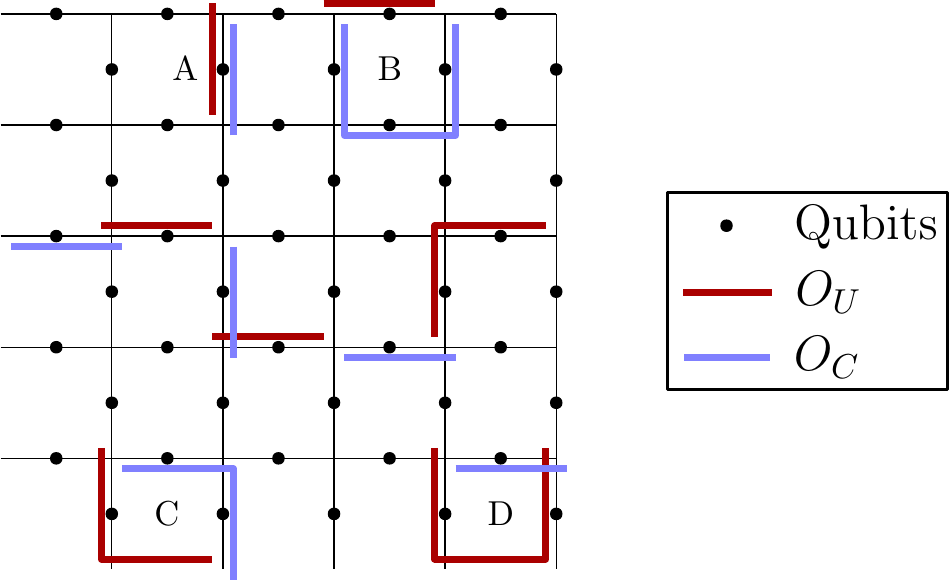}
    \caption{Here we have a partition $(O_U \rho \, O_C)$ of a connected logical string adorned in four different ways by added errors. The errors in red are the uncorrectable part, $O_U$, of the partition, while the errors in blue form the correctable part, $O_C$. The four added errors are labelled A, B, C, and D. In A, the same error has been added to both $O_U$ and $O_C$. In B and D, three errors are added to one side of the partition and one to the other. This produces a minus sign. In C, two errors are added to each side.}
    \label{fig: added error examples}
\end{figure}

Fix a partition $(O_U \rho \, O_C)$ of a short, typical logical string, and consider dressing it with disconnected errors. We can distinguish two types of added errors: incoherent and coherent. If the disconnected error is $D_L$ acting on the density operator from the left, and $D_R$ acting from the right, then if a particular qubit is hit by the same error contained in both $D_L$ and $D_R$, we say that the disconnected error acting on that qubit is incoherent. If a particular qubit is hit by distinct errors contained in $D_L$ and $D_R$, then the error is coherent. The product $D_LD_R$ of the errors added on right and left must be a non-identity stabilizer operator, i.e. a closed loop or a set of disjoint closed loops. (Here, because we are investigating the encoded $Z$ errors in the logical channel, only the $Z$-type physical errors are considered.) The two types of added error --- incoherent and coherent --- are shown in figure \ref{fig: added error examples}, where (A) is an incoherent-type added error and (B)-(D) are coherent-type.

Let us first treat the case of incoherent-type added errors, where $D_L = D_R \equiv D$. These are the ones with the same disconnected error added to both operators in the partition, for example (A) from figure \ref{fig: added error examples}. These terms do not change the phase of the original partition, and they multiply the magnitude by $(\sin \theta/2)^{2m}$ if $m$ is the weight of the error added on each side. The disconnected part contains $\cos^2 \theta/2$ on each qubit corresponding to no disconnected errors plus many configurations of disconnected errors. The incoherent-type added errors on each qubit in the disconnected part supply the $\sin^2 \theta/2$ term to give 1 on the qubits not contained in the connected logical string. This reasoning applies to each incoherent-type added error that does not change how the operators $O_U$ and $O_C$ are  decoded. In other words, if $D$ is the disconnected error we add to $O_U$ and $O_C$, we require that $D O_U$ is an uncorrectable error.

We must be careful because in some cases the added incoherent-type errors can change how the correctable and uncorrectable errors in the partition are decoded. The added error can ``flip" the uncorrectable error to a correctable one. This means that the noise term that contributes to the logical $\tilde{\chi}_{\scriptscriptstyle{Z_1 I}}$ component is not $(D O_U \rho D O_C)$ as we would have expected but is instead $(D O_C \rho D O_U)$. This term has the opposite sign relative to the expected term. This is only possible when the added error $D$ is located very near the connected logical string and only for special partitions. We prove in Lemma \ref{lemma: added error exceptional terms} in Appendix \ref{app: Disconnected errors} that the contribution from these disconnected terms is negligible.

What of the coherent-type added errors? Again, fix a partition of a connected logical string. Let $O_U$ and $O_C$ be the correctable and uncorrectable errors. Now consider choosing a stabilizer operator or a closed loop, $\ell$ that is disjoint from the connected logical string. Let the length of the loop be $|\ell|$. Now choose a subset of $p$ of the qubits in the loop, and let the disconnected error $D_L$ act on these $p$ qubits from the left, while the disconnected error $D_R$ acts on the remaining $|\ell|-p$ qubits from the right. Suppose further that the qubits in the loop and the partition are such that the uncorrectable error $O_U$ plus the additional error $D_L$ remains uncorrectable. This need not always be true; we will consider the case where the $O_U D_L$ is correctable in a moment. 

Supposing that the disconnected error $D_L$ does not change the decoding, we can perform a sum over all the ways of choosing the $p$ errors in $D_L$ from among the $|\ell|$ errors in the loop. The number of ways of choosing $p$ errors is given by a binomial coefficient, and the magnitude of each term is suppressed by $(\sin \theta/2)^{|\ell|}$ relative to the original partition of the connected logical string without any additional disconnected errors added. The phase of each term depends, as always, on the relative weight of the errors on the right and the left. The disconnected part contributes a phase of $(i)^p (-i)^{|\ell|-p}$, and $\ell$ is a closed loop so $|\ell|$ is even. The sum yields
\begin{equation}
    (\text{connected part}) \sum_{p=0}^{|\ell|} (-1)^{p} (-1)^{|\ell|/2} \binom{|\ell|}{p} (\sin \theta/2)^{|\ell|} =0 .
\end{equation}
When we sum over all ways of forming disconnected terms out of the original loop $\ell$, the sum is 0. This holds for any loop such that the disconnected part does not change how the connected part is decoded.

In the examples we considered in figure \ref{fig: added error examples}, the additional disconnected errors did not change how the connected part was decoded. This is the same condition we encountered in the discussion of incoherent-type added errors. In certain cases the error $D_L$ that we add to the $O_U$ side of the partition can be such that $D_L O_U$ is a correctable operator. This means the partition is ``flipped" by the disconnected error. We account for this case in Lemma \ref{lemma: added error exceptional terms} and prove that the contribution to the logical noise from these special disconnected terms is negligible for short logical strings.

Using Lemma \ref{lemma: added error exceptional terms}, we can neglect the added errors that change how the partition is decoded. Then, we can conclude that the net contribution from coherent-type added errors is 0 and the incoherent-type added errors contribute a $\sin^2 \theta/2$ factor on each qubit not in the connected logical string. This implies that the ``Disconnected Sum" term in equation (\ref{eq: coherent sum string form}) is equal to 1 plus a small correction. This implies that
\begin{equation}
\label{eq: disconnected = 1 bound}
    \tilde{\chi}_{\scriptscriptstyle{Z_1 I}} = \left[ \sum_{\mathcal{L}} \sum_{\mathrm{partitions}} (O_U \rho \, O_C)  \right](1 + \mathcal{E}) + \text{High Weight} ,
\end{equation}
where $\mathcal{L}$ is a connected, short, typical logical string, partitions refers to the partitions of $\mathcal{L}$ denoted $(O_U \rho \, O_C)$, and $\mathcal{E}$ is a noise term. The error term satisfies
\begin{equation}
    \mathcal{E} \leq \frac{16 \gamma \zeta^2}{\sqrt{L}} + \mathcal{O}(1/L) .
\end{equation}
This error term is from Lemma \ref{lemma: added error exceptional terms} and comes from the added errors that change how the partition is decoded. The term ``High Weight" in equation (\ref{eq: disconnected = 1 bound}) is the error from Lemma \ref{Lemma: Coherent path counting} corresponding to the contributions of logical strings with length $>L+2 \zeta$. We have not yet justified that this error is small relative to the short strings. This is because we do not have a lower bound on the short strings. The justification comes from our subsequent discussion of the incoherent logical noise components.

\subsection{Incoherent Sum}
\label{section: Incoherent Sum}

Now that we have simplified the sum for the coherent components of the logical noise channel, factored out the disconnected pieces, and performed the sum over syndromes for the connected pieces, we turn our attention to the incoherent logical noise components. We start by making several of the same simplifications we made in the coherent sum. Of the incoherent logical components $(\tilde{L_a} \tilde{\rho} \tilde{L_a})$, we neglect all the terms where $L_a$ is a logical $Y$ operator or acts non-trivially on both encoded qubits. We retain only the terms where $L_a$ is a logical $X$ or $Z$ on one of the two encoded qubits. The reason is the same as for the coherent sum. The neglected terms are much higher weight, such that the path counting excludes them. Then we have the sum
\begin{equation}
    \label{eq: Incoherent Term}
    \left( \tilde{L}_a \tilde{\rho} \tilde{L}_a^\dagger \right) = \sum_{s,x,y} \left( E_s L_a G_x \rho \, G_y^\dagger L_a^\dagger E_s^\dagger \right) ,
\end{equation}
where $L_a$ is an $X$ or $Z$ logical operator on one of the encoded qubits and identity on the other. Again, we suppose that all the angles are equal to some fixed $\theta$ for each single-qubit rotation. We will extend to general rotations in Lemma \ref{Lemma: general rotation angles}.

Again, we will divide each term into connected and disconnected pieces. In this discussion of the incoherent logical noise components, Definition \ref{Def: Connected coherent} must be modified. The noise terms that enter into the incoherent logical noise contain an uncorrectable error on both sides of $\rho$. We will need to consider two logical strings in our definition.
\begin{Definition}
\label{Def: Connected incoherent}
    The ``connected part" of a noise term $(E_s L_a G_x \rho \, G_y^\dagger L_a^\dagger E_s^\dagger)$ is a noise term defined in the following way: let $\mathcal{L}_1$ equal $L_a G_x$ with all topologically trivial closed loops removed and $\mathcal{L}_2$ equal $L_a G_y$ with all trivial closed loops removed. Then, let $A$ denote the set of qubits $\subset \mathcal{L}_1 \cup \mathcal{L}_2$ where either $E_s L_a G_x$ or $E_s L_a G_y$, or both, act non-trivially. The connected part of $(E_s L_a G_x \rho \, G_y^\dagger L_a^\dagger E_S^\dagger)$ is given by 
    \begin{equation}
        ((E_s L_a G_x)|_{A} \, \rho \, ( G_y^\dagger L_a^\dagger E_s^\dagger)|_{A}) .
    \end{equation}
    $|_{A}$ denotes the restriction of an operator to the set of qubits $A$.
\end{Definition}

If the incoherent term is $(O_U \rho \, O_U^\prime)$ then this definition captures the set of qubits in the support of $O_U$ or $O_U^\prime$ that lie along the two logical strings formed by $O_U$ and $E_s$ and $O_U^\prime$ and $E_s$ pruned of all trivial closed loops. Figure \ref{fig: incoherent connected and disconnected} illustrates the connected and disconnected part of a noise term that enters into the incoherent logical noise. The connected part of the noise term in the figure features factors of $\sin \theta/2 \cos \theta/2$ for the qubits that appear in exactly one of $O_U$ or $O_U^\prime$ and $\sin^2 \theta/2$ for the qubits that appear in both $O_U$ and $O_U^\prime$. We can lower bound the connected part of each incoherent noise term by $(\sin \theta/2 \cos \theta/2)^{|O_U|+|O_U^\prime|}$. This will be useful later on when we sum over many possible choices for the operators $O_U$ and $O_U^\prime$.

\begin{figure}
    \centering
    \includegraphics[height = 8cm]{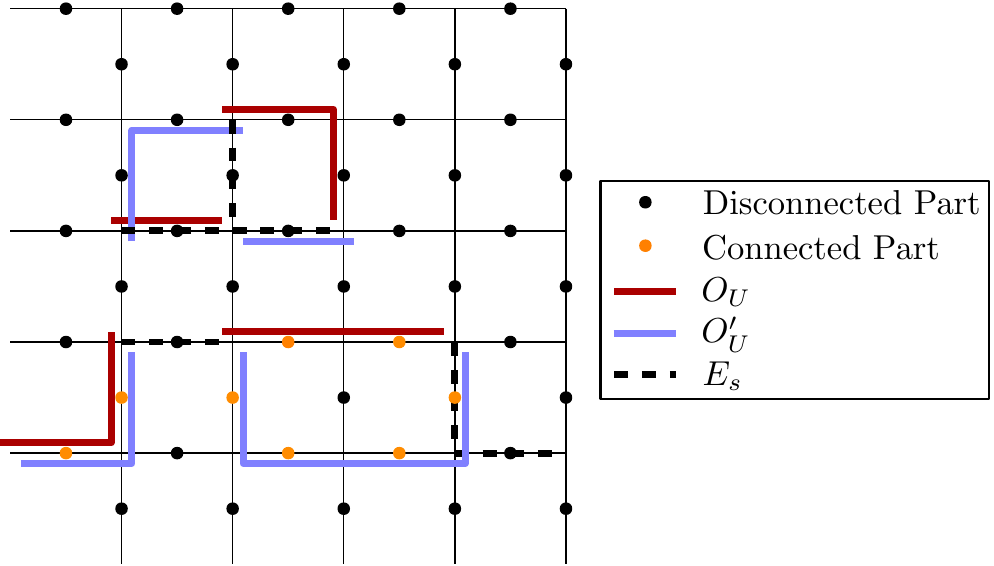}
    \caption{A noise terms $(O_U \rho \, O_U^\prime)$ is shown with $O_U$ in red and $O_U^\prime$ in blue. The standard correction $E_s$ chosen by the minimal-weight decoder is drawn as a dotted black line. The connected part of this noise term is signified by the orange qubits, while the disconnected part contains the black qubits.}
    \label{fig: incoherent connected and disconnected}
\end{figure}

\begin{Definition}
\label{Def: Disconnected Incoherent}
    The ``disconnected part" of a noise term $(E_s L_a G_x \rho \, G_y^\dagger L_a^\dagger E_s^\dagger)$ is the restriction of the noise term to the qubits not in the connected part. In Definition \ref{Def: Connected incoherent} we constructed the set $A$, which contained the qubits in the connected part. The disconnected part is given by 
    \begin{equation}
        ((E_s L_a G_x)|_{A^C} \, \rho \, (G_y^\dagger L_a^\dagger E_s^\dagger)|_{A^C}) ,
    \end{equation}
    where $|_{A^C}$ denotes the restriction of an operator to the complement of the set $A$.
\end{Definition}

For the example in figure \ref{fig: incoherent connected and disconnected}, the disconnected part features factors of $\sin \theta/2 \cos \theta/2$ for the six qubits along the trivial closed loop near the top of the figure and $\cos^2 \theta/2$ for the rest of the qubits. For a given connected part, we can imagine adding disconnected errors to form many different noise terms. The connected part contains factors of $\sin \theta/2 \cos \theta/2$ for each qubit that appears in one of the uncorrectable errors and $\sin^2 \theta/2$ for each qubit that appears in both errors. The disconnected term includes $\cos^2 \theta/2$ for each every qubit not in the connected part plus a sum over all possible coherent and incoherent-type disconnected error. Just as in Section \ref{sec: disconnected pieces}, when the disconnected errors do not change how the connected term is decoded, the incoherent-type errors give $\cos^2 \theta/2 + \sin^2 \theta/2 =1$ on qubits not in the connected part. The coherent-type disconnected errors, which form loops split between left and right, sum to zero because of the alternating signs.

Just as in the case of the coherent logical noise components, some disconnected errors will not be allowed because they change how the connected term is decoded. We will set the disconnected part equal to $1$ plus an error term that comes from these disallowed disconnected errors. In Lemma \ref{lemma: incoherent disconnected part}, we justify this by proving that the error term is small. This is analogous to Lemma \ref{lemma: added error exceptional terms}, where we prove that the disconnected part of the coherent logical noise components is equal to $1$ up to small corrections.

We want to continue to follow a similar argument to the one for the coherent terms. The next step is restricting the set of connected terms we consider. We will break up each error into connected and disconnected pieces and restrict ourselves to noise terms with low-weight connected part, where the total weight of the connected part is bounded by $L+2 \zeta+1$; here $ \zeta$ is the same $L$-independent cutoff as in the coherent sum. Just as for the analysis of the coherent logical noise in Section \ref{sec: Path Counting}, we will require $\theta$ to scale like $1/L$ to justify this truncation of the noise terms contributing to the connected part. 
 
\begin{Lemma}
    Consider an incoherent logical noise component, say $\tilde{\chi}_{\scriptscriptstyle{Z_1 Z_1}}$. We write this logical noise component as a sum over physical noise terms $(O_U \rho \, O_U^\prime)$. Then, if $| \sin \theta| < 1/L$, we can truncate the sum to include only those noise terms where $|O_U| + |O_U^\prime|\leq L+2 \zeta+1$, where $\zeta$ is the same cutoff constant as in Lemma \ref{Lemma: Coherent path counting}. In other words,
    \begin{equation}
        \tilde{\chi}_{\scriptscriptstyle{Z_1 Z_1}} = \sum_{O_U, O_U^\prime : ~|O_U| + |O_U^\prime| \leq L+2\zeta+1} (O_U \rho \, O_U^\prime) \times \mathrm{Disconnected} \times \left( 1 + \mathcal{O}\left( (L \sin \theta)^{2\zeta} \right) \right).
    \end{equation}
    \label{Lemma: High Weigh Incoherent Terms}
\end{Lemma}
\begin{proof}
    We split each noise term into connected and disconnected parts. We show in Lemma \ref{lemma: incoherent disconnected part} that the disconnected part is decreasing as the weight of the connected part increases. Moreover, the disconnected part is approximately equal to $1$ for connected terms with total weight $\leq L+2\zeta+1$. Therefore, we need only consider the connected part as we proceed to truncate the sum and upper bound the error. 
    
    Let us denote the connected part of a noise terms that enter into the logical $\tilde{\chi}_{\scriptscriptstyle{Z_1 Z_1}}$ component by $(O_U \rho \, O_U^\prime)$. All such noise terms have the shape drawn in figure \ref{fig: connected incoherent example}.
    \begin{figure}
        \centering
        \includegraphics[height = 10cm]{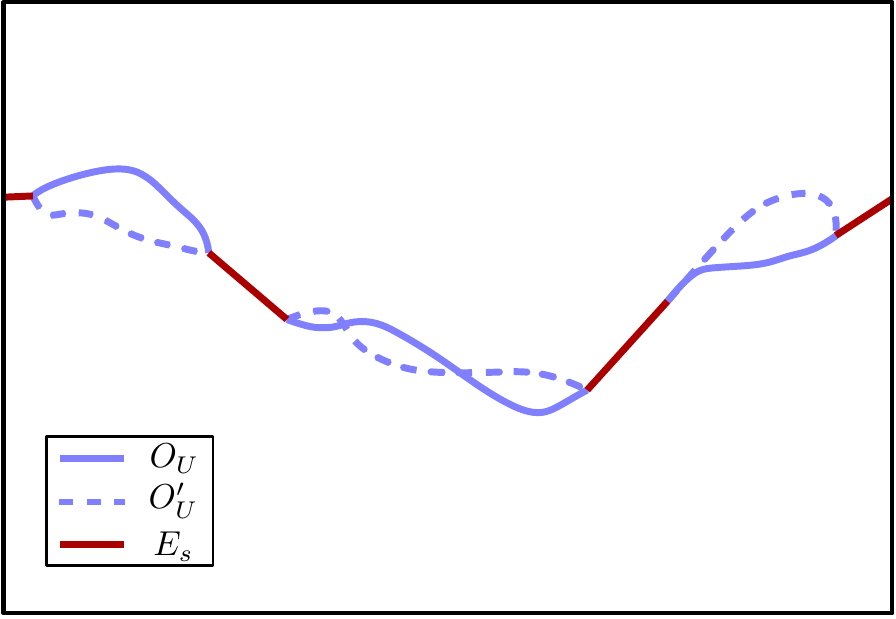}
        \caption{This is a connected incoherent noise term, $(O_U \rho \, O_U^\prime)$. The operator $O_U$ is drawn with the solid blue line. The operator $O_U^\prime$ is drawn with the dashed blue line. The standard correction $E_s$ for the syndrome shared by $O_U$ and $O_U^\prime$ is drawn with the solid red line. Each connected incoherent term is a configuration of loops where each loop is formed by two segments with shared endpoints, one segment from $O_U$ and one from $O_U^\prime$. The operator $E_s$ links together the loops like beads on a string, so that $O_U E_s$ and $O_U^\prime E_s$ are both continuous logical strings spanning the code.}
        \label{fig: connected incoherent example}
    \end{figure}
    The connected part is supported on a series of loops. The loops are joined by the minimal-weight correction. Denote the minimal-weight correction of $O_U$ and of $O_U^\prime$ by $E_s$. Let $w$ be the weight of $O_U$ and $w^\prime$ the weight of $O_U^\prime$. Suppose $w\le w^\prime$. If not, then swap $O_U$ and $O_U^\prime$ in what follows. Let $E_s O_U = \mathcal{L}$. Then, $\mathcal{L}$ is a connected $Z_1$ logical string with length at most $2w-1$. The number of such logical strings is upper bounded by $L^{2w-L}$ using our bound on the number of logical strings in equation (\ref{eq: Logical strings with length l}).
    
    The number of ways of choosing a weight $w$ operator $O_U$ as a subset of $\mathcal{L}$ is upper bounded by $2^{2w-2}$. Now that $O_U$ is fixed, the number of ways of choosing the operator $O_U^\prime$ is upper bounded by $\mathcal{O}(L^{w^\prime - w})$. 
    This is because the lowest-weight operator with the same syndrome and logical action has weight $\leq w$. Then, $O_U^\prime$ consists of this lowest-weight operator combined with a number of additional deviations like we considered to derive equation (\ref{eq: Logical strings with length l}). 
    (Here we are neglecting a factor which is polynomial in $w$ and $w'$; bounding the exponential dependence on $w'-w$ will suffice for what follows.)
    All together we have the following upper bound on the number of noise terms with fixed $w$ and $w^\prime$:
    \begin{equation}
    \label{eq: bound on number of terms of high-weight incoherent}
        \leq 2^{2w-2} L^{2w-L} L^{w^\prime - w} .
    \end{equation}
    Each of these terms has magnitude at most $(\sin \theta/2)^{w+w^\prime}$, which is positive because $w+w^\prime$ is even. As in Lemma \ref{Lemma: Coherent path counting}, we will truncate the sum and keep only those connected noise terms with $w+w^\prime \leq L+2\zeta + 1$. If we let $w+w^\prime = w_{\mathrm{total}}$, for each $w_{\mathrm{total}}$ there are several combinations of $w$ and $w^\prime$ with the same total. Because $w$ and $w^\prime$ must be $> (L+1)/2$, there are less than $w_{\mathrm{total}}-L$ combinations. We perform a sum over $w_{\mathrm{total}}$ from $L+2\zeta+1$ up to the maximum weight. Therefore, if we let $\epsilon$ denote the contribution from the higher weight connected terms to $\tilde{\chi}_{\scriptscriptstyle{Z_1 Z_1}}$, then $\epsilon$ is bounded by
    \begin{equation}
         \epsilon \leq \mathcal{O}((2 \sin (\theta/2))^{L+2\zeta+1} L^{2\zeta+2} ).
    \end{equation}
    Here we have estimated the sum over $w_{\mathrm{total}}$ using the same method as in the derivation of equation (\ref{eq: Path counting coherent tail bound}). We will compare this error $\epsilon$ to the contribution from the lowest weight noise terms. These terms have $w = w^\prime = (L+1)/2$, and contribute at least $\xi$, where
    \begin{equation}
        \xi = L (\sin \theta )^{L+1} .
    \end{equation}
    Then, the relative error associated with our truncation is given by
    \begin{equation}
        \frac{\epsilon}{\xi} \leq \mathcal{O}((L \sin \theta )^{2\zeta}) .
    \end{equation}
    We have neglected a polynomial factor in $L$ in our counting of noise terms. Nevertheless, as long as $L |\sin \theta| <1$, the relative error is exponentially small in $ \zeta$, and the higher-weight connected terms are negligible.
    
\end{proof}

\subsection{The Incoherent Sum Over Strings}
\label{sec: Incoherent sum over strings}

The connected part of the incoherent components is not as simply expressed as a sum over strings as the coherent components because each uncorrectable error $E_s L_a G_x$ can generally be completed to many different logical strings by multiplying by different correctable errors. Nevertheless, we can rewrite the sum in a similar way. This will form our primary tool for comparing coherent and incoherent logical noise components. We will write a sum over each logical string with logical action $L_a$. For each string, we will sum over different ways of choosing the uncorrectable subset $O_U$. We will restrict the subsets we consider for each logical string in order to control the over-counting factor that we will describe shortly. Then for each operator $O_U$, we will sum over all possible uncorrectable operators $O_U^\prime$ with the same syndrome and logical action. Fix a connected logical string $\mathcal{L}$ with length $\ell$ and choose an uncorrectable subset of the logical string $O_U$ with weight $w$. We impose two constraints on $O_U$: first that $w \geq (\ell+1)/2$ and second that the complement of $O_U$ has the same weight as the minimal-weight correction $E_s$. The complement of $O_U$ is $O_U \mathcal{L}$, which we will denote $O_C$. Note that the name $O_C$ is chosen in analogy to the way the errors were labelled in the coherent noise terms, but $O_C$ is not a part of the incoherent noise term, which is notated $(O_U \rho \, O_U^\prime)$. Now that $O_U$ is fixed, choose a second uncorrectable error $O_U^\prime$ with the same syndrome and with weight $w^\prime$. This will produce every incoherent connected term $(O_U \rho \, O_U^\prime)$. However, each uncorrectable error $O_U$ will appear many times as a subset of many different logical strings. This is the over-counting we mentioned above. 

Each operator $O_U$ can be completed to a logical string in many ways. Because of the constraints we imposed on the subset $O_U$, the complement, which we called $O_C$, must have the same weight as the minimal-weight correction to $O_U$, denoted $E_s$. Let $\{O_C^\prime\}$ be the set of possible complements. Then, each possible complement $O_C^\prime \in \{O_C^\prime\}$ defines a logical string $O_U O_C^\prime$, which will appear in the sum over strings. For each operator $O_U$ in the sum over strings, we need to divide by the number of complements $O_C^\prime$ with weight $|E_s|$. Each incoherent logical noise component can be written as a sum over connected logical strings $\mathcal{L}$ times a disconnected factor. This form of the sum will allow us to compare with equation (\ref{eq: coherent sum string form}). We sum over logical strings, and for each logical string we sum over possible choices of $O_U$ and $O_U^\prime$. We divide by the over-counting factor for each $O_U$. This gives
\begin{equation}
\label{eq: incoherent sum string form}
    \tilde{\chi}_{\scriptscriptstyle{Z_1 Z_1}}  = \sum_{\mathcal{L}} \sum_{O_U \subset \mathcal{L}} \frac{1}{| \{O_C^\prime \}|} \sum_{O_U^\prime} (O_U \rho \, O_U^\prime) \times \text{Disconnected}
    ,
\end{equation}
where 
\begin{equation}
    |\mathcal{L}| = \ell, \quad |O_U| \geq \frac{\ell+1}{2}, \quad |O_C^\prime| = |E_s|.
\end{equation}
To reiterate, equation (\ref{eq: incoherent sum string form}) expresses an incoherent logical noise component as a sum over connected logical strings. For each string $\mathcal{L}$ with weight $\ell$ we sum over all uncorrectable subsets $O_U$ of weight $\geq (\ell+1)/2$ such that the complement $O_C$ has weight equal to the minimal-weight correction of $O_U$, namely $E_s$. For each $O_U$ we must divide by the number of complements $O_C^\prime$ with the same syndrome and weight as $E_s$ in order to cancel the over-counting. $\{O_C^\prime \}$ is the set of such operators, and $|\{O_C^\prime \} |$ is its cardinality. Finally, we sum over all uncorrectable operators $O_U^\prime$ with the same syndrome to produce the complete set of incoherent terms. We will prove the following lemma, which provides a lower bound on the contribution of each logical string to the incoherent logical noise component. We will apply this lemma to lower bound the incoherent logical noise strength in terms of the coherent logical noise strength.

\begin{Lemma}
    As long as $|\sin \theta| < 1/L$, we can apply Lemma \ref{Lemma: High Weigh Incoherent Terms}. This means that in equation (\ref{eq: incoherent sum string form}) we can restrict to the case where $|O_U|+|O_U^\prime| \leq L+2 \zeta+1$. Let us also suppose that $|O_U| = |O_U^\prime|$. This assumption will be justified by Lemma \ref{Lemma: Mismatched weights}. Then, we can pick a connected logical string $\mathcal{L}$ with $|\mathcal{L}| = \ell$ such that $\ell \leq L+2\zeta$. $\mathcal{L}$ is a $Z_1$ logical string if we are calculating the $\tilde{\chi}_{\scriptscriptstyle{Z_1 Z_1}}$ logical noise component. $O_U$ is subset of $\mathcal{L}$ such that $O_U$ is corrected to a logical $Z_1$ operator and $|O_U| = (\ell+1)/2$. $O_U^\prime$ is a operator with the same weight, syndrome, and logical action as $O_U$. For each fixed $\mathcal{L}$ with length $\ell \leq L+2\zeta$ the following holds:
    \begin{equation}
        \sum_{O_U} \frac{1}{|\{O_C^\prime \}|} \sum_{O_U^\prime} 1 \geq \sum_{O_U} 1 .
    \end{equation}
    \label{Lemma: Incoherent string sum bound}
\end{Lemma}
\begin{proof}

    For each short logical string $\mathcal{L}$ with length $\ell$, we partition it into an uncorrectable operator $O_U$ of weight $w = (\ell+1)/2$ and a correctable operator $O_C$ of length $|O_C|=|E_s|$. Then, we consider the alternative uncorrectable and correctable paths, $O_U^\prime$ and $O_C^\prime$, with weight $w$ and $|E_s|$, respectively. The logical string $\mathcal{L}$ is short, so we can use Lemmas \ref{Lemma: shape of strings A} and \ref{Lemma: shape of strings B}. Say the logical string runs right to left across the code. We observe by studying figure \ref{fig: multiplicity 2} that we have multiple possible strings of the same weight exactly when both a vertical error and some number of adjacent horizontal errors are contained in either the correctable or uncorrectable part. 
    
    Suppose that for some partition consisting of an uncorrectable operator $O_U^{(1)}$ and a correctable operator $O_C^{(1)}$, denote the operators with the same weight, syndrome, and logical action by $O_U^{\prime \, (1)}$ and $O_C^{\prime \, (1)}$. Suppose $|\{O_C^{\prime \, (1)}\}| > |\{O_U^{\prime \, (1)}\}|$. We construct a new operator $O_U^{(2)}$ by exchanging all but one of the errors in $O_U^{(1)}$ with the errors in $O_C^{(1)}$, so that $O_U^{(2)}$ is equal $O_C^{(1)}$ plus one additional error. This is shown in figure \ref{fig: multiplicity 1}, where we have kept the error on the farthest left vertical segment fixed and flipped the rest relative to the term in figure \ref{fig: multiplicity 2}. For every $O_U$ there are $w$ possible mappings, one for each of the $w$ choices of the single-qubit error that remains fixed. In the same way, every $O_U$ is mapped onto by $w$ different mappings acting on $w$ other operators with the same logical action as $O_U$. Then, there exists a convention that selects exactly one partner for each $O_U$. 
    
    \begin{figure}
        \centering
        \includegraphics[height = 10cm]{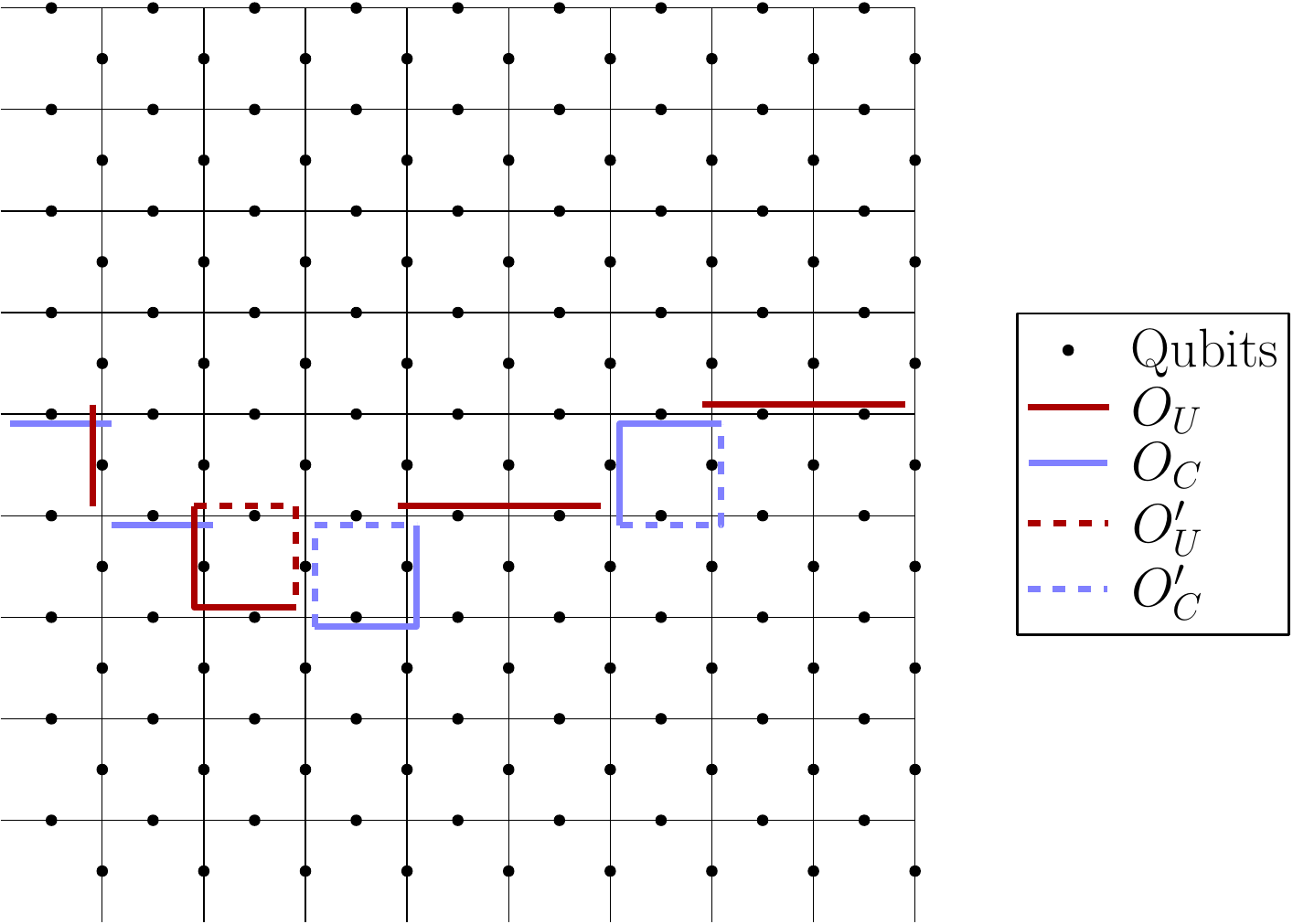}
        \caption{A partition of a length-$13$ logical string is shown in a toric code with $L=9$. The operator $O_U$ is shown in solid red. The operator $O_C$ is shown in solid blue. The alternative operators with the same weight, syndrome, and logical action, which we denoted $O_U^\prime$ and $O_C^\prime$, are drawn with dotted lines. For this partition $|\{O_U^\prime\}| = 2$ and $|\{O_C^\prime \}| = 4$.}
        \label{fig: multiplicity 1}
    \end{figure}
    
    \begin{figure}
        \centering
        \includegraphics[height = 10cm]{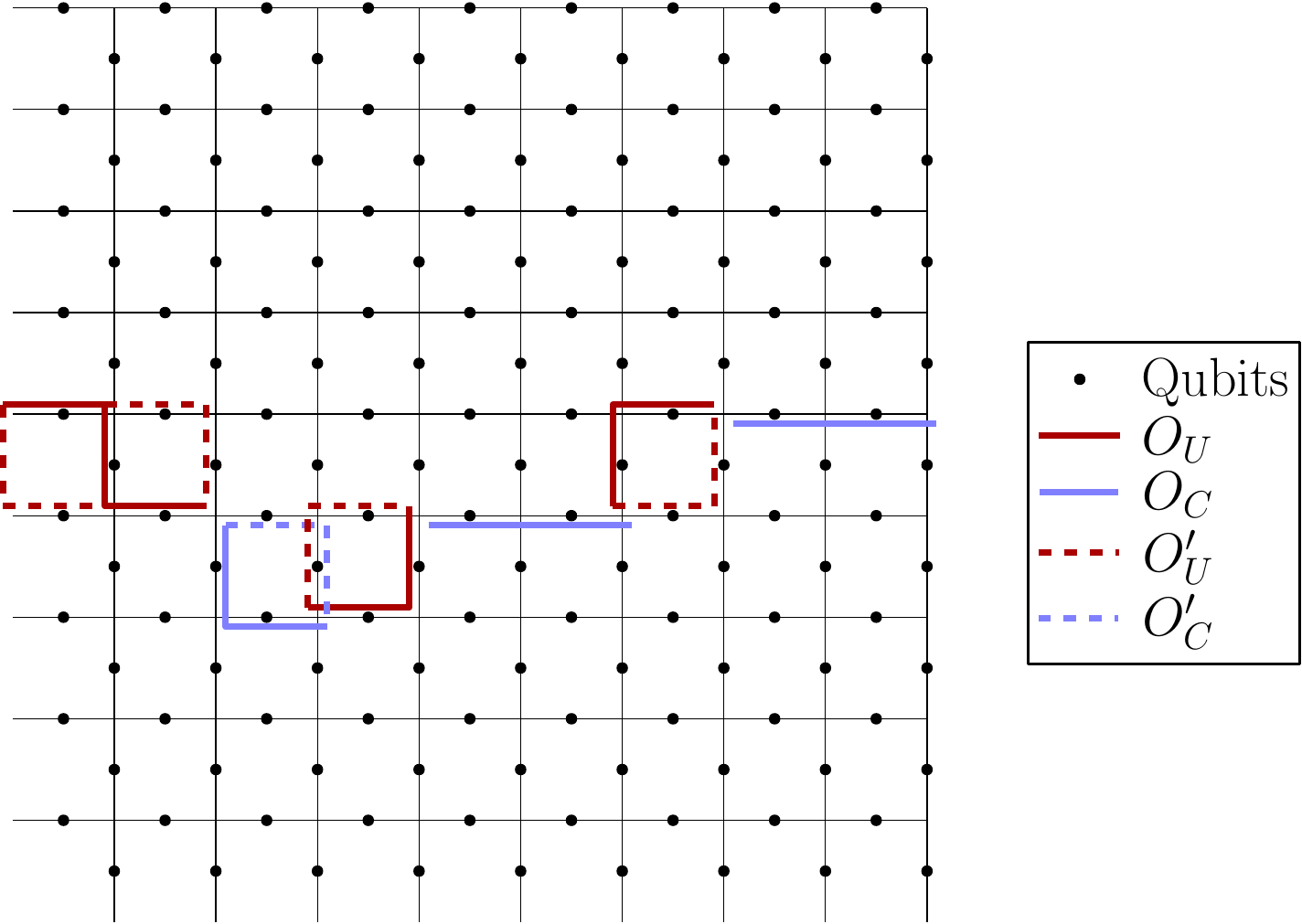}
        \caption{This is a partner of the partition shown in figure \ref{fig: multiplicity 1}. It is another partition of the same logical string, and the errors in $O_U$ and $O_C$ are interchanged except for one qubit. In this case, that qubit is the one that lies on the farthest left vertical segment. The error on that qubit is part of $O_U$ in both partitions. Once again, the operator $O_U$ is in red and the operator $O_C$ is in blue. The alternative operators with the same weight, syndrome, and logical action are given by the dashed lines. For this partition $|\{ O_U^\prime \}| = 12$ and $|\{ O_C^\prime \}| = 2$.}
        \label{fig: multiplicity 2}
    \end{figure}
    
    We assumed that for the $O_U^{(1)}$ that we started with
    \begin{equation}
        |\{O_C^{\prime \, (1)} \}| > |\{O_U^{\prime \, (1)} \}| .
    \end{equation}
    The mapping we described constructs a partner $O_U^{(2)}$ such that 
    \begin{equation}
        |\{O_U^{\prime \, (2)} \}| \geq |\{O_C^{\prime \, (1)} \}| > |\{O_U^{\prime \, (1)} \}| \geq |\{O_C^{\prime \, (2)} \}| .
    \end{equation}
    Then for each pair $O_U^{(1)}$ and $O_U^{(2)}$, we can lower bound the contribution to the incoherent logical noise using
    \begin{align}
        & \frac{|\{O_U^{\prime \, (1)} \}|}{|\{O_C^{\prime \, (1)} \}|}+ \frac{|\{O_U^{\prime \, (2)} \}|}{|\{O_C^{\prime \, (2)} \}|} \nonumber \\
        & \geq \frac{|\{O_U^{\prime \, (1)} \}|}{|\{O_C^{\prime \, (1)} \}|} + \frac{|\{O_C^{\prime \, (1)} \}|}{|\{O_U^{\prime \, (1)} \}|} \nonumber \\
        & = \frac{|\{O_U^{\prime \, (1)} \}|^2 + |\{O_C^{\prime \, (1)} \}|^2 }{|\{O_U^{\prime \, (1)} \}| |\{O_C^{\prime \, (1)} \}|} \nonumber \\
        &\geq 2 .
    \end{align}
    Finally, we apply the lower bound to the entire sum over $O_U$ to conclude
    \begin{equation}
        \sum_{O_U} \frac{|\{O_U^\prime\}|}{|\{O_C^\prime\}|} \geq \sum_{O_U} 1 .
    \end{equation}
    The number of terms in the sum over $O_U$ is at most $\binom{\ell}{w}$, where $\ell$ is the length of the logical string $\mathcal{L}$. For typical, short logical strings the binomial coefficient will be the number of terms in the sum over $O_U$ up to a small correction.
    
\end{proof}

\subsection{Noise Terms with Mismatched Weight}
\label{sec: Mismatched weights}

We have already shown that we can neglect the high-weight noise terms in the incoherent logical noise components, and we can also write the incoherent logical noise components as a sum over logical strings. Next, we will show that among the low-weight noise terms, we may neglect the terms with different weight errors on each side of $\rho$. This is crucial to our proof that the coherence of the logical noise is suppressed. We will construct a lower bound on the incoherent logical noise components and an upper bound on the coherent logical noise components. The noise terms with mismatched weight enter with a phase of $-1$ whenever the difference between the weights on left and right $= 2 \mod 4$. A large contribution from noise terms with mismatched weight could spoil our lower bound on the incoherent logical noise. Fortunately, no such contribution occurs.

\begin{Lemma}
    \label{Lemma: Mismatched weights}
    If $|\sin \theta| < 1/L$, then the incoherent logical noise component $\tilde{\chi}_{\scriptscriptstyle{Z_1 Z_1}}$ can be written
    \begin{equation}
    \label{eq: mismatched weight lemma statement}
        \tilde{\chi}_{\scriptscriptstyle{Z_1 Z_1}} \geq \sum_{\mathcal{L}} \sum_{O_U} \frac{1}{|\{O_C^\prime\}|} \sum_{O_U^\prime} (O_U \rho \, O_U^\prime) \times \mathrm{Disconnected} \times \left(1 + \mathcal{O}\left( L \sin \theta)^{2\zeta} \right) \right) .
    \end{equation}
    The sum over $\mathcal{L}$ includes all typical, short logical strings with length $\ell$ such that $\ell \leq L+2 \zeta$. The sum over $O_U$ includes uncorrectable subsets of $\mathcal{L}$ with weight $(\ell+1)/2$ such that the complement $O_C$ has minimal weight. The sum over $O_U^\prime$ has the same syndrome and the same weight as $O_U$. The error term comes from the high-weight terms we neglected in Lemma \ref{Lemma: High Weigh Incoherent Terms}.
\end{Lemma}

\begin{proof}

    Using Lemma \ref{Lemma: High Weigh Incoherent Terms}, we can truncate the sum over noise terms in the incoherent logical noise component $\tilde{\chi}_{\scriptscriptstyle{Z_1 Z_1}}$ to include only those noise terms with total weight $\leq L+2 \zeta+1$. In doing so we make an error that is exponentially small in the cutoff $ \zeta$, assuming that the single-qubit angle of rotation $\theta$ satisfies $|\sin \theta|< 1/L$. We will use equation (\ref{eq: incoherent sum string form}) to express the incoherent logical noise components as a sum over strings. We will begin by reviewing how we construct that form of the sum.

    We denote the weight of $O_U$ by $w$ and $O_U^\prime$ by $w^\prime$. We upper bound the mismatched-weight terms where $w \neq w^\prime$ by letting $w>w^\prime$ and multiplying by two. As in equation (\ref{eq: incoherent sum string form}), we can generate the complete set of connected incoherent terms with fixed $w$ and $w^\prime$ by summing over connected logical strings $\mathcal{L}$. Denote the length of the logical string by $|\mathcal{L}| = \ell$. To produce the incoherent terms with fixed $w + w^\prime$, it will suffice to sum over logical strings with $\ell < w+w^\prime$. We already restricted to low-weight terms, so $w+w^\prime \leq L+2 \zeta+1$. For each logical string, we sum over the uncorrectable subsets $O_U$ with weight $w$. We will also require that the complement of $O_U$, which we denoted $O_C$, has minimal weight. This is to control the over-counting of each incoherent term $(O_U \rho \, O_U^\prime)$. Then for each $O_U$, we sum over the operators $O_U^\prime$ with weight $w^\prime$ that have the same syndrome and logical action as $O_U$. As discussed in Section \ref{sec: Incoherent sum over strings}, we must also divide by an over-counting factor $1/|\{O_C\}|$ that is a function of $O_U$ and equals one over the number of times $O_U$ appears in the sum over $\mathcal{L}$. The contribution to the incoherent logical noise is lower bounded by
    \begin{equation}
    \label{eq: incoherent contribution lower bound from w, w prime terms}
        \text{Contribution from string $\mathcal{L}$} \geq \sum_{O_U} \frac{|\{O_U^\prime\}|}{|\{O_C^\prime\}|} \left(\frac{\sin\theta}{2} \right)^{w+w^\prime}.
    \end{equation}
    The inequality comes from the cosine factors. If the operators $O_U$ and $O_U^\prime$ act on the same set of qubits, then we have $(\sin \theta/2)^{2w}$ with no cosine factors in the connected part. The lower bound corresponds to the case where $O_U$ and $O_U^\prime$ act on disjoint sets of qubits, and we pick up a cosine factor on each qubit in the connected part. We also have an upper bound:
    \begin{equation}
    \label{eq: incoherent contribution upper bound from w, w prime terms}
        \text{Contribution from string $\mathcal{L}$} \leq \sum_{O_U} \frac{|\{O_U^\prime\}|}{|\{O_C^\prime\}|} \left(\sin(\theta/2) \right)^{w+w^\prime}.
    \end{equation}
    In this bound, we have $\left(\sin(\theta/2) \right)^{w+w^\prime}$. This corresponds to the case when $O_U$ and $O_U^\prime$ act on the same qubits, yielding no cosine terms.
    
    Consider first the terms with $w=w^\prime$. These are generated from logical strings of length $\ell \leq 2w-1$. Some strings $\mathcal{L}$ with length $\ell$ such that $L \leq \ell \leq L+2 \zeta$ have typical shape and some do not. We will prove first that the contribution from a given string of atypical shape is no greater than that of a string of typical shape, in fact it will be less. We will conclude that we can safely neglect the contribution from strings of atypical shape (because there are fewer such strings). This is the same simplification we made in our discussion of the coherent logical noise components in Section \ref{sec: sum over partitions}.
    
    We require that $O_U$ is an uncorrectable error, so we cannot choose any subset of $w$ qubits in $\mathcal{L}$. We ignore the subsets that correspond to exceptional partitions like we discussed in Section \ref{sec: sum over partitions}. Now, if we have two connected logical strings of the same length, one with a typical shape and one with an atypical shape, we want to compare the terms with $w=w^\prime$. The first thing we notice is that exceptional terms are exponentially unlikely for the string with typical shape, while for the atypical string, exceptional terms may be a significant fraction of the total partitions. This tells us that in the sum over $O_U$ there are many more terms for the typical string than for the atypical string. We have argued about the number of terms in the sum in equations (\ref{eq: incoherent contribution lower bound from w, w prime terms}) and (\ref{eq: incoherent contribution upper bound from w, w prime terms}), but we must also consider the magnitude of each term, which is given by the ratio of $|\{O_U^\prime\}|$ over $|\{O_C^\prime\}|$. 
    
    We must argue that after summing the ratio of $\{O_U^\prime\}$ and $\{O_U^\prime\}$ over $O_U$, the result is less for an atypical string than for a typical string. $\{O_U^\prime\}$ here is the set of uncorrectable operators with the same weight and syndrome as $O_U$ and $\{O_C^\prime\}$ is the set of correctable weight-$\ell - w$ operators with the same syndrome. Suppose the logical string runs left to right across the code. The set $\{O_U^\prime\}$ contains more than one element whenever $O_U$ contains a set of contiguous qubits around one or more of the vertical steps in the logical string. This was discussed in detail in Section \ref{sec: Incoherent sum over strings}. $|\{O_U^\prime\}|$ and $|\{O_C^\prime\}|$ are large when either $O_U$ or $O_C$ contain contiguous sets of qubits around the vertical steps. The typical logical string has at least $\gamma \sqrt{L}$ horizontal steps around each of the vertical steps. The atypical string does not. This means that the typical string has more possible sets of qubits around each vertical step that make $|\{O_U^\prime\}|$ or $|\{O_C^\prime\}|$ large. Therefore, $|\{O_U^\prime\}|$ and $|\{O_C^\prime\}|$ will tend to be larger for the typical string. The ratio of $|\{O_U^\prime\}|$ to $|\{O_C^\prime\}|$ is what determines the contribution to the incoherent logical noise. In Lemma \ref{Lemma: Incoherent string sum bound}, we showed how we can match up terms such that for each $O_U$ such that $|\{O_U^\prime \}|/ |\{O_C^\prime\}| = c$, the partner has $|\{O_U^\prime \}|/ |\{O_C^\prime\}| \geq 1/c$. If $c$ is large, then $c+\frac{1}{c} \gg 2$. It follows that because the typical string has more operators $O_U$ in the sum and the $|\{O_U^\prime \}|/ |\{O_C^\prime\}|$ factors tend to be larger, the contribution to the incoherent logical noise is smaller for an atypical string than the contribution from a typical string. When we combine this fact with the fact the atypical strings represent an small minority, this means we can neglect the atypical strings among the $w=w^\prime$ terms in the incoherent logical noise. The error is given by Lemma \ref{Lemma: shape of strings B}.
    
    Now, consider the mismatched-weight terms that are the subject of this lemma. Fix $w+w^\prime$ and suppose $w > w^\prime$. For each $\mathcal{L}$ we can construct a number of incoherent terms with mismatched weight depending on the length and shape of $\mathcal{L}$. Let $|\mathcal{L}| = \ell$. Once again $O_U$ is an uncorrectable subset of the logical string $\mathcal{L}$ with weight $w$. Then, for each $O_U$ we have the possibility that there may exist an operator $O_U^\prime$ with the same syndrome as $O_U$ and lower weight. We sum over the set of $O_U$ such that for each $O_U$ there exists an $O_U^\prime$ with weight $w^\prime$. If we sum over all logical strings of length $<w+w^\prime$, we produce every connected incoherent term with $|O_U| = w$ and $|O_U^\prime| = w^\prime$. We will proceed by fixing a logical string and upper bounding the sum over $w$ of the noise terms derived from this logical string with $w+w^\prime$ fixed. In this sum the terms will alternate sign as $w$ increases. The terms with $w = w^\prime$ have a positive sign. As we seek to bound the contribution of these mismatched-weight terms to the incoherent logical noise, there are two things we need to bound. First, we must understand the combinatorics that govern the number of operators $O_U$ that permit lower-weight $O_U^\prime$. Second, we must bound the factor $|\{O_U^\prime \}|/|\{O_C^\prime\}|$ for each such $O_U$.
    
    Suppose that the logical string $\mathcal{L}$ has a typical shape. To be concrete, consider the set of operators $O_U$ with weight $w = (\ell+3)/2$. If there exist lower-weight $O_U^\prime$, then $O_U$ must contain all of the qubits around a ``cap", which is similar condition to the one we discussed in Section \ref{sec: sum over partitions}. Each cap has width at least $\gamma \sqrt{L}$ because the string is typical. This means that such $O_U$ are exponentially few relative to the total set of uncorrectable $O_U$ with weight $w$. This is the same calculation as in Lemma \ref{Lemma: exceptional terms}. We compare these terms to the terms with $|O_U| = (\ell+1)/2 =|O_U^\prime|$. There are exponentially more of these terms where $|O_U| = |O_U^\prime|$. This means that in equations (\ref{eq: incoherent contribution lower bound from w, w prime terms}) and (\ref{eq: incoherent contribution upper bound from w, w prime terms}) the sum over $O_U$ contains exponentially more terms when $w = w^\prime$ for a typical logical string. The summand also tends to be less for the $w > w^\prime$ terms. The argument is similar to the one we used earlier when we were discussing the $w = w^\prime$ terms from typical and atypical string. $|\{O_U^\prime\}|$ is large when $O_U$ contains many qubits around several of the different steps up and down along the logical string. In this case, the terms with $w>w^\prime$ feature operators $O_U$ that contain at least one of the ``caps" along the logical string. This removes at least two the vertical steps. These steps cannot contribute to $|\{O_U^\prime \}|$. Then by the argument we used above, the ratio $|\{O_U^\prime \}|/|\{O_C^\prime \}|$ tends to less for the $w>w^\prime$ terms relative to the $w=w^\prime$. We chose $w = (\ell+3)/2$, but we could have chosen any $w> (\ell+1)/2$ and any $w^\prime < w$. We would find that there are $2^{\gamma \sqrt{L} (w - w^\prime)}$ fewer of the mismatched weight terms. We have a factor of $2^{\gamma \sqrt{L}}$ for each cap contained in $O_U$. We conclude that the mismatched-weight terms are negligible for strings of typical shape.
    
    Finally, consider a logical string $\mathcal{L}$ with an atypical shape. Fix $w+w^\prime$. We already neglected the contribution of atypical strings to the $w=w^\prime$ terms. We seek a bound on the contribution from the terms with $w>w^\prime$ for atypical strings. We will compare two sets of terms for fixed $\mathcal{L}$ with length $\ell$. On the one hand, take the terms with $w= w_1$ for some $w_1 > (\ell-1)/2$ and $w^\prime = w_2 < w_1$. On the other hand, take the terms with $w = w_1 + 1$ and $w^\prime = w_2 -1$. We will show that the latter set of terms contribute less than the former. This will tell us that the sum over mismatched-weight terms for the fixed string $\mathcal{L}$ is bounded by the contribution from terms with $|O_U| = |O_U^\prime|$.

    When $O_U^\prime$ has lower weight than $O_U$, $O_U$ must contain all the qubits along a cap. If $O_U - O_U^\prime = 2j$, then $O_U$ must contain at caps with total height at least $j$. Because the logical string $\mathcal{L}$ has an atypical shape, these caps may have width one or height greater than one. It will not be exponentially unlikely that all qubits around a small cap are contained in $O_U$. For the terms with $w=w_1$ and $w^\prime = w_2$, relative to $O_U^\prime$, $O_U$ contains all the qubits around $w_1-w_2$ of the caps. These terms will be a fraction of the $\binom{\ell}{w_1}$ subsets of $w_1$ qubits in the logical string $\mathcal{L}$. We compare these terms to the ones with $w=w_1+1$ and $w^\prime = w_2-1$ keeping our logical string $\mathcal{L}$ fixed. These terms include all the qubits around an additional cap. On one of the caps, instead of containing at least one and less than all of the qubits, $O_U$ contains all of the qubits around that cap. This stricter condition of $O_U$ means that the fraction of the total $\binom{\ell}{w_1+1}$ weight $w_1$ subsets of $\mathcal{L}$ that feature an $O_U^\prime$ operator with weight $w_2-1$ is smaller than the fraction of the total $\binom{\ell}{w_1}$ weight $w_1$ subsets of $\mathcal{L}$ that feature an $O_U^\prime$ operator with weight $w_2$. This means that in equations (\ref{eq: incoherent contribution lower bound from w, w prime terms}) and (\ref{eq: incoherent contribution upper bound from w, w prime terms}) in the sum over $O_U$ for our fixed logical string the number of possible $O_U$ at a given weight $w > (\ell+1)/2$ is given by a binomial coefficient times a function that decreases monotonically as $w$ increases. As for the summand $|\{O_U^\prime \}|/|\{O_U^\prime \}|$, we apply the same reasoning as above. For each cap contained in $O_U$, there are fewer vertical steps to create many operators $O_U^\prime$. This implies that the summand will tend to be smaller as $w$ increases. The sum over the different values of $w$ has the form
    \begin{equation}
    \label{eq: sum over w monotonic function}
        \sum_{w=c}^{\ell} (-1)^w f(w) \binom{\ell}{w} < \binom{\ell}{c}  ,
    \end{equation}
    where $c\geq (l+1)/2$ and $f$ is a monotonically decreasing function. The inequality in equation (\ref{eq: sum over w monotonic function}) is proven by pairing the adjacent terms in the sum, positive and negative, to produce small positive contributions bounded by the contributions in the case where $f(w)=1$ for all $w$. It follows that the sum over the mismatched-weight terms derived from the logical string $\mathcal{L}$ is positive, and moreover is bounded by the $w=w^\prime$ terms. We already argued that the $w=w^\prime$ terms from atypical strings are negligible relative to those terms from typical strings. Finally, we can lower bound the incoherent logical noise component by neglecting the atypical strings and for the typical strings, neglecting the mismatched-weight terms. This yields equation (\ref{eq: mismatched weight lemma statement}).
\end{proof}

We are left with only the incoherent terms that have the same weight of uncorrectable error on each side and the weight is $\leq \frac{L+2 \zeta+1}{2}$. These terms all have the same phase $+1$, so the incoherent terms with different weights will add constructively. This gives us lower bounds on the logical incoherent noise strength. Each logical string with length $\ell \leq L+2 \zeta$ contributes at least $\binom{\ell}{\frac{\ell+1}{2}} \left(\frac{\sin \theta}{2}\right)^{\ell+1}$. When $\ell$ is much larger than the minimum logical string length, $L$, the number of logical strings is given by equation (\ref{eq: path counting formula}). 

In particular, the incoherent logical noise components must be larger than the lowest order term. This at last completes the argument begun in Section \ref{sec: Path Counting} about neglecting the contribution to the coherent logical noise from connected logical strings with length $>L+2\zeta$ for a cut-off constant $\zeta$. In Lemma \ref{Lemma: Coherent path counting}, we proved that the contribution from long logical strings is upper bounded by $\alpha L^{2 \zeta+1} |\sin \theta|^{L+2 \zeta}$, where $\alpha$ is $(1-L|\sin \theta|)^{-1}$. This bound is exponentially small in $\zeta$ relative to the lowest order incoherent logical noise component, $L \binom{L}{\frac{L+1}{2}} \left( \frac{\sin \theta}{2}\right)^{L+1}$. Our aim is to compare the logical coherent and incoherent noise components, and we have shown that the contribution to the coherent logical noise from long strings is small relative to the incoherent logical noise components. Therefore, we can safely neglect the long connected logical strings. The same applies for the truncation error in Lemma \ref{Lemma: High Weigh Incoherent Terms}. The truncation error is negligible relative to the lowest order incoherent terms for large enough $\zeta$.

\subsection{More General Rotation Angles}
\label{sec: more general rotation angles}

In Sections \ref{section: The Coherent Sum} and \ref{section: Incoherent Sum}, we simplified the problem by assuming that all qubits are rotated by the same single-qubit unitary rotation. Now we want to extend our result to more general single-qubit rotations. We will allow the magnitude of the rotation angle to vary from qubit to qubit and will also allow different axes of rotation for different qubits. Here we will assume that each rotation axis is contained in the $X{-}Z$ plane. Physical $Y$ errors are treated in Appendix \ref{app: Physical Y Errors}, where we prove that rotations partly along the $Y$-axis produce less coherent logical noise channels than those arising from rotations along axes in the $X{-}Z$ plane.

The idea of the proof is the same as that of Lemma \ref{lemma:minimization}. We will consider the coherent and incoherent logical noise components as functions of individual qubit rotation angles and prove the coherent component is maximized relative to the incoherent component when all rotation angles are equal.

\begin{Lemma}
    Consider the toric code with qubits subject to single-qubit rotations, where each rotation axis lies in the $X{-}Z$ plane, and both the rotation axis and angle of rotation may vary from qubit to qubit. The bound on the coherence of the logical noise channel proved in Theorem \ref{theorem: Big Theorem} continues to apply if the rotations are sufficiently close to uniform; that is, provided that each rotation axis and angle deviates from a fixed constant value within a bounded region.
    \label{Lemma: general rotation angles}
\end{Lemma}

\begin{proof}
    
    Suppose at first that all rotations are about the $Z$-axis and denote the rotation angle for the $i$th qubit by $\theta_i$. Each logical coherent or incoherent component is a sum of physical noise terms, which are functions of all the angles $\theta_i$. We will refer to the coherent or incoherent logical noise strength; by this we mean the sum of norms squared of the off-diagonal or diagonal components of the chi matrix for the logical noise channel. We are interested in the coherence of the logical noise channel, that is, the relative magnitude of the coherent and incoherent logical noise strength. Our approach will be to fix the coherent logical noise strength and calculate how the incoherent logical noise strength varies as we change rotation angles while remaining in the submanifold with constant coherent logical noise strength.
    
    We begin at a point where all single-qubit rotation angles are equal. Suppose that this rotation angle is $>0$. The proof will be similar if the angle is $<0$. We will perturb away from this point, moving along the submanifold with fixed coherent logical noise strength. These perturbations can be built out of small elementary steps, in which two qubits, $i$ and $j$, are selected. We require that $\theta_i \geq \theta_j$. Then, the elementary step consists of increasing $\theta_i$ by some amount and decreasing $\theta_j$ such that we remain on the submanifold with constant coherent logical noise strength. We will prove that such elementary steps increase the incoherent logical noise strength. Therefore, we will conclude that the coherence of the logical noise is maximized when all single-qubit rotation angles are equal. Our calculation will be limited to configurations of angles not too far from the point where all angles are equal.

    In Lemmas \ref{Lemma: Coherent path counting} and \ref{Lemma: High Weigh Incoherent Terms}, we proved that when all the rotation angles are equal and satisfy $|\sin \theta |<1/L$, the logical noise is dominated by the contributions of the low-weight connected terms. We bounded the absolute magnitude of the sum over high-weight connected terms. These high-weight terms were negligible relative to the low-weight connected terms in the incoherent logical noise. If the rotation angles are allowed to differ, so long as all the angles $\theta_i$ satisfy $|\sin \theta_i|<1/L$, our upper bound on the absolute magnitude of the error from the high-weight terms continues to hold. We require that this error is negligible relative to the low-weight terms we keep in the incoherent logical noise components. This was true when all angles are equal and will continue to be true for a wide range of configurations; only certain edge cases will violate this condition. For instance, one such edge case arises if all the rotation angles are $0$ except for the qubits along a long logical string with a shape such that it contains no low-weight uncorrectable subsets.

    We previously defined the connected and disconnected parts of a noise term (Definitions \ref{Def: Connected coherent}, \ref{Def: Disconnected Coherent}, \ref{Def: Connected incoherent}, and \ref{Def: Disconnected Incoherent}). As we described in Section \ref{sec: disconnected pieces} and Lemmas \ref{lemma: added error exceptional terms} and \ref{lemma: incoherent disconnected part}, the disconnected part has a value of $1$ up to corrections. These corrections are small for low-weight connected terms when all rotation angles are equal. If the angles are different, we can still apply our analysis, so long as the absolute error from the corrections is small relative to the low-weight connected terms in the incoherent logical noise components. This holds in a region around the point where all angles are equal. Hence, in this proof we will compare only the connected terms in the coherent and incoherent logical noise components with the understanding that the error terms we are neglecting are small relative to the low-weight connected terms we have kept.

    We can build any general perturbation out of an elementary (non-infinitesimal) perturbation where we increase one rotation angle $\theta_i$ and decrease a second angle, $\theta_j$, such that the connected contribution to the coherent logical noise strength is unchanged. The perturbation will look different depending on how the two qubits are positioned. If the qubits $i$ and $j$ are adjacent to each other and aligned in the correct direction, they will appear together in many short logical strings. Otherwise, $i$ and $j$ will not appear together in short logical strings. Throughout this section, we will approximate $\sin \theta_i/2 \approx \theta_i/2$ to simplify the equations. We will incur a relative error of $\theta_i^2/4$, that will always be small, since we have assumed that $|\sin \theta_i| < 1/L$ for every $i$. Then the contribution to the coherent logical noise strength from the low-weight connected terms as a function of $\theta_i$ and $\theta_j$ is
    \begin{equation}
        \label{eq: coherent two qubit expansion}
        \text{coherent} = \gamma_0 + \gamma_1 (\theta_i + \theta_j) + \gamma_2 \theta_i \theta_j .
    \end{equation}
    The coherent logical noise strength is a sum of norms squared and is therefore positive. This implies that $\gamma_0 > 0$. Moreover, the sum over partitions has the same phase for every short logical string as in equation (\ref{eq: Sum over partitions}). This means that each logical string makes a positive contribution to the noise strength. Therefore, $\gamma_1$ and $\gamma_2$ are both non-negative. The relative size of $\gamma_1$ and $\gamma_2$ depends on how close the two chosen qubits $i$ and $j$ are. When $i$ and $j$ are both along the same horizontal or vertical line, many low-weight logical strings will contain both qubits. These strings contribute to $\gamma_2$, so that the $\gamma_2$ term may be comparable to the $\gamma_1$ term. On the other hand, if qubits $i$ and $j$ are not along a horizontal or vertical line, then none of the minimal-weight logical strings contain both qubits. Also, for any fixed length $\ell \leq L+2 \zeta$, the number of logical strings of length $\ell$ that contain both qubits $i$ and $j$ is negligible relative to the number of length $\ell$ logical strings that contain qubit $i$ and not qubit $j$. In this case, the $\gamma_2$ term is negligible relative to the $\gamma_1$ term. In either case, we can write down the perturbation that leaves the coherent logical noise strength unchanged. Let $\theta_i = c_i \theta$ and $\theta_j = c_j \theta$ for some $\theta$, and then we will solve for $c_j$ such that the connected coherent sum is constant. This yields
    \begin{equation}
        \label{eq: Coherent sum level set condition}
        c_j = \frac{(2-c_i) \gamma_1 + \gamma_2 \theta}{\gamma_1 + c_i \gamma_2 \theta} ,
    \end{equation}
    so that when $\gamma_2 = 0$ we have $c_j = 2-c_i$.
    
    We can expand the incoherent logical noise strength in the same way. The noise terms that enter into the incoherent logical noise have the form $(O_U \rho \, O_U^\prime)$. As we expand in the angles $\theta_i$ and $\theta_j$, we have cases where the qubits $i$ and $j$ are contained in neither, one of, or both $O_U$ and $O_U^\prime$:
    \begin{equation}
        \label{eq: incoherent two qubit expansion}
        \text{incoherent} = \delta_0 + \delta_1 (\theta_i^2 + \theta_j^2) + \delta_2 \theta_i^2 \theta_j^2 + \delta_3 (\theta_i + \theta_j) +\delta_4 \theta_i \theta_j + \delta_5 (\theta_i^2 \theta_j + \theta_i \theta_j^2) .
    \end{equation}
     By Lemma \ref{Lemma: High Weigh Incoherent Terms}, the contributions of high-weight logical strings to the logical incoherent noise components are negligible. The contributions for each short logical string are positive due to Lemma \ref{Lemma: Mismatched weights}. If we fix a short logical string that contains both qubit $i$ and qubit $j$ and require that $O_U$ and $O_U^\prime$ both contain $i$ and $j$ or one of $i$ and $j$, then the same proof as in Lemma \ref{Lemma: Mismatched weights} implies that these contributions are positive. Therefore, the coefficients $\delta_1$ and $\delta_2$ are positive when all rotation angles are equal. We can now substitute the perturbation from equation (\ref{eq: Coherent sum level set condition}) into each of the terms in equation (\ref{eq: incoherent two qubit expansion}). We compute the perturbed value of the incoherent logical noise strength and subtract the initial value when all the angles of rotation are equal. Let $\textrm{incoherent}(a)$ denote the value of the incoherent term in equation (\ref{eq: incoherent two qubit expansion}) with $c_i=a$. The difference between the perturbed and initial values is
    \begin{align}
    \label{eq: incoherent c_i expansion}
        \textrm{incoherent} ( c_i ) - \textrm{incoherent}(1) & = \delta_1 \frac{(c_i-1)^2 ( 2 \gamma_1^2 \theta^2 + 2 (2+c_i) \gamma_1 \gamma_2 \theta^3 + (c_i^2+2 c_i +1) \gamma_2^2 \theta^4) }{(\gamma_1 + c_i \gamma_2 \theta)^2} \nonumber
        \\
        &\quad - \delta_2 \frac{(c_i-1)^2 \left( \gamma_1^2 \theta^4 (2 - (c_i-1)^2) + 2 c_i \gamma_1 \gamma_2 \theta^5 \right) }{(\gamma_1 + c_i \gamma_2 \theta)^2} \nonumber
        \\
        &\quad + \delta_3 \frac{(c_i-1)^2 (\gamma_1 \gamma_2 \theta^2 + c_i \gamma_2^2 \theta^3)}{(\gamma_1 + c_i \gamma_2 \theta)^2} \nonumber
        \\
        &\quad - \delta_4 \frac{(c_i-1)^2(\gamma_1^2 \theta^2 + c_i \gamma_1 \gamma_2 \theta^3)}{(\gamma_1 + c_i \gamma_2 \theta)^2} \nonumber
        \\
        &\quad - \delta_5 \frac{(c_i-1)^2 (2 \gamma_1^2 \theta^3 + c_i^2 \gamma_1 \gamma_2 \theta^4 - c_i \gamma_2^2 \theta^5)}{(\gamma_1 + c_i \gamma_2 \theta)^2} .
    \end{align}
    We see that the $\delta_1$ term is positive for all $c_i > 1$. We will show that the positive terms are larger than the other terms for all $c_i>1$. Each term has the same denominator and contains a factor of $(c_i-1)^2$ in the numerator. This means immediately that the first derivative with respect to $c_i$ vanishes at the point $c_i = 1$. We pull out the shared factors of $\frac{(c_i-1)^2}{(\gamma_1 + c_i \gamma_2 \theta)^2}$ and rearrange terms in equation (\ref{eq: incoherent c_i expansion}):
    \begin{align}
    \label{eq: perturbation rewritten}
        & = \frac{(c_i-1)^2}{(\gamma_1 + c_i \gamma_2 \theta)^2} \times \nonumber
        \\
        & \quad \left( \gamma_1^2 (2 \theta^2 \delta_1 - 2 \theta^4 \delta_2 - \theta^2 \delta_4 - 2 \theta^3 \delta_5) \right. \nonumber
        \\
        &\quad \quad + \gamma_1 \gamma_2( 4 \theta^3 \delta_1 + \theta^2 \delta_3) \nonumber
        \\
        &\quad \quad + \gamma_1 \gamma_2 (2 c_i \theta^3 \delta_1 - 2 c_i \theta^5 \delta_2 - c_i \theta^3 \delta_4) \nonumber
        \\
        &\quad \quad + \gamma_2^2 (2 c_i \theta^4 \delta_1 + c_i \theta^3 \delta_3 + c_i \theta^5 \delta_5) \nonumber
        \\
        &\quad \quad + \left.\gamma_1^2 (c_i-1)^2 \theta^4 \delta_2 + \gamma_2^2 (c_i^2+1)\theta^4 \delta_1 - \gamma_1 \gamma_2 c_i^2 \theta^4 \delta_5 \right) .
    \end{align}
    If the conditions,
    \begin{align}
    \label{eq: delta conditions}
        2 \theta^2 \delta_1 &> 2 \theta^4 \delta_2 + \theta^2 |\delta_4| + 2 \theta^3 |\delta_5|, \nonumber
        \\
        4 \theta^3 \delta_1 &> \theta^2 |\delta_3|, \nonumber
        \\
        \theta^4 \delta_1&> \theta^3 |\delta_3| + \theta^5 |\delta_5| ,\nonumber
        \\
        \textrm{and} \quad  \gamma_2^2 (c_i^2 +1) \theta^4 \delta_1  + \gamma_1^2 (c_i-1)^2 \theta^4 \delta_2 &> \gamma_1 \gamma_2 c_i^2 \theta^4 |\delta_5| ,
    \end{align}
    are satisfied, then each line of equation (\ref{eq: perturbation rewritten}) is greater than $0$. Recall that $\gamma_1$, $\gamma_2$, $\delta_1$, and $\delta_2$ are non-negative. Each elementary perturbation increases the value of $c_i$. Therefore, if the conditions in equation (\ref{eq: delta conditions}) are satisfied, then each elementary step along the submanifold with constant coherent noise strength increases the incoherent logical noise strength. It remains for us to argue that these conditions are satisfied when the rotation angles are close to equal.
    
    Consider two cases for the relative positions of qubits $i$ and $j$. In the first case, suppose qubits $i$ and $j$ are positioned so that no short logical strings contain both. Then the strings that contribute to $\gamma_2$ as well as $\delta_2$, $\delta_4$, and $\delta_5$ are all long. Such strings do not appear in the sum over low-weight connected terms, so $\gamma_2 = 0$ in equation (\ref{eq: perturbation rewritten}). Therefore, the only condition is the first line of equation (\ref{eq: delta conditions}). In this inequality $\delta_2$, $\delta_4$ and $\delta_5$ are $0$, and the condition is satisfied.
    
    In the other case, qubits $i$ and $j$ are in a horizontal or vertical line so that both qubits are contained in several short logical strings. In this case $\theta \gamma_2$ is comparable to $\gamma_1$. Now consider the incoherent contribution from the strings that contain both qubits $i$ and $j$. For each weight-$2w-1$ logical string containing both qubits $i$ and $j$, the errors $O_U$ are weight-$w$ subsets of the logical string. Nearly half will contain exactly one of qubits $i$ and $j$ and one quarter will contain both qubits $i$ and $j$. This means that more terms contribute to $\theta^2 \delta_1$ and $\theta \delta_3$ than to $\theta^4 \delta_2$, $\theta^2 \delta_4$, and $\theta^3 \delta_5$. In particular, $\theta^2 \delta_1 > 2 \theta^4 \delta_2$. For each $O_U$, the set $\{O_U^\prime\}$ contains operators with the same syndrome, logical action and weight. $O_U^\prime$ differs from $O_U$ only near certain transverse steps along the logical string as we described in Section \ref{sec: Incoherent sum over strings}. For most logical strings, if $O_U$ does not contain qubit $i$, then most of the operators $O_U^\prime$ also will not. If $O_U$ contains qubit $i$, most operators $O_U^\prime$ will as well. This implies that 
    \begin{align}
        \theta^2 \delta_1 &\gg \theta |\delta_3|, \nonumber
        \\
        \theta^4 \delta_2 &\gg \theta^2 |\delta_4| ,\nonumber
        \\
        \theta^4 \delta_2 & \gg \theta^3 |\delta_5| , \nonumber
        \\
        \textrm{and} \quad  \theta^2 \delta_1 &\gg \theta^3 |\delta_5|.
    \end{align}
    Together with our earlier statement that $\theta^2 \delta_1 > 2 \theta^4 \delta_2$, this implies that even when $i$ and $j$ are near to each other, the conditions in equation (\ref{eq: delta conditions}) are satisfied. 
    
    We have proven that each of the elementary perturbations starting from uniform angles increases the incoherent logical noise strength. However, we must also consider elementary perturbations that are applied after a different elementary perturbation has taken us away from uniform angles. In that case, in equations (\ref{eq: coherent two qubit expansion}) and (\ref{eq: incoherent two qubit expansion}) we would no longer have exact symmetry between $\theta_i$ and $\theta_j$. In other words, equation (\ref{eq: coherent two qubit expansion}) would read 
    \begin{equation}
        \textrm{coherent} = \gamma_0 + \gamma_1^{(i)} \theta_i + \gamma_1^{(j)} \theta_j + \gamma_2 \theta_i \theta_j ,
    \end{equation}
    where $\gamma_1^{(i)}$ and $\gamma_1^{(j)}$ are different coefficients. However, as long as we are not too far from uniform angles, $\gamma_1^{(i)} \approx \gamma_1^{(j)}$, and the change in the incoherent logical noise strength will be almost the same as in equation (\ref{eq: incoherent c_i expansion}). We have argued that the conditions in equation (\ref{eq: delta conditions}) are satisfied by a large margin. Therefore, there exists a region around the point with uniform angles where every elementary perturbation on the submanifold with fixed coherent logical noise strength increases the incoherent logical noise strength. We conclude that in a region around the symmetric point where every single-qubit rotation is by the same angle $\theta$, this symmetric point gives the largest coherent logical noise strength relative to the incoherent logical noise strength. This means that the connected contribution to the incoherent logical noise strength has a local minimum at the point with uniform rotations within the submanifold with constant connected contribution to the coherent logical noise strength. This implies that the upper bound on logical coherence we derive in Theorem \ref{theorem: Big Theorem} for the case where all angles are equal also upper bounds the logical coherence in a region around the point where all angles are equal.
    
    Now consider changing the axes of rotation while keeping the total rotation angle the same. Let the noise model be a rotation in the $X{-}Z$ plane such that the rotation angles are $\theta_X$ and $\theta_Z$ for every qubit. We show that $Y$ rotations on the qubits will produce less coherent logical noise in Lemma \ref{lemma: Y rotations}. The $X$ and $Z$ rotations contribute to independent components of the logical noise channel. Then, each noise term that enters into the $X$-type logical noise components depend on at least $L$ powers of $\theta_X$. Similarly, the $Z$-type logical noise depends on at least $L$ powers of $\theta_Z$. Therefore, if the total rotation angle for each qubit, $\sqrt{\theta_X^2 + \theta_Z^2}$, is fixed, the logical noise strength is greatest when either $\theta_X$ or $\theta_Z$ is $0$. Also, because the $X$ and $Z$-type errors contribute to different logical noise components, we can apply our analysis of coherent and incoherent logical noise components to the two types of errors separately. If both $|\sin \theta_X|$ and $|\sin \theta_Z|$ are $< 1/L$, then Lemmas \ref{Lemma: Coherent path counting} and \ref{Lemma: High Weigh Incoherent Terms} imply that the logical coherent and incoherent $X$ and $Z$ noise is dominated by the contributions of short logical strings. With this noise model, we will in general expect logical $Y$ noise, but Lemma \ref{Lemma: Other logical maps} implies that logical $Y$-type errors are negligible relative to $X$ and $Z$-type errors. The $\theta_X$ rotations contribute to the $X_1$ and $X_2$ logical noise, while the $\theta_Z$ rotations contribute to the $Z_1$ and $Z_2$ logical noise. The bounds we proved on coherent and incoherent logical noise components apply equally well to the $X$ and $Z$-type noise separately in this noise model.
\end{proof}

Lemma \ref{Lemma: general rotation angles} states that among noise models consisting of single-qubit rotations where each rotation is close to the same, the coherence of the logical channel is greatest for the noise model consisting of $Z$-axis rotations on every qubit by the same angle. The same does not necessarily hold for noise models consisting of single-qubit rotations with wildly different angles of rotation on each qubit. This is not surprising because if we allow for wildly different rotation angles, we encounter the case where all the rotation angles are $0$ except for the qubits along some very long logical string. This kind of high-weight connected term is beyond the scope of the present work, \textit{cf.} Lemmas \ref{Lemma: Coherent path counting} and \ref{Lemma: High Weigh Incoherent Terms}.

\subsection{Correlations}
\label{sec: toric code with correlations}

We can apply Theorem \ref{Theorem: Correlated noise} to study the toric code with minimal-weight decoding subject to correlated unitary noise. In the repetition code, we found that adding two-body correlations did not change the relation between coherent and incoherent components of the logical noise when the code size $n$ is large. We can transfer that to the toric code using what we have already proven. Consider a single logical string. We can sum over its partitions. With correlated unitary noise, instead of $\sin \theta/2 \cos \theta/2$ for each qubit in the logical string, we have a sum over the one and two-body couplings in the Hamiltonian, $h_1$ and $h_2$. The model of correlated noise that we considered in Theorem \ref{Theorem: Correlated noise} included two-body coupling terms between every pair of qubits in the code. Therefore, the magnitude of each multi-qubit error is a function only of its weight. We found that the coherent and incoherent logical noise in the toric code is dominated by the contributions from short logical strings with typical shape. In Theorem \ref{theorem: Big Theorem}, we will finish proving a relation between the coherent and incoherent terms that is based on the number of terms and their magnitudes, which we always assumed to be $\sin \theta/2$ raised to the weight of the terms. In the correlated case we alter the magnitude of each term, but Theorem \ref{Theorem: Correlated noise} tells us that string by string we can bound the coherent logical noise contribution in terms of the incoherent logical noise contribution.

\subsection{Main Theorem}
\label{sec: Main Theorem}

\begin{Theorem}
    \label{theorem: Big Theorem}
    Consider the $L \times L$ toric code without boundaries subject to single-qubit unitary noise acting on each qubit. We chose minimal-weight decoding and assume that syndrome extraction is perfect. Suppose that each qubit is rotated by an angle $\theta$ about some axis and that $|\sin \theta|<1/L$. Our conclusion will also hold for angles and axes that differ among the qubits, so long as the deviation is small as discussed in Lemma \ref{Lemma: general rotation angles}. Let $\mathcal{\tilde{N}}$ be the logical noise channel produced by encoding into the toric code, acting with single-qubit unitary noise, and then decoding. Denote by $\tilde{\chi}$ the chi matrix for the logical noise channel $\mathcal{\tilde{N}}$. Then, the coherent and incoherent components of the logical channel are related by
    \begin{equation}
        \sum_{i,j| i \neq j} \left| \tilde{\chi}_{i,j} \right|^2 < \frac{2}{(\sin \theta)^2} \left( \sum_{l|l \neq 0} \tilde{\chi}_{l,l} \right)^2 (1+ \mathcal{E}) ,
        \label{eq: coherent and incoherent final statement}
    \end{equation}
    where 
    \begin{equation}
    \label{eq: Epsilon error term upper bound}
        |\mathcal{E}| \leq  \frac{24 \gamma \zeta^2}{\sqrt{L}} + \mathcal{O} \left( \frac{1}{L} \right) + \mathcal{O} \left((L \sin \theta)^{2\zeta} \right),
    \end{equation}
    and $\zeta$ is an arbitrary $L$-independent constant.
    We denote the diamond norm distance of $\mathcal{\tilde{N}}$ from the identity channel by $D_\Diamond (\mathcal{\tilde{N}})$. It follows that
    \begin{equation}
        \label{eq: Main theorem diamond distance}
        D_\Diamond (\mathcal{\tilde{N}}) \leq c r ,
    \end{equation}
    for a constant $c$ given by
    \begin{equation}
        c^2 = \frac{(d_L+1)^2}{2} \left(1 + \frac{1}{(\sin \theta)^2} \right) (1 + \mathcal{E}),
    \end{equation}
    where $d_L$ is the dimension of the code space ($d_L=4$ for the 2D toric code without boundaries). Here $r$ is the average infidelity of the logical noise channel $\mathcal{\tilde{N}}$, and  $\mathcal{E}$ is the error term bounded in equation (\ref{eq: Epsilon error term upper bound}). (If the logical noise channel $\mathcal{\tilde{N}}$ were unitary, then $D_\Diamond (\mathcal{\tilde{N}})$ would be proportional to $\sqrt{r}$.) We can also consider the growth of the average infidelity as we apply the logical noise channel many times in succession. Let $r_m$ be the average infidelity after $m$ applications of $\mathcal{\tilde{N}}$; then using equation (\ref{eq: coherent and incoherent final statement}), we can write
    \begin{equation}
    \label{eq: Main theorem growth of infidelity}
        r_m = m r \left(1 + \frac{d_L}{(d_L+1) \sin^2 \theta} (m-1) r (1+ \mathcal{E}) \right) ,
    \end{equation}
   where $d_L = 4$ for the 2D toric code without boundaries, $\mathcal{E}$ is the error term that is upper bounded in equation (\ref{eq: Epsilon error term upper bound}), and $r$ is the average infidelity for a single application of the logical noise channel. As long as the physical noise strength is below the fault-tolerant threshold, $r$ will be exponentially small in the code distance $L$. Therefore, equation (\ref{eq: Main theorem growth of infidelity}) states that the term growing quadratically in $m$ is exponentially small in $L$ relative to the term growing linearly with $m$. In this sense, the coherence of the logical channel is heavily suppressed. 
\end{Theorem}

\begin{proof}

    We start with a noise model consisting of $Z$ rotations by angle $\theta$ on every qubit in the $L \times L$ block of toric code. We seek to approximate the coherent logical noise component $\tilde{\chi}_{\scriptscriptstyle{Z_1 I}}$ and relate it to the incoherent logical noise component $\tilde{\chi}_{\scriptscriptstyle{Z_1 Z_1}}$. First, let us calculate the coherent component. We write $\tilde{\chi}_{\scriptscriptstyle{Z_1 I}}$ as a sum over strings and partitions with a connected and disconnected part as in equation (\ref{eq: coherent sum string form}). Applying Lemma \ref{Lemma: Coherent path counting} to the connected part, we neglect high-weight terms, leaving only the logical strings with length $\leq L+2\zeta$ for a fixed $\zeta$, and making an error which is exponentially small in $\zeta$. For this step, the magnitudes of the sines of the single-qubit rotation angles are required to be below a threshold value $1/L$. We apply Lemma \ref{lemma: added error exceptional terms} to argue that disconnected part is equal to $1$ up to a small correction. Lemmas \ref{Lemma: shape of strings A} and \ref{Lemma: shape of strings B} tell us that we can treat all short logical strings $\mathcal{L}$ as typical and make another small error. Now that we have only short typical logical strings, we apply Lemma \ref{Lemma: exceptional terms} to perform the sum over partitions. We conclude that
    \begin{equation}
    \label{eq: coherent component final}
        \tilde{\chi}_{\scriptscriptstyle{Z_1 I}} = \sum_{\mathcal{L}} i^\ell \binom{\ell-1}{\frac{\ell-1}{2}} \left( \frac{\sin \theta}{2}\right)^{\ell} \left( 1 + \mathcal{E}_{shape} \right) + \mathcal{E}_{Long} ,
    \end{equation}
    where the sum is over all connected logical $Z_1$ strings $\mathcal{L}$ with length $\ell$ such that $\ell \leq L+2\zeta$. The error from logical strings with length greater than $L+2 \zeta$ is bounded 
    \begin{equation}
        \left| \mathcal{E}_{Long} \right| \leq \alpha L^{2\zeta+1} (\sin \theta)^{L+2\zeta},
    \end{equation}
    where $\alpha = (1-L \sin \theta)^{-1}$. This error is from Lemma \ref{Lemma: Coherent path counting}. The other error term from from Lemmas \ref{Lemma: shape of strings B} and \ref{lemma: added error exceptional terms} is bounded
    \begin{equation}
            \left|\mathcal{E}_{shape} \right| \leq \frac{16 \gamma \zeta^2}{\sqrt{L}} + O(1/L) .
        \end{equation}
    A further error due to neglecting exceptional partitions is subdominant according to Lemma \ref{Lemma: exceptional terms}, and is not shown in equation (\ref{eq: coherent component final}).
    
    Now, we will lower bound the incoherent logical incoherent noise component $\tilde{\chi}_{\scriptscriptstyle{Z_1 Z_1}}$. Lemma \ref{Lemma: High Weigh Incoherent Terms} implies that we can neglect the contributions of noise terms $(O_U \rho \, O_U^\prime)$ such that $|O_U| + |O_U^\prime| >L+2\zeta+1$. The error we make by truncating the sum is exponentially small in $\zeta$, so long as the rotation angle $\theta$ satisfies $|\theta| < 1/L$. The incoherent logical noise component can be put in the form of a sum over logical strings as in equation (\ref{eq: incoherent sum string form}). Using Lemma \ref{Lemma: Mismatched weights}, we restrict to the connected terms where the logical string $\mathcal{L}$ is short and has typical shape and $|O_U| = |O_U^\prime|$. We can also keep only the terms with $|O_U| = (\ell+1)/2$, because just as in the repetition code, these higher-weight partitions are suppressed by factors of $\sin \theta^2$ and the binomial coefficients are decreasing as we consider higher-weight $O_U$. The disconnected part is equal to $1$ up to small correction according to Lemma \ref{lemma: incoherent disconnected part}. Finally, Lemma \ref{Lemma: Incoherent string sum bound} gives a lower bound on the contribution of each logical string to the incoherent logical noise. All together, we have the following lower bound on $\tilde{\chi}_{\scriptscriptstyle{Z_1 Z_1}}$:
    \begin{equation}
    \label{eq: incoherent component final}
        \tilde{\chi}_{\scriptscriptstyle{Z_1 Z_1}} \geq \sum_{\mathcal{L}} \binom{\ell}{\frac{\ell+1}{2}} \left(\frac{\sin \theta}{2} \right)^{\ell+1} \left( 1 - \frac{8 \gamma \zeta^2 }{\sqrt{L}} + \mathcal{O}(1/L) + \mathcal{O} \left((L \sin \theta)^{2\zeta} \right) \right) .
    \end{equation}
    The error terms come from Lemmas \ref{Lemma: shape of strings B}, \ref{Lemma: High Weigh Incoherent Terms}, and \ref{lemma: incoherent disconnected part}. 
    A subdominant error term from Lemma \ref{Lemma: exceptional terms} is suppressed in equation (\ref{eq: incoherent component final}).
   
    Putting together equation (\ref{eq: coherent component final}) and equation (\ref{eq: incoherent component final}), we conclude the following about how the coherent and incoherent terms in the logical chi matrix are related:
    \begin{equation}
    \label{eq: coherent incoherent ratio}
        \frac{|(\tilde{L_a} \tilde{\rho})|}{|(\tilde{L_a} \tilde{\rho} \tilde{L_a})|} \leq \frac{1}{|\sin \theta|} (1 + \mathcal{E}) ,
    \end{equation}
    where
    \begin{equation}
    \label{eq: relative error bound for ratio}
        |\mathcal{E}| \leq \left( \frac{24 \gamma \zeta^2}{\sqrt{L}} \right) + \mathcal{O} \left( \frac{1}{L} \right) + \mathcal{O} \left( (L \sin \theta)^{2\zeta}  \right) .
    \end{equation}
    We used
    \begin{equation}
        \binom{\ell}{\frac{\ell-1}{2}} \approx 2 \binom{\ell-1}{\frac{\ell-1}{2}}
    \end{equation}
    to arrive at equation (\ref{eq: coherent incoherent ratio}).
    
    We have restricted our attention to the coherent logical component $(\tilde{L_a} \tilde{\rho})$ and the corresponding incoherent component $(\tilde{L_a} \tilde{\rho} \tilde{L_a})$, where $L_a$ was either $X$ or $Z$ on one of the two encoded qubits. We did this because these are the largest components in the logical noise. This is proven in appendices \ref{app: X1X2 or Y1 terms} and \ref{app: a and b nontrivial}. In Lemma \ref{Lemma: Other logical maps}, we prove that we can neglect the logical noise components $(\tilde{L_a} \tilde{\rho})$ where $L_a$ is a logical $Y$ or a non-trivial operator on both encoded qubits. In Lemma \ref{Lemma: More general coherent terms}, we prove that we can neglect coherent terms of the form $(\tilde{L_a} \tilde{\rho} \tilde{L_b})$ with $a \neq b$. Using these results, we can bound the sum of all coherent logical terms relative to the sum of all incoherent logical terms. There are two off-diagonal terms for each diagonal term, e.g. $\tilde{\chi}_{\scriptscriptstyle{Z_1 I}}$ and $\tilde{\chi}_{\scriptscriptstyle{I Z_1}}$ are matched with $\tilde{\chi}_{\scriptscriptstyle{Z_1 Z_1}}$, so we have
    \begin{equation}
        \sum_{i,j|i \neq j} \left| \tilde{\chi}_{i,j} \right|^2 < \frac{2}{(\sin \theta)^2} \left( \sum_{k| k \neq 0} \tilde{\chi}_{k,k} \right)^2 (1+\mathcal{E}) .
        \label{eq: total coherent and incoherent}
    \end{equation}
    We have proven equation (\ref{eq: coherent and incoherent final statement}). The term on the right hand side is proportional to the infidelity by equation (\ref{eq: infidelity identity}). Going back to Lemma \ref{lemma:off-diagonal} and equation (\ref{eq:coherence-angle-bound}), we can write
    \begin{equation}
        r_m \leq m r + m (m-1) \frac{d_L}{2(d_L +1)} \sum_{i,j| i \neq j} \tilde{\chi}_{i,j}^2 ,
    \end{equation}
    where $d_L$ is the dimension of the physical Hilbert space, which is $4$ for the toric code without boundaries. We can combine this with equation (\ref{eq: total coherent and incoherent}).
    \begin{equation}
        r_m \leq m r \left(1 + \frac{d_L}{(d_L+1) \sin^2 \theta} (m-1) r (1+\mathcal{E}) \right).
    \end{equation}
    Finally, we can use Lemma \ref{Lemma: diamond distance bound} with equation (\ref{eq: total coherent and incoherent}) to derive equation (\ref{eq: Main theorem diamond distance}).
    
    So far we have considered a noise model consisting of the same $Z$ rotation by angle $\theta$ on every qubit in the code block. We can use Lemma \ref{Lemma: general rotation angles} and Lemma \ref{lemma: Y rotations} to prove that this noise channel produces maximally coherent logical noise in a region around uniform rotations. The single-qubit rotation angles are allowed to differ so long as the deviation is not too great. Therefore, the relation we found between coherent and incoherent logical noise components for the $Z$ rotation noise model bounds the coherence of the logical noise channel for small rotations about any axis, so long as the rotations are close to uniform across the qubits.
\end{proof}

There are some subtleties in the interpretation of Theorem \ref{theorem: Big Theorem}. We address these in the next subsection, but first we will make a remark about the error bound in equation (\ref{eq: Epsilon error term upper bound}). This error bound is satisfactory for finite code size, $L$; however, we will need make a small modification before the bound is suitable for the $L\rightarrow \infty$ limit. This is because the term $\mathcal{O}\left( (L \sin \theta)^{2\zeta} \right)$ in equation (\ref{eq: Epsilon error term upper bound}) contains a factor polynomial in $L$. If the single qubit rotation angle $\theta$ satisfies $|\sin \theta| \propto 1/L$, then this polynomial factor would make the truncation error large as $L \rightarrow \infty$. The polynomial factor comes partly from the fact that the truncation error in equation (\ref{eq: coherent component final}) has a factor of $2^L$ relative to the factor of $\binom{L}{\frac{L+1}{2}}$ that appears in the lowest-weight incoherent noise terms. The ratio is proportional to $\sqrt{L}$. The other contribution to this polynomial factor comes from the path counting in Lemma \ref{Lemma: High Weigh Incoherent Terms}, where we neglected a factor polynomial in $w+w^\prime$ in equation (\ref{eq: bound on number of terms of high-weight incoherent}). We can cancel the polynomial factor by slightly modifying our truncation procedure. Denote this polynomial factor $p(L)$. Instead of neglecting noise terms with weight $>L+2\zeta$ in the coherent logical noise components, we neglect noise terms with weight $>L+2 \lceil \zeta^\prime \log (L) \rceil $ for a constant $\zeta^\prime$. We perform a similar truncation for the incoherent logical noise components. Then, we can choose $\zeta^\prime$ large enough that $(L \sin \theta)^{2\zeta^\prime} p(L)$ is decreasing with $L$. The minimum value for $\zeta^\prime$ such that this is decreasing depends on the degree of $p(L)$ and the magnitude of $L \sin \theta$. If $\zeta^\prime$ is greater than this minimum value, then $(L \sin \theta)^{2\zeta^\prime} p(L)$ is bounded above by $a |L \sin \theta|^{\lambda \zeta^\prime \log(L) }$, where $a$ is a constant that is determined by the coefficients in the polynomial $p(L)$, and $\lambda$ is a constant that depends on $\zeta^\prime$, the degree of $p(L)$, and the magnitude of $L \sin \theta$. Our new truncation rule slightly alters the other error terms in equation (\ref{eq: relative error bound for ratio}). The new error term is
\begin{equation}
\label{eq: new error term main theorem}
    |\mathcal{E}| \leq \frac{24 \gamma \zeta^{\prime 2} (\log L)^2}{\sqrt{L}} + \mathcal{O} \left(\frac{1}{L} \right) + a |L \sin \theta|^{\lambda \zeta^\prime \log(L) } ,
\end{equation}
where $a$, $\zeta^\prime$, and $\lambda$ are constants. $a$ and $\lambda$ are determined as we described above. We are free to choose $\zeta^\prime$, so long it is greater than a minimum value. Now, if we take the limit $L \rightarrow \infty$, we find that the error term in equation (\ref{eq: new error term main theorem}) remains small. Therefore, we may apply Theorem \ref{theorem: Big Theorem} in the limit of $L \rightarrow \infty$ with equation (\ref{eq: new error term main theorem}) replacing equation (\ref{eq: Epsilon error term upper bound}).

\subsection{Interpreting Bounds on Coherence}

We have proved a relation between the diagonal and off-diagonal components of the chi matrix of the logical noise channel. The interpretation is a bit subtle, so it is worth commenting on here. 
We upper bounded the off-diagonal components by $1/|\sin \theta|$ times the diagonal components, and we were forced to assume that $|\sin \theta| < 1/L$ because our analysis only applies to logical strings with length $\leq L+2\zeta$ where $\zeta$ is a constant. With this assumption, the factor of $1/|\sin \theta|$ implies that the coherent component of the logical channel may be  $L$ times larger than the incoherent component. This might seem to indicate that the coherence of the logical channel is not suppressed for large $L$, but that is not the best way to think about the comparison.

In equation (\ref{eq: Main theorem growth of infidelity}), the term quadratic in $m$ has a coefficient proportional to $r/\left(\sin \theta\right)^2$ relative to the linear term. But the average infidelity $r$ is exponentially small in $L$. Thus the coefficient of the quadratic term is really exponentially smaller in $L$ relative to the coefficient of the linear term. In equation (\ref{eq: Main theorem diamond distance}), the constant $c^2$ is not really a constant, since it scales like $L^2$ if $\theta$ scales like $1/L$. The point is that if the logical noise channel were fully coherent, i.e. unitary, then $c$ would scale like $1/\sqrt{r}$, but we find that $1/\sqrt{r}$ scales like $L^{L/2}$, which is vastly greater than $L^2$. We conclude that, although the logical noise channel is not exactly incoherent, it is quite close to an incoherent channel  as measured by our statements about the growth of average infidelity and the relation between diamond distance and average infidelity.

We could also consider writing our logical noise channel as a product of a unitary rotation and a Pauli channel. We can solve for the single parameter in each of these two channels. In the limit of low logical noise strength, the angle of rotation of the unitary channel approximately equals one of the off-diagonal chi matrix elements, and the probability of error in the Pauli channel is comparable to one of the diagonal components of the chi matrix. Theorem \ref{theorem: Big Theorem} implies that the logical channel can be written as a product of a unitary channel and a Pauli channel where the angle of rotation of the unitary is larger than the error probability of the Pauli channel by a factor which is approximately $1/|\sin\theta|$, and therefore enhanced by a factor of $L$ if $\theta$ scales like $1/L$. Again, this might make it seem like the coherence is not suppressed; however, the coherent channel makes a contribution to the average infidelity proportional to the rotation angle squared. This is why we find that the growth of average infidelity becomes nearly linear in $m$ as the code size $L$ increases. As the code block becomes large, the diamond distance for the logical noise channel is much smaller than what one would expect for a coherent channel based on the value of the average infidelity $r$. This is another way of making the same point as in the previous paragraph.

One might wonder whether a tighter upper bound than Theorem \ref{theorem: Big Theorem} can be derived on the strength of the coherent part of the logical channel relative to the incoherent part. 
In fact, a substantially tighter upper bound is not possible, if we want this bound to hold for arbitrary small rotation angles. For instance, we could choose to set every rotation angle equal to zero except for the qubits along a single length $L$ logical string. For this case, the computation of the logical channel is similar to our computation for the repetition code, where we were able to compute the logical channel quite precisely. Alternatively, for a fixed code size we could choose sufficiently small uniform rotation angles $\theta$ such that the lowest-weight terms dominate in the logical noise. In this case the computation of the logical channel is again similar to that of the repetition code. Since the bound we proved for the toric code nearly matches what we found for the repetition code, we know that our result is optimal in this special case. Of course, for some other particular set of single-qubit rotations, the logical noise channel may be less coherent than our upper bound predicts.

\section{Conclusions}
\label{sec:conclusions}

We have studied characterizations of coherence in quantum channels. One useful method for diagnosing the coherence of a channel $\mathcal{N}$ is to consider applying $\mathcal{N}$ $m$ times in succession, and to investigate how the average infidelity $r$ of the composite channel $\mathcal{N}^m$ increases with $m$.  For incoherent channels $r$ is linear in $m$, while for highly coherent channels it can grow quadratically with $m$. Another useful diagnostic is provided by the relationship between $r$ and $D_\Diamond(\mathcal{N})$, the distance between $\mathcal{N}$ and identity channel as measured by the diamond norm. For incoherent channels this distance scales linearly with $r$, while for highly coherent channels it scales like $\sqrt{r}$.

Using these criteria we have investigated the coherence properties of \emph{logical} channels. To define a logical channel, we choose a particular quantum error-correcting code and decoding method; then we consider encoding an initial input state, subjecting the physical qubits to a noise model, and finally applying the decoder to obtain the channel's output. Our main conclusion is that, for the code families we examined, even if the physical noise model is highly coherent, the coherence of the logical channel is heavily suppressed in the limit of a large code block.

For the case of the quantum repetition code, we can compute the logical channel precisely, and verify that the logical channel is highly incoherent for large block size. Most of this paper was devoted to the analysis of a more challenging case, the $L\times L$ two-dimensional toric code subject to independent unitary noise. Our main conclusion about this case is encompassed by Theorem \ref{theorem: Big Theorem}. Regrettably, for the case of the toric code we were able to prove that the coherence  of the logical channel is suppressed only under an unrealistic assumption:  that as the size $L$ of the code block increases, the rotation angle $\theta$ applied to each qubit scales like $1/L$. 

Under this assumption, we can estimate the logical channel well enough for our purposes by expanding it to a constant ($L$-independent) order in $\theta$, and argue that the higher-order terms we ignore make a contribution that can be safely neglected. A key step in our argument is the observation, backed up by Lemmas \ref{Lemma: shape of strings A} and \ref{Lemma: shape of strings B}, that, for $L | \sin \theta| < 1$, the logical channel is dominated by logical strings with an easily characterized typical shape. For the logical strings of this typical shape, Lemmas \ref{Lemma: exceptional terms}, \ref{Lemma: Mismatched weights}, \ref{lemma: added error exceptional terms}, and $\ref{lemma: incoherent disconnected part}$ provide a sufficiently accurate estimate of the logical channel to prove Theorem \ref{theorem: Big Theorem}.

Our main conclusion, that the coherence of the logical coherence is heavily suppressed, applies to unitary physical noise such that each qubit is rotated independently, even if the rotation axis and rotation angle vary from qubit to qubit, as long as the rotations are close to the same and sufficiently small. It also applies for some highly correlated noise models. The result also extends to physical noise channels which are convex combinations of unitary channels, or convex combinations of unitary channels and depolarizing channels. (Depolarizing physical noise is mapped to an incoherent depolarizing logical channel under error correction.)

We emphasize that our result is an asymptotic statement in the limit of large code size $L$, albeit under the assumption that the noise strength scales like $1/L$. For codes of fixed size our results may not be tight; the coherence of logical channels for finite code blocks has been studied elsewhere \cite{bravyi2017correcting,Iyer17, greenbaum2017modeling}. Our goal was to study a family of codes with an accuracy threshold instead. When the noise is below threshold, the logical channel approaches the identity as the code block increases in size. In addition, under conditions where Theorem \ref{theorem: Big Theorem} applies, the coherent component of the logical channel vanishes much more rapidly than the incoherent component.

It is reasonable to expect that our conclusion --- that the logical channel becomes increasingly incoherent as $L$ grows --- continues to hold even if we allow $L$ to increase while the rotation angle $\theta$ has a fixed constant value. But proving this will be challenging. For one thing, if $\theta$ is a constant we cannot accurately estimate the logical channel by expanding to a constant order in $\theta$. Instead, logical strings with length $\leq L(1+\beta)$ need to be included for some constant $\beta$. These logical strings are not easy to count. A logical string can be regarded as a self-avoiding walk on the square lattice whose endpoints are a distance $L$ apart, but previously derived upper bounds on the number of self-avoiding walks with specified length \cite{Hammersley1960,Guttman1984,GuttmanConway2001} do not treat the case where the distance between the endpoints differs from the length of the walk by an $\mathcal{O}(1)$ multiplicative factor. And even if we could count the logical strings accurately, we would still need to overcome some additional obstacles to prove the coherence of the logical channel is suppressed. 

First, to prove Theorem \ref{theorem: Big Theorem}, we disposed of the ``exceptional'' terms (Definition \ref{Def: Exceptional term}), those in which the uncorrectable error on a logical string has lower weight than the correctable error, by arguing that these terms are sufficiently rare as to make a negligible contribution to the coherent part of the logical channel. But for logical strings with length $L(1+\beta)$, exceptional terms will be far more common.

Second, when we calculated the contribution to the coherent or incoherent logical noise, we separated the computation into a sum over a connected part and a disconnected part, and argued  in Lemmas \ref{lemma: added error exceptional terms} and \ref{lemma: incoherent disconnected part} that the disconnected part contributes a multiplicative factor close to 1. But the proofs of these lemmas required the logical strings to be short, of length $L+2\zeta$ for constant $\zeta$; these proofs don't apply for longer logical strings of length $L(1+\beta)$.

Third, and even more dauntingly, our proof of Theorem \ref{theorem: Big Theorem} made use of a relationship between the physical noise terms that contribute to the coherent and incoherent logical noise. But as the logical string length increases, the contributions to the coherent and incoherent component of the logical channel become less and less alike. Each contribution to the coherent logical noise is associated with a logical string. In contrast, each contribution to the incoherent logical noise is associated with a pair of logical strings; these strings have segments in common, but they fluctuate relative to one another apart from those shared segments. For short logical strings, these fluctuations are relatively mild, and did not prevent us from relating the incoherent and coherent logical noise, as described in Section \ref{sec: Incoherent sum over strings}. For longer logical strings, the combinatorics become much harder to handle. 

Unable to overcome these obstacles ourselves, we settled for proving a weaker result that applies for $L | \sin \theta|< 1$ rather than constant $\theta$. Perhaps a more ambitious combinatoric analysis can push the proof through even for constant $\theta$. Or perhaps a completely different approach will be more successful. Conceivably, it's not true that the coherence of the logical channel becomes heavily suppressed for large $L$ and sufficiently small constant $\theta$, though we consider that possibility unlikely. 

Further numerical studies of logical coherence may also prove to be instructive. The problem has already been studied numerically \cite{bravyi2017correcting, suzuki2017efficient, Gutierrez16, Iyer17}; however, our methods for organizing the estimate of the logical channel suggest different approaches to numerically simulating the logical channel. Numerics could help to resolve the issues that prevented us from extending Theorem \ref{theorem: Big Theorem} to the case where $\theta$ is an $L$-independent constant. 

In our analysis of the toric code subject to single-qubit unitary rotations, we used minimal-weight decoding because it can be systematically analyzed. 
However, we don't expect our conclusion about suppression of logical coherence to be very sensitive to the choice of decoding method. The suppression arises from averaging over many error syndromes, and therefore should occur for other families of stabilizer codes with good decoders. Many of the elements from which we built the proof of Theorem \ref{theorem: Big Theorem} can be applied to more general stabilizer codes, including ``logical strings,'' partitions, exceptional terms, and the decomposition into connected and disconnected parts.

We analyzed the toric code because it has an accuracy threshold, and we aspired to study the coherence of the logical channel for a fixed nonzero value of $\theta$ as the linear size $L$ of the code block gets large. That aspiration eluded us, so we settled for investigating the logical coherence in the regime $L |\sin \theta| < 1$. In that regime, asymptotic results similar to ours, derived using similar methods, may be applicable to other code families that do not have an accuracy threshold. For example, for the Bacon-Shor code family subjected to depolarizing noise with error probability $p$, the optimal logical failure probability, computed analytically in \cite{napp2013optimal}, is achieved by the code with distance $d= \mathcal{O}(1/p)$. We anticipate that, for unitary noise, decoding the optimal Bacon-Shor yields a logical channel with strongly suppressed coherence, though we have not done a careful analysis. 

It has long been suspected that error correction suppresses the coherence of noise. Such suppression had been observed numerically for the toric code \cite{bravyi2017correcting}, but no rigorous argument supporting this conclusion had been previously known for any code family with an accuracy threshold. Our goal in this project was to prove that, for the toric code subject to sufficiently weak independent or weakly correlated unitary noise, the logical channel after decoding is highly incoherent in the limit of a large code block. We fell short of this goal, settling for a proof that coherence is suppressed in the case where the noise strength decreases as the code block grows. Nevertheless, we hope and expect that the tools we have developed will prove to be useful in future studies of quantum error correction.

\section*{Acknowledgments}
\addcontentsline{toc}{section}{\protect\numberline{}Acknowledgments}
We would like to thank Michael Beverland, Robin Blume-Kohout, Benjamin Brown, Aaron Chew, Andrew Darmawan, Andrew Doherty, Steven Flammia, Daniel Gottesman, Tomas Jochym-O'Connor, Aleksander Kubica, Richard Kueng, David Poulin, and Leonid Pryadko for valuable discussions. We gratefully acknowledge support from ARO-LPS (W911NF-18-1-0103) and  NSF (PHY-1733907).
The Institute for Quantum Information and Matter is an NSF Physics Frontiers Center.

\bibliographystyle{unsrt}
\bibliography{IversonPreskillVersion2}

\appendix

\section{Chi matrix and Pauli transfer matrix for qubits}
\numberwithin{equation}{section}
\setcounter{equation}{0}
\label{appendix:off-diagonal}
Here we verify Lemma \ref{lemma:off-diagonal} for qubits by expressing all non-diagonal terms in $N_{kl}$ in terms of $\chi_{ij}$ explicitly:
\begin{align}
I &\mapsto \left(\chi_{XI} + \chi_{IX}\right)X + \left(\chi_{YI} + \chi_{IY}\right)Y + \left(\chi_{ZI} + \chi_{IZ}\right)Z \notag\\
&+ \left(\chi_{XY} - \chi_{YX}\right) (iZ) + \left(\chi_{ZX} - \chi_{XZ}\right) (iY) + \left(\chi_{YZ} - \chi_{ZY}\right) (iX),\notag\\ \notag\\
X &\mapsto \left(\chi_{XY} + \chi_{YX}\right)Y + \left(\chi_{XZ} + \chi_{ZX}\right)Z + \left(\chi_{IX} + \chi_{XI}\right) I  \notag\\
&+ \left(\chi_{ZY} - \chi_{YZ}\right)(iI) +\left(\chi_{IY} - \chi_{YI}\right) (iZ)  + \left(\chi_{ZI} - \chi_{IZ}\right)(iY),\notag\\ \notag\\
Y &\mapsto \left(\chi_{XY} + \chi_{YX}\right)X + \left(\chi_{ZY} + \chi_{YZ}\right)Z + \left(\chi_{IY} + \chi_{YI}\right) I  \notag\\
&+ \left(\chi_{ZX} - \chi_{XY}\right)(iI) +\left(\chi_{IZ} - \chi_{ZI}\right) (iX)  + \left(\chi_{XI} - \chi_{IX}\right)(iZ),\notag\\ \notag\\
Z &\mapsto \left(\chi_{ZY} + \chi_{YZ}\right)Y + \left(\chi_{XZ} + \chi_{ZX}\right)X + \left(\chi_{IZ} + \chi_{ZI}\right) I  \notag\\
&+ \left(\chi_{YX} - \chi_{XY}\right)(iI) +\left(\chi_{IX} - \chi_{XI}\right) (iY)  + \left(\chi_{YI} - \chi_{IY}\right)(iX).
\end{align}
When we collect all the terms in $\sum_{a\ne b} N_{ab}^2$ which are quadratic in $\{\chi_{XY},\chi_{YX}\}$, we obtain
\begin{equation}
2\left(\chi_{XY} + \chi_{YX}\right)^2 - 2\left(\chi_{XY} - \chi_{YX}\right)^2 = 8|\chi_{XY}|^2= 4\left(|\chi_{XY}|^2 + |\chi_{YX}|^2\right),
\end{equation}
using $\chi_{ij} = \chi_{ji}^*$, as required by complete positivity.
The same applies to the terms involving $\chi_{IX},\chi_{IY}, \chi_{IZ}, \chi_{ZX}, \chi_{YZ}$, and their complex conjugates. 

To prove the claim we must verify that the linear terms cancel. This can be shown using the general argument in Lemma \ref{lemma:off-diagonal}, but in the qubit case it may be easier to verify the cancellation explicitly. For example, the contributions to $N_{ab}$ involving $\chi_{IX}, \chi_{XI}, \chi_{YZ}, \chi_{ZY}$ are
\begin{align}
N_{IX} &= \left(\chi_{XI} + \chi_{IX}\right) + i\left(\chi_{YZ} - \chi_{ZY}\right) + \cdots,\notag\\
N_{XI} &= \left(\chi_{XI} + \chi_{IX}\right) - i\left(\chi_{YZ} - \chi_{ZY}\right) + \cdots, \notag\\
N_{YZ} &= \left(\chi_{ZY} + \chi_{YZ}\right) + i\left(\chi_{XI} - \chi_{IX}\right)+\cdots,\notag\\
N_{ZY} &= \left(\chi_{ZY} + \chi_{YZ}\right) - i\left(\chi_{XI} - \chi_{IX}\right)+\cdots,\notag\\
\end{align}
and we therefore see that the cross terms cancel in $N_{IX}^2 +N_{XI}^2$ and in $N_{YZ}^2 +N_{ZY}^2$. Similar cancellations occur for all other cross terms. 

\section{Approximating sums}\label{appendix:sum}
We wish to evaluate the sum in equation (\ref{eq:eps-delta-all-orders}):
\begin{equation}
P_n(p) =\sum_{w=0}^{(n-1)/2} \binom{n}{w} p^{n-w}\left(1-p\right)^{w=}\sum_{w=(n+1)/2}^{n}    \binom{n}{w} p^w (1-p)^{n-w},
\end{equation}
where $p=s^2 = \sin^2\theta/2$, and  $(1{-}p) = c^2 = \cos^2\theta/2$. Note that $P_n(p)$ is the probability of a decoding error for the $n$-bit repetition code subject to independent noise with bit-flip probability $p$. It is convenient to redefine the summation index obtaining
\begin{equation}
P_n(p)
=p^{(n+1)/2}(1-p)^{(n-1)/2}\sum_{r=0}^{(n-1)/2} \binom{n}{\frac{n+1}{2}+r}\left(\frac{p}{1-p}\right)^{r}.
\end{equation}
From the Stirling approximation, we have
\begin{equation}
\binom{n}{\frac{n+1}{2}+r} \approx \left(\sqrt{\frac{2}{\pi n}}\right) 2^n\exp\left(-\frac{2}{n}\left(r+\frac{1}{2}\right)^2\right),
\end{equation}
neglecting a multiplicative $\left( 1+ \mathcal{O}(1/n)\right)$ correction. Making another $\left(1+\mathcal{O}(1/n)\right)$ multiplicative error, we may replace the exponential inside the sum over $r$ by 1, obtaining
\begin{equation}
P_n(x) \approx p^{(n+1)/2}(1-p)^{(n-1)/2}\left(\sqrt{\frac{2}{\pi n}}\right) 2^n\sum_{r=0}^{(n-1)/2}\left(\frac{p}{1-p}\right)^{r},
\end{equation}
and we also make a negligible error (assuming $p < \frac{1}{2}$) by extending the upper limit on the sum to infinity, finding
\begin{equation}
\sum_{r=0}^{\infty}\left(\frac{p}{1-p}\right)^{r}= \frac{1-p}{1-2p}.
\end{equation}
We conclude that
\begin{equation}
P_n(x)= \left(\sqrt{\frac{2}{\pi n}}\right)\left(\frac{\sqrt{p(1-p)}}{1-2p}\right)\left[4p(1-p)\right]^{n/2}\left(1 + \mathcal{O}\left(\frac{1}{n}\right)\right),
\end{equation}
assuming $p < \frac{1}{2}$. Using
\begin{equation}
    4p(1-p) = \left( 2sc\right)^2 = \sin^2\theta, \quad 1-2p = c^2 - s^2 = \cos\theta,
\end{equation}
we find
\begin{equation}
P_n(x)=\frac{1}{\sqrt{2\pi n}}
\left(\frac{\sin^{n+1}\theta}{\cos\theta} \right)\left(1 + \mathcal{O}\left(\frac{1}{n}\right)\right),\quad \textrm{for} \quad \sin^2\theta/2 < 1/2.
\end{equation}


\section{Correlated Noise: Leading behavior for large $n$}
\label{app: Correlated noise}

Here we'll describe an alternative way of understanding equation (\ref{eq:Omega-q-correction}), where we saw that the coefficient of $h_2^q$ in the logical channel is $\mathcal{O}(m^{3q/2})$. This leading behavior results from cancellations of terms higher order in $m$ which occur when we perform the sum over $k_R$ in equation (\ref{eq:sum-Omega-Delta}).  What is the explanation for these cancellations?

In Section \ref{sec: Correlated unitary noise} we calculated the coherent and incoherent logical components for the bit-flip code of size $n$ subject to correlated unitary rotations given by a Hamiltonian of the form:
\begin{equation}
    H = \sum_i h_1 X_i + \sum_{j,k| j \neq k} h_2 X_j X_k .
\end{equation}
We expressed the logical coherent component $\tilde{\chi}_{\scriptscriptstyle XI}$ and the logical incoherent component $\tilde{\chi}_{\scriptscriptstyle XX}$ in terms of functions $\Omega$ and  $\Delta$ such that
\begin{align}
    \tilde{\chi}_{\scriptscriptstyle XI}(q) &= \sum_{k_R=0}^q \Omega(q-k_R, k_R) h_2^q h_1^{n-2q}, \nonumber
    \\
    \tilde{\chi}_{\scriptscriptstyle XX}(q)& = \sum_{k_R=0}^q \Delta(q-k_R, k_R) h_2^q h_1^{n+1-2q} + \mathcal{O}(h_2^q h_1^{n+3-2q}),
\end{align}
where 
\begin{equation}
\tilde{\chi}_{\scriptscriptstyle XI}=\sum_{q=0}^{(n-1)/2} \tilde{\chi}_{\scriptscriptstyle XI}(q), 
\quad  
\tilde{\chi}_{\scriptscriptstyle XX}=\sum_{q=0}^{(n+1)/2} \tilde{\chi}_{\scriptscriptstyle XX}(q), 
\end{equation}
and only even values of $q$ contribute.
Here $k_R$ is the number of times the Hamiltonian term $i h_2 XX$ acts on the density operator from the right, and $k_L = q-k_R$ is the number of times $-i h_2 XX$ acts from the left. 

We were able to compute $\Omega$ and $\Delta$ by counting the ways of decomposing each physical noise term into combinations of one and two-body Hamiltonian terms. Repeating equations (\ref{eq: Omega}) and (\ref{eq: Delta}), we found
\begin{align}
\label{eq: Omega 2}
    \Omega(q-k_R,k_R) &=  \frac{(i)^{n-q} (-1)^{m} (m+1) }{ (n-2q) 2^q} \binom{n}{m}  \nonumber
    \\
    &\quad \times  (-1)^{k_R} \frac{(m!)^2}{k_R! (q-k_R)! (m-2k_R)! (m-2q+2k_R)!},
\end{align}
and
\begin{align}
\label{eq: Delta 2}
    \Delta(q-k_R,k_R) & =  \frac{(i)^{q}}{2^q} \binom{n}{m}  \nonumber
    \\
    &\quad  \times  (-1)^{k_R} \frac{ ((m+1)!)^2}{(m+1-2q+2k_R)! (m+1-2k_R)! k_R! (q-k_R)!} .
\end{align}

Let's evaluate the sum over $k_R$ to leading order in $1/m$ in both equation (\ref{eq: Omega 2}) and (\ref{eq: Delta 2}). We focus on the second factor in each equation, which contains all of the $k_R$ dependence. In equation (\ref{eq: Omega 2}) this factor is
\begin{equation}
\label{eq: kR term from Omega}
    (-1)^{k_R} \frac{(m)(m-1) \dots (m-2k_R+1) \times (m)(m-1)\dots (m-2q+2k_R+1)}{(k_R)! (q-k_R)!}.
\end{equation}
The dominant term for $m$ large is given by
\begin{equation}
    (-1)^{k_R} \frac{(m)^{2q}}{(k_R)! (q-k_R)!} = (-1)^{k_R} \frac{(m)^{2q}}{q!}\times \binom{q}{k_R} .
\end{equation}
Then when we sum over $k_R$ we have
\begin{equation}
    \frac{(m)^{2q}}{q!} \sum_{k_R=0}^q (-1)^{k_R}\binom{q}{k_R} = 0,
\end{equation}
where we have made use of the identity
\begin{equation}
\label{eq: binomial alternating sum identity}
    \sum_{b=0}^a P_c(b) (-1)^b \binom{a}{b} = 0 \quad \forall c<a.
\end{equation}
Here $P_c(b)$ denotes any  polynomial in $b$ of degree $c< a$. The situation for $\Delta$ is the same except that $m$ is replaced by $m+1$. Therefore, the leading term for $m$ large in equations (\ref{eq: Omega 2}) and (\ref{eq: Delta 2}) vanishes.

Similar cancellations occur for higher-order corrections which are suppressed by powers in $1/m$. These corrections are computed by expanding the numerator, equation (\ref{eq: kR term from Omega}), as a $2q$th order polynomial in $m$. For example, the coefficient of $m^{2q-1}$ is
\begin{equation}
   (-1)^{k_R} (-1) \frac{2q^2-4q k_R + q + 4 k_R^2}{q!} \binom{q}{k_R},
\end{equation}
in which the prefactor multiplying $\binom{q}{k_R}$ is a second-degree polynomial in $k_R$, so that equation (\ref{eq: binomial alternating sum identity}) implies that the sum over $k_R$ vanishes if $q > 2$. Likewise, the coefficient of $m^{2q - r}$ is a polynomial in $k_R$ of degree of $2r$, and the sum over $k_R$ vanishes if $2r < q$.
Recalling that only even $q$ contribute, we see that the leading term that survives the summation over $k_R$ has $r= q/2$ and is therefore order $m^{2q-r}= m^{3q/2}$. We have now seen why terms higher order in $m$ cancel. The term of order $m^{3q/2}$ can be evaluated using the identity
\begin{equation}\label{binomial-alternating-power-a}
    \sum_{b=0}^a (-1)^b b^a \binom{a}{b} = (-1)^a a!.
\end{equation}
The identities in equations (\ref{eq: binomial alternating sum identity}) and (\ref{binomial-alternating-power-a}) can be derived by performing the binomial expansion of $(1+x)^a$, differentiating repeatedly, and then setting $x=-1$.

\section{The Shape of the Logical String}
\label{app: Shape of string}
In this appendix, we prove that among short logical strings nearly all have typical shape as in Definition \ref{Def: typical shape}.

\begin{Lemma}
In a size $L$ toric code, all but order $1/L$ of the logical strings running left to right across the code with length $\leq L+2\zeta$ consist of single steps up and down, so that no vertical segment is longer than one qubit.
\label{Lemma: shape of strings A}
\end{Lemma}
\begin{proof}
If the size of the code is $L$ and we consider all length $L+2\zeta$ logical strings for fixed $\zeta$, we will count the number of strings that satisfy the condition that each step up or down is only length one. First, we start with a horizontal logical string of length $L$ and then pick $\zeta$ sites along it. We have $\zeta$ upward steps and $\zeta$ downward steps, and we need to fix an ordering. Alternatively, we could think of choosing $\zeta$ sites for the upward steps and another $\zeta$ sites for the downward steps. In total the number of strings of this type is
\begin{equation}
    \label{eq: single steps up and down}
    \text{number of strings with steps of one} = \binom{L}{2 \zeta} \binom{2 \zeta}{\zeta} = \binom{L}{\zeta} \binom{L-\zeta}{\zeta} .
\end{equation}
The $L$ dependence in equation (\ref{eq: single steps up and down}) is 
\begin{equation}
    \label{eq: single steps up and down L dependence}
    \frac{L!}{(L-2\zeta)!} .
\end{equation}
Next, we will count the total number of strings that consist of no backwards steps, that is starting from the left of the code block the strings move only right, up, and down. These strings potentially contain upward and downward steps of more than one. In general, such a string involves $q_1$ distinct steps up with $\zeta$ total length and $q_2$ steps down also totaling $\zeta$ in length. The number of ways of writing $\zeta$ as a sum of $q_1$ terms, not ignoring order, is given by the number of compositions of the integer $\zeta$ into $q_1$ terms, which is $\binom{\zeta-1}{q_1-1}$. Each of the $q_1$ steps up and $q_2$ steps down can be placed independently. This gives us $\binom{L}{q_1} \binom{L-q_1}{q2}$ combinations of possible configurations. In total we have
\begin{equation}
    \label{eq: large steps up and down}
    \text{number of strings with $q_1$ steps up and $q_2$ steps down} = \binom{\zeta-1}{q_1-1} \binom{L}{q_1} \binom{\zeta-1}{q_2-1} \binom{L-q_1}{q_2}
\end{equation}
such strings. When $q_1 = q_2 = \zeta$, we recover the case where each step up or down is by one lattice site. Then, we can isolate the $L$ dependence in equation (\ref{eq: large steps up and down}):
\begin{equation}
    \frac{L!}{(L-q_1-q_2)!} .
\end{equation}
We can compare this to equation (\ref{eq: single steps up and down L dependence}), and we see that there are fewer paths with steps larger than one. The ratio is proportional to 
\begin{equation}
    \frac{\text{number with $q_1$ steps up and $q_2$ steps down}}{\text{number with $\zeta$ single steps up and down}} = \mathcal{O}(L^{q_1+q_2-2 \zeta}) .
\end{equation}
Then, if we count the paths with a single step of two and the other steps are all one, there are order $1/L$ of these relative to the number of paths with single steps up and down.

We must also count the the number of logical strings where the string backtracks on itself. There are even fewer of these than the strings with jumps up or down by two. Each string with backtracking can be produced from a string with a jump up or down of at least two lattice spacings. We add some additional cap onto the vertical segment of at least length 2. The number of strings with one instance of backtracking like figure \ref{fig: backtracking} will be proportional to the number of strings of length two shorter that also have at least one step up or down of more than one. For this reason, strings like the one in figure \ref{fig: backtracking} are an exponentially smaller minority that the strings with steps up and down of more than one. Then, we conclude that nearly all short logical strings spanning the code left to right consist of steps up and down by only one qubit.

\begin{figure}
    \centering
    \includegraphics[height = 7cm]{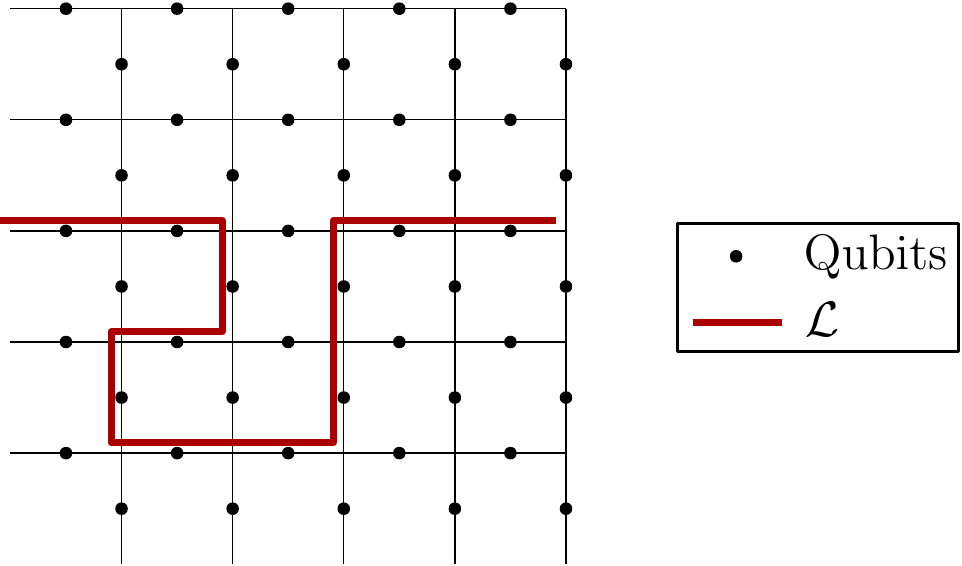}
    \caption{The logical string $\mathcal{L}$ has backtracking. Among short logical strings, those with backtracking are unlikely relative to strings without.}
    \label{fig: backtracking}
\end{figure}

\end{proof}

\begin{Lemma}
For the class of length $L+2 \zeta$ strings described in Lemma \ref{Lemma: shape of strings A} (those with exactly $\zeta$ steps up and $/zeta$ steps down as they span the code block from left to right), for large $L$ nearly all will have spacings between the steps growing proportional to $\sqrt{L}$. We choose a small constant $\gamma$ and define typical strings as those for which all vertical steps are separated by at least $\gamma \sqrt{L}$. If we fix the length of logical strings and combine this lemma with Lemma \ref{Lemma: shape of strings A}, we can make the following statement about the fraction of strings of that length that have atypical shape:
\begin{equation}
    \frac{\text{Number of strings with atypical shape}}{\text{Total number of strings}} = \frac{8 \gamma \zeta^2}{\sqrt{L}} + \mathcal{O}\left(\frac{1}{L}\right) .
\end{equation}
\label{Lemma: shape of strings B}
\end{Lemma}

\begin{proof}
The total number of strings of the type in Lemma \ref{Lemma: shape of strings A} is $\binom{L}{2\zeta} \binom{2\zeta}{\zeta}$. Now let us compute the number of strings such that each step up or down is separated from others by $\gamma \sqrt{L}$ for some constant $\gamma$. We can lower bound the number by starting with a length $L$ string running left to right across the code and placing our steps up and down. We suppose that each step we place prohibits placing another step on a further $2 \gamma \sqrt{L}$ of the sites. This is a lower bound because in the true answer these intervals will sometimes overlap. The lower bound is 
\begin{equation}
    \frac{L \left( L- (2 \gamma \sqrt{L} +1) \right) \left( L- 2 (2 \gamma \sqrt{L} +1) \right) \cdots \left( L- (2\zeta-1) (2 \gamma \sqrt{L} +1) \right)}{(2\zeta)!} \binom{2\zeta}{\zeta} .
\end{equation}
Compared to the total, $n-i$ has been replaced by $n-i(2 \gamma +1)$ for each $i$, so that when $\gamma = 0$ we recover the total number of strings. In general, the ratio of this limited set to the total for fixed $\zeta$ and $\gamma$ is given by
\begin{equation}
    \frac{\text{number of length $L+2\zeta$ strings with widely separated steps}}{\text{number of length $L+2\zeta$ strings}} \approx \prod_{i=1}^{2\zeta-1} \left( 1- \frac{2 i \gamma}{\sqrt{L}} \right) .
\end{equation}
We can lower bound this by 
\begin{equation}
    > \left( 1- \frac{4\zeta \gamma}{\sqrt{L}} \right)^{2\zeta} = 1- \frac{8 \gamma \zeta^2}{\sqrt{L}} + \mathcal{O}\left(\frac{1}{L}\right) .
\end{equation}
This approaches $1$ as $L$ increases, and we see that with high probability a short logical string will have the property that the steps up and down are separated by more than $\gamma \sqrt{L}$, as $L$ becomes large.

\end{proof}

\section{Disconnected Errors}
\label{app: Disconnected errors}

Fix a coherent logical noise component and consider the sum in equation (\ref{eq: coherent sum string form}). In Section \ref{sec: disconnected pieces} we argued that the disconnected term is 1 for disconnected errors that do not change how a given connected term is decoded. This allows us to write the sum as
\begin{equation}
\label{eq: Coherent sum with error term}
    \tilde{\chi}_{\scriptscriptstyle{Z_1, I}} = \left( \sum_{\mathcal{L}} \sum_{p \in \mathcal{P}(s)} \text{Connected part} \right) + \textrm{Error} .
\end{equation}
The sum over $\mathcal{L}$ includes all typical short connected logical strings. $\mathcal{P}(\mathcal{L})$ is the set of likely partitions of connected logical string $\mathcal{L}$. This excludes the partitions we called ``exceptional terms" in Definition \ref{Def: Exceptional term} and Lemma \ref{Lemma: exceptional terms}. ``$\textrm{Error}$" contains all the terms we have neglected. This includes the contribution of long logical strings, short logical strings with atypical shape, and exceptional terms. It also includes the terms with disconnected pieces that we did not consider in Section \ref{sec: disconnected pieces}. These are all of the terms where the disconnected errors flip the way the partition is decoded, where we start with a partition and after adding disconnected errors to each side, the error that was originally uncorrectable becomes correctable and vice versa. These terms will not follow the analysis we did in Section \ref{sec: disconnected pieces}. We will describe these terms now and show that they are negligible in the following lemma.

\begin{Lemma}
    In equation (\ref{eq: Coherent sum with error term}) the error from the neglected terms $\mathcal{E}$ can be expressed 
    \begin{equation}
        \textrm{Error} = \mathcal{E}_1 + \mathcal{E}_2 ,
    \end{equation}
    where $\mathcal{E}_1$ contains the contributions that we have already proven are negligible---long connected logical strings, logical strings with an atypical shape, and exceptional partitions. $\mathcal{E}_2$ contains the contributions from terms where the disconnected errors have flipped the way the partition is decoded. These are the terms we neglected in Section \ref{sec: disconnected pieces}. The following is true:
    \begin{equation}
        |\mathcal{E}_2| \le | \mathcal{E}_1| .
    \end{equation}
    \label{lemma: added error exceptional terms}
\end{Lemma}
\begin{proof}

    We start with a typical short connected logical string and take a partition into a correctable operator and an uncorrectable operator, denoted $(O_U \rho \, O_C)$. Now we add disconnected errors, $D_L$ and $D_R$ to the left and right side of the partition. In some cases the uncorrectable error may become correctable and vice versa. That is, $O_U D_L$ will be correctable, while $O_C D_R$ is uncorrectable. For example, the term that contributes to the $\tilde\chi_{\scriptscriptstyle Z_1I}$ component of the logical noise might be $(O_C D_R \rho \, O_U D_L)$. Our treatment of the disconnected part in Section \ref{sec: disconnected pieces} assumed that the added errors did not flip the correctable and uncorrectable sides of the original partition. Now we will justify this assumption by proving that such terms are negligible.
    
    First, we must understand the conditions when an added error will turn the uncorrectable side of a partition into a correctable error. $O_U$ is the uncorrectable side of the partition, so the minimal-weight correction to $O_U$ is equal to $O_C$ up to stabilizers. For $O_U D_L$ to be correctable requires that the minimal-weight correction is equal to $O_U D_L$ up to stabilizers and not equal to $O_C D_L$ up to stabilizers. Note that we write $D_L$ and not $D_R$ because $D_L$ and $D_R$ have the same syndrome, so as far as the decoder is concerned they are equivalent. This implies
    \begin{equation}
        \min_{G_x \in \mathcal{S}} |G_x O_U D_L| < \min_{G_x \in \mathcal{S}} |G_x O_C D_L| ,
    \end{equation}
    where $G_x$ is an element of $\mathcal{S}$, which denotes the stabilizer group, and $|\cdot|$ denotes the weight of a Pauli operator. The weight of the minimal-weight operator equivalent up to stabilizers to $O_C D_L$ is no greater than the sum of the weights of the minimal-weight operators equivalent up to stabilizers to $O_C$ and $D_L$ individually. We can continue:
    \begin{equation}
    \label{eq: added errors flip condition}
        \min_{G_x \in \mathcal{S}} |G_x O_U D_L| < \min_{G_x \in \mathcal{S}} |G_x O_C D_L| \le \min_{G_x \in \mathcal{S}} |G_x O_C|+ \min_{G_y \in \mathcal{S}} |G_y D_L| < \min_{G_x \in \mathcal{S}} |G_x O_U| +\min_{G_y \in \mathcal{S}} |G_y D_L| .
    \end{equation}
    We conclude that the added error must be such that the minimal-weight correction of $O_U D_L$ is less than the minimal-weight correction of $O_U$ plus the minimal-weight correction of $D_L$. This happens when the disconnected error $D_L$ lies near $O_U$ such that the minimal-weight decoder will tend to form a loop out of parts of $D_L$ and $O_U$. This is possible only in cases like the one in figure \ref{fig: new exceptional term}.
    
    The condition in equation (\ref{eq: added errors flip condition}) requires a special combination of disconnected error and original partition. This is possible for both coherent- and incoherent-type disconnected errors as defined in Section \ref{sec: disconnected pieces}. Let us consider incoherent-type disconnected errors first. This is what is illustrated in figure \ref{fig: new exceptional term}. The disconnected error causes the uncorrectable side of the partition to become correctable when $D_L$ contains at least two errors in a row adjacent to $O_U$. Based on our condition, we observe that the number of added errors that flip the partition is greatest for the lowest-weight partitions. These terms require the fewest added errors to flip. We also see that the number of these added errors increases with the length of the logical string. A longer string has more adjacent qubits. This implies that the value of the disconnected part is decreasing with string length. This fact was used in Lemma \ref{Lemma: Coherent path counting}.
    
    \begin{figure}
        \centering
        \includegraphics[height = 7cm]{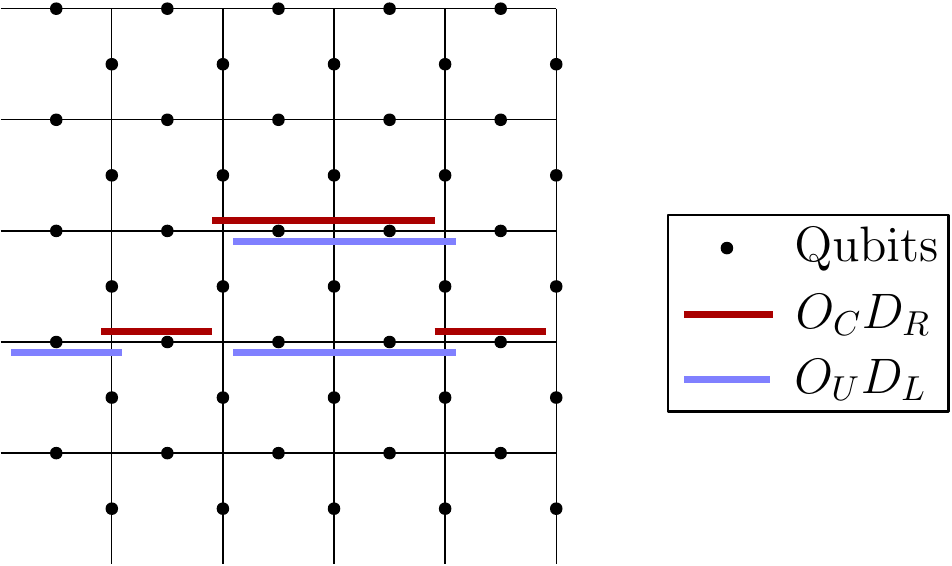}
        \caption{This figure shows a partition of a connected logical string together with a disconnected error. The disconnected error is incoherent-type, so $D_L = D_R$. The uncorrectable error $O_C D_R$ is in red, while the correctable error $O_U D_L$ is in blue. The two form a length-5 connected logical string that runs left to right across the code. Without the disconnected errors, $O_C$ would be correctable and $O_U$, uncorrectable. Therefore, the  added disconnected errors have flipped the original partition.}
        \label{fig: new exceptional term}
    \end{figure}
    
    We seek to prove that terms like the one in figure \ref{fig: new exceptional term} are negligible in the coherent logical noise components. We will do this by mapping each combination of a partition of a connected logical string and a set of disconnected errors such that uncorrectable and correctable errors in the partition are flipped to a partition of a longer connected logical string. There exists a unique stabilizer operator that will multiply the starting partition plus disconnected errors and produce a partition of a longer logical string. This is illustrated in figure \ref{fig: new exceptional term cousin}. Our condition in equation (\ref{eq: added errors flip condition}) says that the minimum stabilizer-equivalent operator to $O_U$ is lower weight than $O_U$. The stabilizer operator we need to map figure \ref{fig: new exceptional term} to figure \ref{fig: new exceptional term cousin} is the product of $O_U D_L$ and its minimal-weight correction. The resulting connected logical string is longer than the original connected logical string, but the total weight of the noise term (connected and disconnected) is smaller. This must be true because we have lowered the weight of the errors in blue $(O_U D_L)$, and we have not changed the weight of the errors in red $(O_C D_R)$.
    
    \begin{figure}
        \centering
        \includegraphics[height = 7cm]{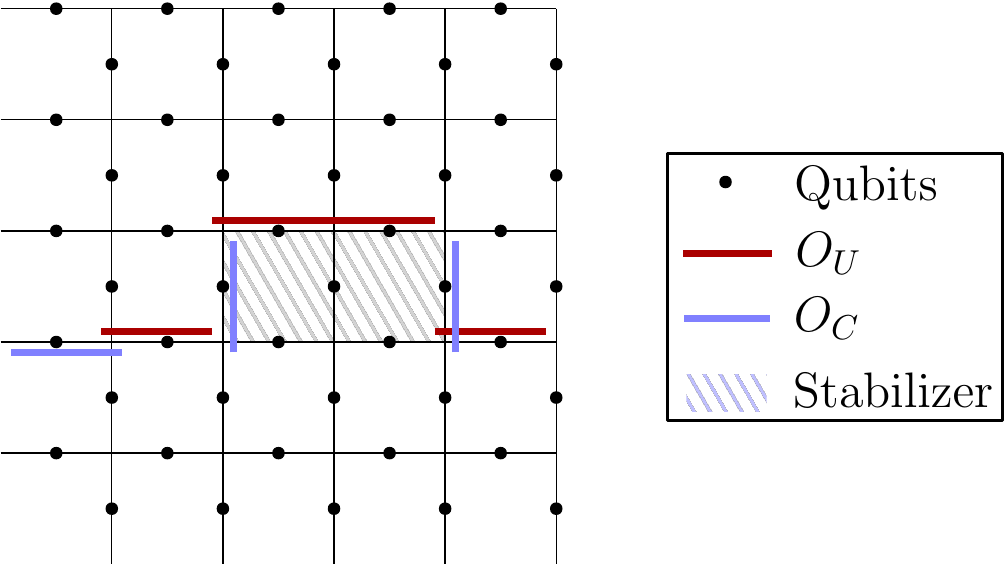}
        \caption{Here is the partition of a connected logical string corresponding to figure \ref{fig: new exceptional term}, with the new uncorrectable error $O_U$ in red and the new correctable error $O_C$ in blue. This term shares the same syndrome as the one in figure \ref{fig: new exceptional term}. We can always multiply right or left hand sides by a stabilizer to produce different coherent terms. This term is produced from figure \ref{fig: new exceptional term} by multiplying the correctable side in blue by the stabilizer operator in gray crosshatching. Notice that the connected logical string is longer but the total weight of the term is smaller.}
        \label{fig: new exceptional term cousin}
    \end{figure}
    
    In our previous analysis of the connected part of the coherent logical noise, we neglected logical strings with length $>L+2 \zeta$ for a cut-off constant $\zeta$. Then, we neglected short logical strings with atypical shape. Finally, we neglected the unlikely partitions of each string, which we called exceptional terms. In this proof we began with a likely partition of a short, typical connected logical string. We added disconnected errors, and in cases like the one in figure \ref{fig: new exceptional term} where the added errors flipped the partition, we mapped these terms to partitions of a different connected logical string like in figure \ref{fig: new exceptional term cousin}. The final part of our proof is to argue that we can neglect this class of terms, where disconnected errors changed how the connected partition is decoded.
    
    First, we observed above that the weight of the new connected term produced by our mapping is less than the weight of the original term with disconnected errors. This means that the term in figure \ref{fig: new exceptional term} is suppressed in powers of $\sin \theta/2$ relative to the term in figure \ref{fig: new exceptional term cousin}. Second, we will argue that the new connected term is one we have already neglected. Recall how we constructed new connected terms like the one in figure \ref{fig: new exceptional term cousin}. We took a likely partition of a typical, short, connected logical string and added disconnected errors to it in such a way that we flipped how the partition was decoded. The original uncorrectable side became correctable and vice versa. Then, we multiplied the correctable error $O_U D_L$ by a particular stabilizer operator to produce a new term that is a partition of a new connected logical string. We make the following observation. One of two things must be true. One is that the new logical string has an atypical shape, specifically if the logical string runs left to right across the code, the steps up and down in the lattice are separated by less than $\gamma L$. $\gamma$ was our chosen constant from Lemma \ref{Lemma: shape of strings B} that lower bounded the separation between the vertical steps in the logical string. The alternative is that the new connected logical string has a typical shape but the partition we produce is unlikely. If the stabilizer we multiply by in figure \ref{fig: new exceptional term cousin} has a width of at least $\gamma L$ the shape of the new connected logical string may be typical. However, in that case the partition we get for the longer connected logical string has a row of $\gamma L$ qubits all belonging to the blue error. We proved in Lemma \ref{Lemma: exceptional terms} that partitions with this feature are an exponentially small (in $\sqrt{L}$) fraction of the total partitions. We neglected these partitions in our earlier sum over connected terms. We find that the terms with disconnected errors that flip how the partition is decoded are in one to one correspondence with terms we have already neglected and moreover, the magnitudes of the terms with disconnected errors are smaller by a number of powers of $\sin \theta/2$. We conclude that such terms contribute less to the logical noise than the terms we have already neglected.
    
    So far in this proof we dealt with incoherent-type added errors. This was for simplicity, so that we had only one picture in mind. The argument for coherent-type added errors is the same. Coherent-type disconnected errors can also flip the correctable and uncorrectable sides of a partition of a connected logical string. We have already stated the condition when this occurs. The disconnected terms on the uncorrectable side $D_L$ must contain a contiguous set of errors near a contiguous set of errors in the uncorrectable part of the partition $O_U$.
    
    We bound the contribution of the disconnected coherent-type errors that flip the correctable and uncorrectable sides of the partition in the same way as we did the incoherent-type. We will use a mapping that takes such a term and produces a partition of a longer connected logical string. The mapping multiplies by a suitable stabilizer operator as depicted in figure \ref{fig: new exceptional term loop cousin}. In this case the new connected term is negligible for the same reasons as in the incoherent-type added error case. The connected logical string produced from the original partition plus the disconnected coherent-type errors either has an atypical shape or the original partition was exponentially unlikely (in $\sqrt{L}$). We conclude that the contribution of terms with disconnected errors that flip the partition is negligible in the logical noise.
    
    \begin{figure}
        \centering
        \includegraphics[height = 7cm]{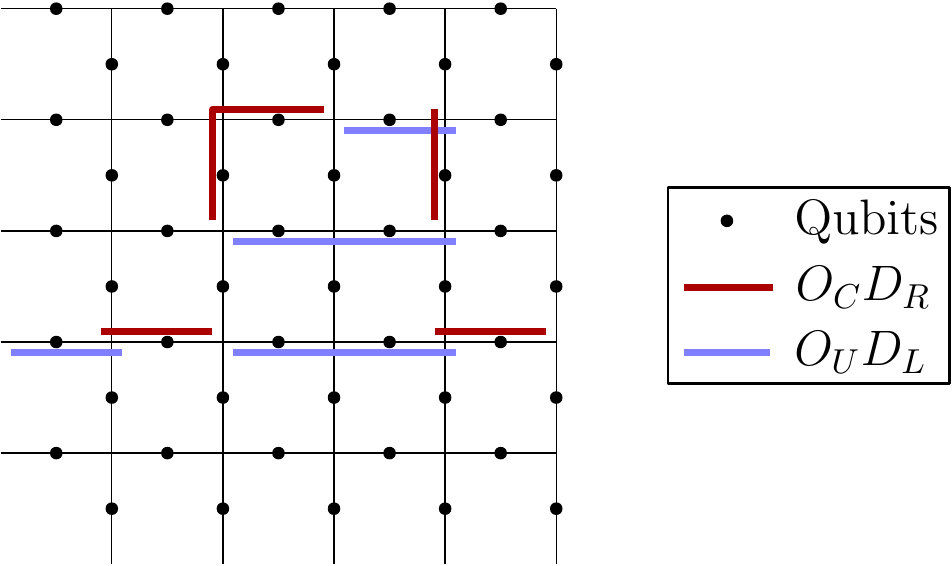}
        \caption{Here we have a partition of a connected logical string together with a disconnected error. The partition is the same as the one in figures \ref{fig: new exceptional term} and \ref{fig: new exceptional term cousin}. The disconnected error is now a coherent-type error. The uncorrectable error $O_C D_R$ is in red, while the correctable error $O_U D_L$ is in blue. Without the disconnected errors, $O_C$ would be correctable and $O_U$, uncorrectable. The added loop of disconnected errors has flipped the original partition.}
        \label{fig: new exceptional term loop}
    \end{figure}
    
    \begin{figure}
        \centering
        \includegraphics[height = 7cm]{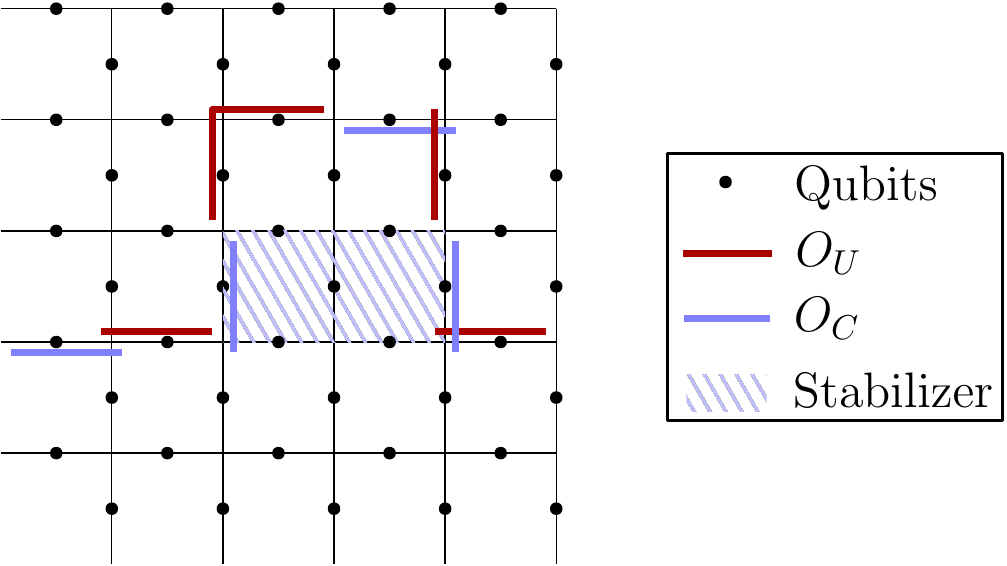}
        \caption{This is connected string and partition that corresponds to figure \ref{fig: new exceptional term loop}, with the new uncorrectable error $O_U$ in red and the new correctable error $O_C$ in blue. We can always multiply right or left hand sides by a stabilizer to produce different coherent terms. This term is produced from figure \ref{fig: new exceptional term loop} by multiplying the correctable side in blue by the stabilizer operator in gray crosshatching. The connected logical string is longer than the one in figure \ref{fig: new exceptional term loop} but the total weight of the term is less.}
        \label{fig: new exceptional term loop cousin}
    \end{figure}
    
    \end{proof}

\section{The Disconnected Part of the Incoherent Logical Noise}
\label{app: incoherent disconnected part}

In Appendix \ref{app: Disconnected errors} we proved that for the dominant noise terms in the coherent logical noise components, the disconnected part was equal to $1$ up to small corrections. We will now prove the same statement for the dominant noise terms in the incoherent logical noise components. 

\begin{Lemma}
\label{lemma: incoherent disconnected part}
    In equation (\ref{eq: incoherent sum string form}), we wrote an incoherent logical noise component as a sum over the contributions from individual logical strings. This included a disconnected factor. Here we prove that we can set the disconnected factor equal to $1$ and make only a small error. In other words suppose we write
    \begin{equation}
        \tilde{\chi}_{\scriptscriptstyle{Z_1 Z_1}} \geq \sum_{\mathcal{L}} \sum_{O_U \subset \mathcal{L}} \frac{1}{|\{O_C^\prime \}|} \sum_{O_U^\prime} (O_U \rho \, O_U^\prime) (1 + \mathcal{E}_1 + \mathcal{E}_2),
    \end{equation}
    where the sum over $\mathcal{L}$ includes all short typical logical strings, $O_U$, $|\{O_C^\prime\}|$, and $O_U^\prime$ are all as described in Section \ref{sec: Incoherent sum over strings}, and $\mathcal{E}_1$ is the error we make by neglecting various connected terms, including high-weight terms and terms with mismatched weight. In Lemma \ref{Lemma: High Weigh Incoherent Terms}, we proved that
    \begin{equation}
        |\mathcal{E}_1| < \mathcal{O}\left((\sin \theta)^{2\zeta} \right) .
    \end{equation}
    $\mathcal{E}_2$ represents the error we make when we set the disconnected factor equal to $1$. Then,
    \begin{equation}
        |\mathcal{E}_2| < \frac{8 \gamma \zeta^2}{\sqrt{L}} + \mathcal{O}\left( \frac{1}{L} \right) .
    \end{equation}
\end{Lemma}

\begin{proof}
    We can follow the argument from Section \ref{sec: disconnected pieces} and Appendix \ref{app: Disconnected errors}. The connected noise terms we considered in those sections had the form $(O_U \rho \, O_C)$. Here, we will consider noise terms with the form $(O_U \rho \, O_U^\prime)$. We will imagine arriving at these noise terms in the manner of Section \ref{sec: Incoherent sum over strings}. Namely, we begin with a short logical string with a typical shape. We partition the logical string into $O_U$ and $O_C$. Then we choose an operator $O_U^\prime$ with the same syndrome as $O_U$. We denoted the set of possible $O_U^\prime$ by $\{O_U^\prime\}$ 
    
    Now that we have a connected noise term $(O_U \rho \, O_U^\prime)$, we can think of dressing it with disconnected errors in exactly the same way as we did in the coherent case. In Section \ref{sec: disconnected pieces} we observed that the added errors that make up the disconnected part can be divided into coherent and incoherent-type. The coherent-type added errors are when we add different errors to $O_U$ and $O_U^\prime$. In this case the errors we add to $O_U$ and $O_U^\prime$ form a loop (with nonzero area). We saw that as long as the loop was positioned such that the added errors did not change how the connected noise term was decoded, the sum over the possible ways of dividing the errors in the loop between $O_U$ and $O_U^\prime$ gave zero. We considered incoherent-type added errors where we added the same error to $O_U$ and $O_U^\prime$. In this case the contribution was nonzero. As long as the added error did not change how the connected noise term was decoded, the incoherent-type added errors contributed a $\sin^2 \theta/2$ term on each qubit. Together with the $\cos^2 \theta/2$ term corresponding to no error on each qubit, this gives $1$ for the disconnected part. This applies to the added errors that do not change how the connected part was decoded. Therefore, our approach is to write the disconnected part as $1$ plus a correction that comes from the configurations of added errors that change how the connected part is decoded.
    
    In Lemma \ref{lemma: added error exceptional terms} we considered added errors that change how $O_U$ is decoded in the coherent logical noise components. For connected noise terms that enter into the incoherent logical noise the correction to the disconnected part comes from the same source, certain added errors that change how the connected term is decoded. In that case the added errors flipped the correctable and uncorrectable sides of the partition, which gives a phase of $-1$. In the incoherent case, if $O_U$ is made correctable by the added errors, then so is $O_U^\prime$, and the resulting term contributes to the identity part of the logical noise. In effect, there are disallowed added errors, which reduce the value of the disconnected part. The counting of such terms is identical to what we did in Lemma \ref{lemma: added error exceptional terms}. Recall that these added error terms were related to connected terms that either had an atypical shape or an unlikely partition. The contribution to $\mathcal{E}_2$ from these terms is proportional to the fraction of atypical logical strings from Lemma \ref{Lemma: shape of strings B}. The contribution from the unlikely partitions is exponentially small in $\sqrt{L}$ as before.
    
    There is another class of added errors that contribute to the correction to the disconnected part. Some errors are near the correction $E_s$ so that once the errors have been added, they become part of a new connected term. An example is shown in figure \ref{fig: incoherent added errors near Es}. This class of terms contains only incoherent-type added errors. Coherent-type added error placed here will still give $0$ after we sum over the ways of splitting the errors into the left and right disconnected errors. This is the same as for coherent-type added errors far from the connected logical string. For a likely partition of a short logical string with typical shape, $E_s$ will have the same weight as $O_C$. The incoherent-type added errors that join the connected part may either lie near to $E_s$ or they may be contained in $E_s$. We will first study the case where the added errors are not contained in $E_s$. These terms are closely related to the situation we just analyzed where the added errors sit next to $O_U$. The condition on the added error is analogous to equation (\ref{eq: added errors flip condition}). Let $D$ denote the added incoherent-type error and $\mathcal{S}$ denote the stabilizer group. Then,
    \begin{equation}
        \min_{G_x \in \mathcal{S}} |G_x O_C D| < \min_{G_y \in \mathcal{S}} |G_y O_c| + \min_{G_z \in \mathcal{S}} |G_z D| . 
    \end{equation}
    If this condition is satisfied, then the added error becomes part of a new connected term $(O_U D \rho \, O_U^\prime D)$. Together with the new lowest-weight correction $O_U D$ forms a new connected logical strings. This logical string is similar to the old logical string, but it contains a detour where it veers off to include the error $D$. This logical string either has an unlikely shape because it includes two closely spaced vertical steps, or the error $D$ has width $> \gamma L$. In the latter case we also require that the original correction $E_s$ included $>\gamma L$ consecutive qubits. This is exponentially unlikely according to the counting we did in Lemma \ref{Lemma: exceptional terms} and the bound we wrote in equation (\ref{eq: Exceptional terms error bound}). We conclude that the contribution to the error term $\epsilon_2$ from the noise terms we arrive at by adding errors proximate to $O_C$ is small.
    
    \begin{figure}
    \centering
    \includegraphics[width=16cm]{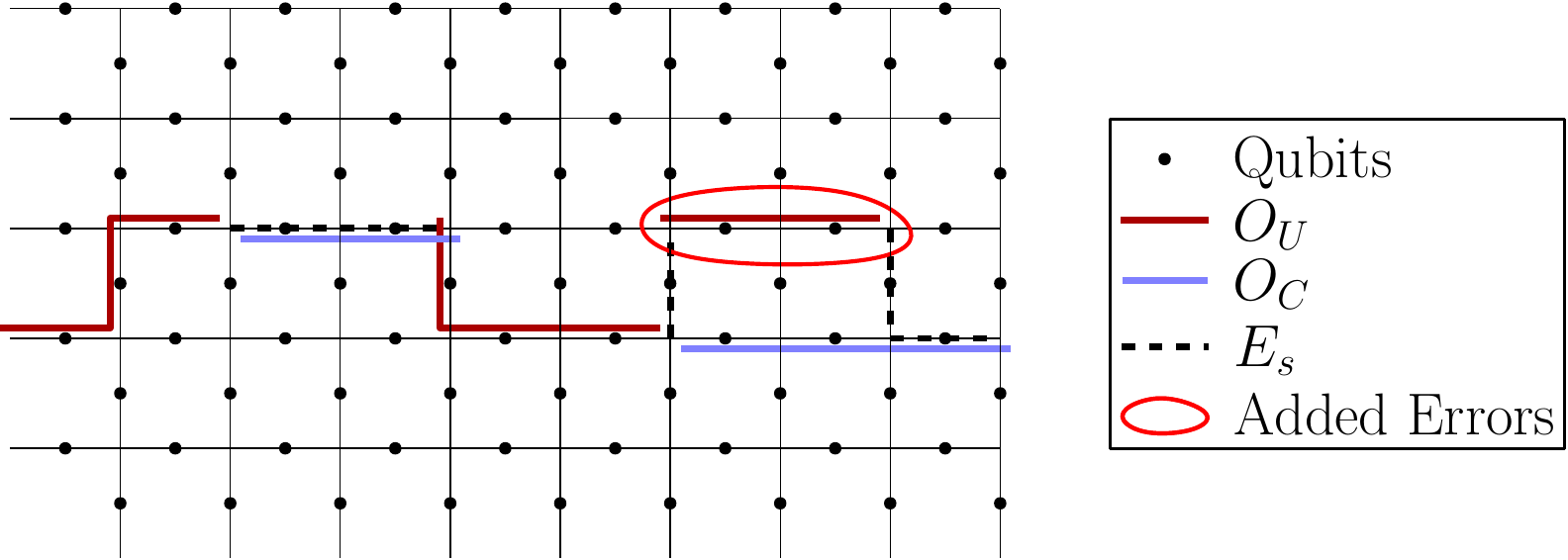}
    \caption{This figure illustrates a certain type of added error term. We start with a short, typical logical string and partition it into operators $O_U$ and $O_C$. Then, we would choose some operator $O_U^\prime$ (not pictured) with the same weight and syndrome as $O_U$ to produce an incoherent term. The errors in the red circle are added to both $O_U$ and $O_U^\prime$. The minimal-weight correction is shown as a black dashed line. The added errors are not disconnected but form a new connected term.}
    \label{fig: incoherent added errors near Es}
    \end{figure}

    This leaves the added errors the lie within $E_s$ or one of the operators with the same syndrome and weight. In this case the new connected logical string that results from adding the errors is the same as the old logical string. We will compare the set of terms we arrive at by adding errors in this manner to the set of terms with the same $O_U$ but a higher-weight $O_U^\prime$. We will argue that there are more of the terms with higher-weight $O_U^\prime$. We have already neglected such noise terms in Lemma \ref{Lemma: Mismatched weights}, so we conclude that the correction to the disconnected part is small.
    
    We start with a connected noise term $(O_U \rho \, O_U^\prime)$, where $|O_U| = |O_U^\prime| = w$. We form a higher-weight noise term by adding an incoherent-type error $D$ within one of the operators $O_C^\prime$ to produce a new connected noise term $(O_U D \rho \, O_U^\prime D)$. This adds a total of two to the weight of the term. We can place the added error anywhere within one of the operators $O_C^\prime$. The number of possibilities is $\mathcal{O}(w)$. Now consider the possible connected noise terms with the same $O_U$ but instead of choosing $O_U^\prime$ with weight $w$, we set $|O_U^\prime|= w+2$. Suppose we start with an operator $O_U^\prime$ with weight $w$. We can construct one with $w+2$ by adding an extra ``cap" consisting of three qubits around a single plaquette or star. This is illustrated in figure \ref{fig: incoherent added errors inside Oc}, where three possible choices of $O_U^\prime$ are drawn with the red dashed line. The number of possible choices is at least $\mathcal{O}(w)$, because we can place a cap at each location along $O_U$. The full set of possibilities will generally be larger. If we consider a pair of added errors lying within $O_C^\prime$, then we compare to the connected noise terms with $|O_U^\prime| = |O_U|+4$. The noise terms where $O_U$ and $O_U^\prime$ have different weights were discussed in Section \ref{sec: Mismatched weights}. We proved in Lemma \ref{Lemma: Mismatched weights} that the contribution of these terms is negligible.
    \begin{figure}
    \centering
    \includegraphics[width=16cm]{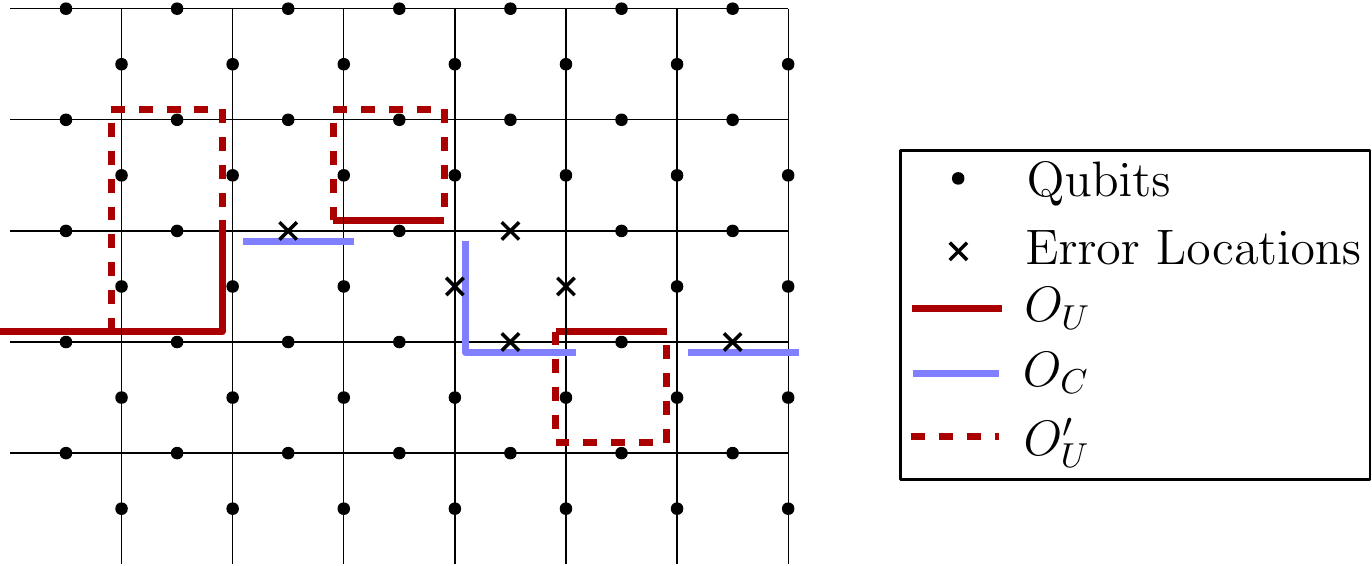}
    \caption{This figure illustrates the idea of the proof that a certain type of added error term contributes only a small error. Consider a connected logical string that we partition into $O_U$ and $O_C$, the solid red and blue lines. The class of added errors we consider are those where the added error lies along $O_C$ or one of the operators $O_C^\prime$ with the same weight and syndrome. The possible locations for such an added error are marked by the $\times$ symbols. We prove that such terms are negligible by comparing to the noise terms where the operator $O_U^\prime$ is chosen with weight 2 greater than $O_U$. Three of the possible choices for $O_U^\prime$ are drawn with the dashed red lines. In this example there are 12 possible $O_U^\prime$ operators and six possible added error terms.}
    \label{fig: incoherent added errors inside Oc}
    \end{figure}
    Finally we conclude, just as in Lemma \ref{lemma: added error exceptional terms}, that each of the disconnected errors that contribute to the error $\mathcal{E}_2$ can be matched with a connected noise term that we have already neglected. In other words, the error $\mathcal{E}_2$ is less than $\mathcal{E}_1$. Therefore, we can say that the disconnected part $=1$ and make only a small error for low-weight connected terms.
    
\end{proof}

\section{Physical $Y$ Errors}
\label{app: Physical Y Errors}

In Lemma \ref{Lemma: general rotation angles} we considered rotations in the $X{-}Z$ plane where the single-qubit rotation angles were allowed to differ. Here we prove that allowing for rotations partly around the $Y$ axis on the physical qubits will decrease the coherence of the logical noise channel.

\begin{Lemma}
\label{lemma: Y rotations}
    Consider an $L \times L$ toric code and a noise channel that consists of single-qubit rotations by an angle $\theta$ about an arbitrary axis. Suppose that $| \sin \theta|<1/L$ as in Lemmas \ref{Lemma: Coherent path counting} and \ref{Lemma: High Weigh Incoherent Terms}.
    Then, the connected contribution to the logical noise from low-weight terms is most coherent when the single-qubit rotations are about an axis in the $X{-}Z$ plane. We proved elsewhere that the low-weight connected contribution dominates the logical noise components.
\end{Lemma}

\begin{proof}
Let $\theta_X$, $\theta_Y$, $\theta_Z$ denote the rotation angles about the $X$, $Y$, and $Z$-axis, respectively, so that $\theta_X^2 + \theta_Y^2 + \theta_Z^2 = \theta^2$. Of the coherent logical noise components, according to Lemmas \ref{Lemma: Other logical maps} and \ref{Lemma: More general coherent terms}, the dominant components are the ones $(\tilde{L}_a \tilde{\rho})$, where $\tilde{L}_a$ is a logical $X$ or $Z$ operator on one of the encoded qubits. We apply several of our lemmas to restrict the noise terms we consider, just as in Theorem \ref{theorem: Big Theorem}.
Among the noise terms that contribute to the coherent logical noise, we keep the terms with short, typical logical strings and non-exceptional partitions. Among the noise terms that contribute to the incoherent logical noise, we keep the terms where the logical string $\mathcal{L}$ is short and typical, $|O_U| = (|\mathcal{L}|+1)/2$, $|O_C| = |E_s|$, and $|O_U^\prime| = |O_U|$.

First suppose that $\theta_Z = 0$. Then, the logical $(X_1 \tilde{\rho})$ noise component is generated from noise terms $(O_U \rho \, O_C)$ where $O_U$ and $O_C$ together contain $X$ acting on every qubit along an $X_1$ logical string. Meanwhile, the incoherent logical $(X_1 \tilde{\rho} X_1)$ noise component is also generated by $X_1$ logical strings. In Theorem \ref{theorem: Big Theorem}, we state a bound on the relative magnitude of these logical noise components. Here $\theta_X$ plays the role of $\theta$ in equation (\ref{eq: coherent incoherent ratio}). Under our $\theta_X$ and $\theta_Y$ rotation noise model, we also have a non-zero logical $(Z_1 \tilde{\rho})$ noise component. This is generated by connected noise terms, $(O_U \rho \, O_C)$, where $O_U$ and $O_C$ together contain both $X$ and $Y$ acting on every qubit along a $Z_1$ logical string. The number of $Z_1$ logical strings with length $\ell$ is the same as the number of $X_1$ logical strings with length $\ell$. However, the weight of the noise terms that contribute to $(Z_1 \tilde{\rho})$ is $\ell^2$. The contribution of each noise term is $(\sin \theta_X/2)^\ell (\sin \theta_Y/2)^\ell$. In contrast, the noise terms that contribute to $(X_1 \tilde{\rho})$ are all proportional to $(\sin \theta_X/2)^\ell$. Therefore, $(Z_1 \tilde{\rho})$ is exponentially smaller in $L$ relative to $(X_1 \tilde{\rho})$ for any choice of rotation axis in the $X{-}Y$ plane. The $(Z_1 \tilde{\rho})$ noise component has a negligible effect on the relative magnitudes of the coherent and incoherent logical noise components. We also have an incoherent $(Z_1 \tilde{\rho} \, Z_1)$ noise component. This is generated by noise terms, $(O_U \rho \, O_U^\prime)$, where $O_U$ and $O_U^\prime$ contains $Y$ errors along an uncorrectable subset of a $Z_1$ logical string. These noise terms have magnitude $(\sin \theta_Y/2)^{\ell+1}$, which is exponentially large relative to the noise terms that contributed to $(Z_1 \tilde{\rho})$. It follows that the logical coherence is maximized when $\theta_Y = 0$ and $\theta_X=\theta$. We began by supposing $\theta_Z = 0$. Next, we will consider the case where $\theta_X$, $\theta_Y$, and $\theta_Z$ are all nonzero.

Suppose $|\theta_Z| \geq |\theta_X|$. If not, switch the role of $X$ and $Z$ in what follows. Fix a $Z_1$ logical string $\mathcal{L}$. The contribution of the logical string $\mathcal{L}$ to $\tilde{\chi}_{\scriptscriptstyle{Z_1 I}}$ is a sum over the partitions of $\mathcal{L}$. For each partition $(O_U \rho \, O_C)$, we can replace a $Z$ error in $O_U$ with a $Y$ error if we add an $X$ error on the same qubit to $O_C$. Similarly we can replace a $Z$ error in $O_C$ by a $Y$ error if we add an $X$ error on the same qubit to $O_U$. The $Z$ syndrome is unchanged, but now we also have a non-trivial $X$ syndrome corresponding to the $X$ error on the chosen site. This does not change how any partitions are decoded, but it does change the weight. The contribution of each partition to $\tilde{\chi}_{\scriptscriptstyle{Z_1 I}}$ is a sum over all combinations of either a $Z$ error or a $Y$ and an $X$ error on every qubit in $\mathcal{L}$. The terms with $Y$ errors have higher weight. This means they contain extra factors of $\sin \theta_Y/2$, which is small since $|\sin \theta| <1/L$. At the same time, the logical string $\mathcal{L}$ contributes to the $\tilde{\chi}_{\scriptscriptstyle{Z_1 Z_1}}$ logical noise component. These noise terms include some that feature only $Z$ errors and others with some number of $Z$ errors replaced by $Y$ errors. Unlike the contributions to $\tilde{\chi}_{\scriptscriptstyle{Z_1 I}}$, these terms with $Y$ errors are not higher-weight. There are no extra factors of $\sin \theta_Y/2$, and we conclude that the incoherent logical noise components are made larger relative to the coherent logical noise components. Therefore, the logical coherence is maximized when $\theta_Y = 0$.

\end{proof}

\section{Other Logical Maps}
\label{app: X1X2 or Y1 terms}

In the Section \ref{section: The Coherent Sum}, we restricted our attention to logical coherent terms of the form $(\tilde{L_a} \tilde{\rho})$, where $L_a$ is an $X$ or $Z$ operator on one of the encoded qubits. Now we would like to consider the case where $L_a$ acts nontrivially on both encoded qubits or as $Y$ on one or both of the encoded qubits. 

\begin{Lemma}
\label{Lemma: Other logical maps}
    Consider the toric code with minimal-weight decoding and a noise model that consists of uniform single-qubit unitary rotations about a fixed axis. Then, the coherent logical noise components, $(\tilde{L_a} \tilde{\rho})$, where $\tilde{L}_a$ is a $Y$-type logical operator or $\tilde{L}_a$ is a non-trivial logical operator on both encoded qubits, are negligible relative to the components where $\tilde{L}_a$ is an $X$ or $Z$-type logical operator on one encoded qubit.
\end{Lemma}
\begin{proof}
Suppose we have $\tilde{L}_a = X_1 Z_2$. Logical strings of this type are the product of two operators of the type we have already considered. Each connected noise term that contributes to the logical noise component, $(X_1 Z_2 \tilde{\rho})$, is a product of a connected noise term that contributes to $(X_1 \tilde{\rho})$ and a connected noise term that contributes to $(\tilde{\rho} Z_2)$. It follows that up to corrections that come from the disconnected part, $(X_1 Z_2 \tilde{\rho}) \approx (X_1 \tilde{\rho}) (Z_2 \tilde{\rho})$. The logical components, $(X_1 \tilde{\rho})$ and $(Z_2 \tilde{\rho})$, are both small if error correction is working, so the logical $(X_1 Z_2 \tilde{\rho})$ component will be negligible. If $\tilde{L}_a$ is $Y$-type operator on the first encoded qubit, the argument is the same, since $Y$-type logical strings are products of $X$ and $Z$-type logical strings, $Y_1 = X_1 Z_1$. 

If we have $\tilde{L}_a = Z_1 Z_2$, the logical component $(Z_1 Z_2 \tilde{\rho})$ is no longer a product of $(Z_1 \tilde{\rho})$ and $(Z_2 \tilde{\rho)}$. This is because $Z_1$ and $Z_2$-type logical strings can overlap, and this changes the counting of logical strings of a fixed weight. Figure \ref{fig: z1 z2 error} shows two examples of this kind of logical string. At length $2L$, where $L$ is the code distance, there are many connected logical strings because we can have a single connected string that wraps the torus along both directions. If we count the shortest paths between two points in the square lattice separated by distance $l_1$ in the horizontal and $l_2$ in the vertical, we get 
\begin{equation}
    \label{eq: square lattice paths}
    \text{number of shortest paths travelling $l_1$ horizontal and $l_2$ vertical spaces} = \binom{l_1+l_2}{l_1} .
\end{equation}
We can use this to bound the number of weight-$2L$ logical $Z_1 Z_2$ strings. Fix two sites in the code, qubit $i$ along the vertical edge of the code and qubit $j$ along the horizontal edge. Now count the number of shortest paths that connect these points. We have
\begin{equation}
    \binom{i+j}{i} \binom{2L-2-i-j}{L-1-i} + \binom{L-1-i+j}{j} \binom{i+L-1-j}{i}
\end{equation}
for the two ways of linking the edge points. We simply apply the result from equation (\ref{eq: square lattice paths}). In the end we find that
\begin{equation}
    \text{number of weight-$2L$ $Z_1 Z_2$ logical strings} \approx \frac{4^L}{\sqrt{\pi L}} .
\end{equation}
In Section \ref{sec: Path Counting}, we counted logical strings that act as $X$ or $Z$ on one of the encoded qubits starting from length $L$, and we found exponentially many logical strings at higher weights. If we consider weight-$2L$ logical strings, we find order $\mu^{2 L}$ logical strings, where $\mu \approx 2.64$ for the 2D square lattice. This is more than $4^L$, so we have more of the high-weight logical strings that act on only one encoded qubit. Further, in our path counting in Lemma \ref{Lemma: Coherent path counting}, we neglected all logical strings of length $>L+2 \zeta$ for a constant $\zeta$. The strings of length $\geq 2L$ contribute negligibly for large $L$. Then, we conclude that the logical noise components, $(\tilde{L}_a \tilde{\rho})$, where $L_a$ is a $Y$-type logical operator or acts on both encoded qubits are negligible relative to the noise components where $L_a$ acts as $X$ or $Z$ on one of the encoded qubits.

\begin{figure}
    \centering
    \includegraphics[height = 9cm]{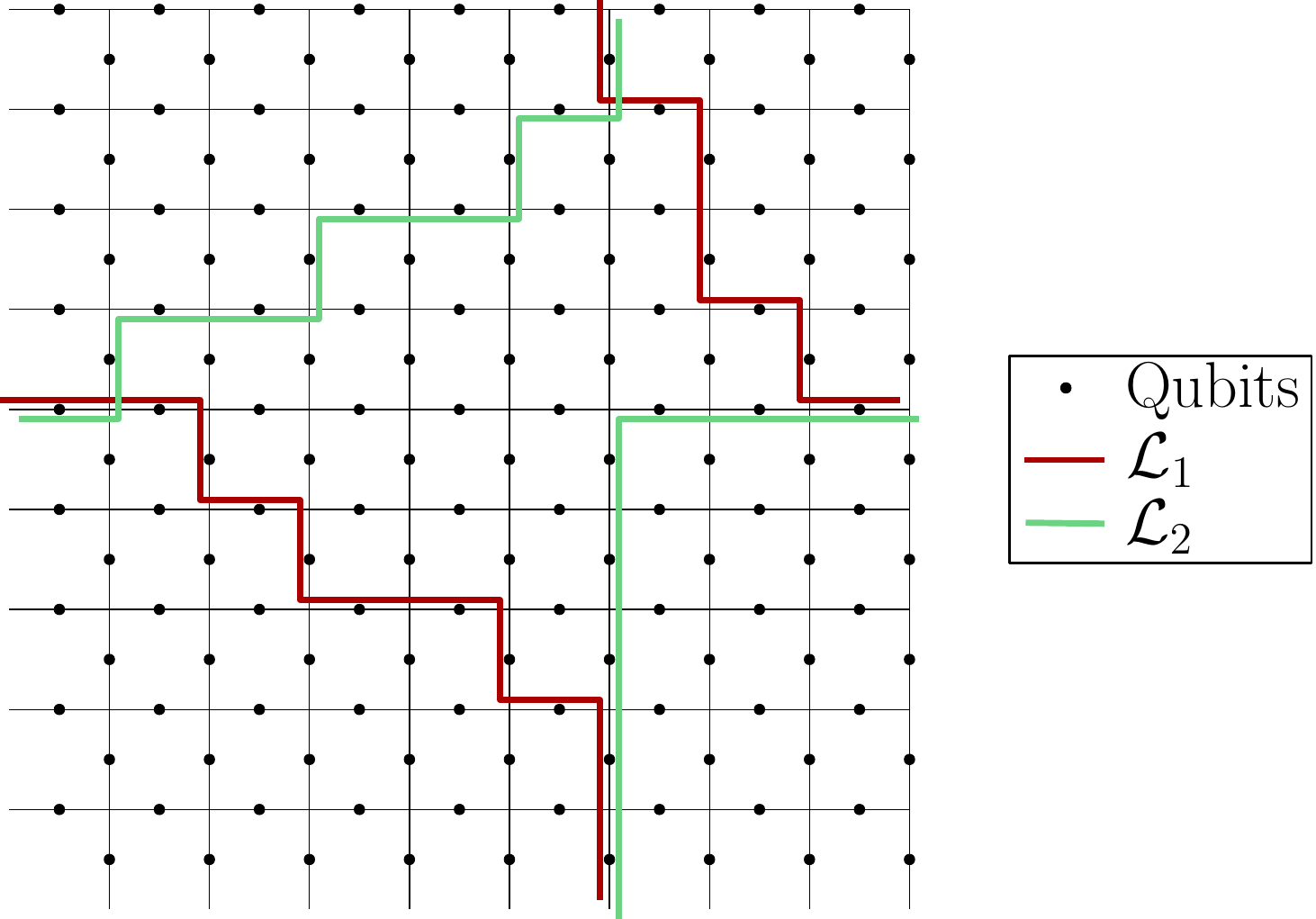}
    \caption{Here are two examples of lowest-weight $Z_1 Z_2$ logical strings, $\mathcal{L}_1$ and $\mathcal{L}_2$, that act as $Z$ on both encoded qubits. Notice that red and green connect the edge points in different (but topologically equivalent) ways.}
    \label{fig: z1 z2 error}
\end{figure}
\end{proof}

\section{More General Coherent Terms}
\label{app: a and b nontrivial}

We have considered coherent logical noise components $(\tilde{L}_a \tilde{\rho})$, where $\tilde{L}_a$ is a logical operator that acts as $X$ or $Z$ on exactly one of the encoded qubits. We must also consider logical noise components $(\tilde{L}_a \tilde{\rho} \tilde{L}_b)$, where $\tilde{L}_a$ and $\tilde{L}_b$ are different non-trivial operators on the encoded qubits.

\begin{Lemma}
    Consider the $L\times L$ toric code with noise that consists of single-qubit unitary rotations about a fixed axis by angle $\theta$ on every qubit, where $|\sin \theta|$ is $<1/L$ as in Lemma \ref{Lemma: Coherent path counting}. Each coherent logical noise component of the form $(\tilde{L}_a \tilde{\rho} \tilde{L}_b)$, where $\tilde{L}_a$ and $\tilde{L}_b$ are different nontrivial logical operators, is negligible relative to the coherent logical noise components with $\tilde{L}_b = \tilde{id}$. Each of the more general coherent terms is given by
    \begin{equation}
        (\tilde{L}_a \tilde{\rho} \tilde{L}_b) \approx (\tilde{L}_a \tilde{\rho}) (\tilde{\rho} \tilde{L}_b) .
    \end{equation}
    $(\tilde{L}_a \tilde{\rho})$ and $(\tilde{\rho} \tilde{L}_b)$ are both small (because we are interested in the regime where error correction succeeds with high probability.) Therefore, we may safely neglect all logical noise components $(\tilde{L}_a \tilde{\rho} \tilde{L}_b)$, where $\tilde{L}_a$ and $\tilde{L}_b$ are different nontrivial logical operators.
    \label{Lemma: More general coherent terms}
\end{Lemma}

\begin{proof}

Our approach here is to bound the coherent logical noise components $(\tilde{L}_a \tilde{\rho} \tilde{L}_b)$, where $\tilde{L}_a$ and $\tilde{L}_b$ are different nontrivial logical operators, by the coherent logical noise components we have already considered. This follows because the short connected logical strings with different logical action do not overlap much. Overlap here means that the strings contain the same error acting on the same qubits. One possible overlap is between $Z_1$ and $Z_2$ logical strings. Pick a $Z_1$ logical string, $\mathcal{L}_1$, and a $Z_2$ logical string, $\mathcal{L}_2$. One string runs left to right, and the other runs top to bottom. If the horizontal string is longer than $L$, the code distance, then it has vertical steps along it, and these steps may overlap with the vertical logical string. An example is given in figure \ref{fig: z1 string and z2 string}. We assume $\mathcal{L}_1$ and $\mathcal{L}_2$ both have length $\leq L+2/zeta$ because of Lemma \ref{Lemma: Coherent path counting}. Then, we use Lemma \ref{Lemma: shape of strings A} to restrict to the case where all the steps are one lattice spacing at a time. Any possible overlap is on at most two sites as shown in figure \ref{fig: z1 string and z2 string}. Further, if we consider all possible pairs of a $Z_1$ logical string and a $Z_2$ logical string, only order $1/L$ strings have any overlap at all, so we can neglect possible overlap.
\begin{figure}
    \centering
    \includegraphics[height = 7cm]{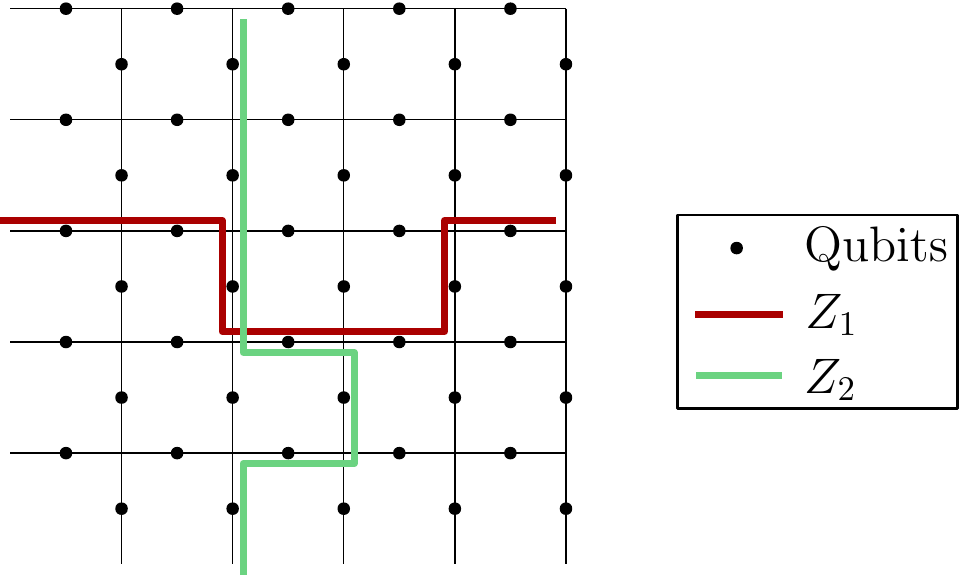}
    \caption{Here we have a $Z_1$ and a $Z_2$ logical string. They have an overlap of two qubits, but if we fix one string and consider all possible paths for the other string, we see that only order $1/L$ have any overlap.}
    \label{fig: z1 string and z2 string}
\end{figure}

Because the two logical strings $\mathcal{L}_1$ and $\mathcal{L}_2$ are approximately disjoint, when we sum over partitions, each partition approximately factors into a partition of $\mathcal{L}_1$ times a partition of $\mathcal{L}_2$. That is, each connected noise term in the sum for the logical $\tilde{\chi}_{\scriptscriptstyle{Z_1 Z_2}}$ is a partition $(O_U^{(1)} O_C^{(2)}\rho \, O_C^{(1)} O_U^{(2)})$ which is approximately equal to $(O_U^{(1)} \rho \, O_C^{(1)}) (O_C^{(2)} \rho \, O_U^{(2)})$ where $O_U^{(1)} O_C^{(1)} = \mathcal{L}_1$ and $O_U^{(2)} O_C^{(2)} = \mathcal{L}_2$. Therefore, $\tilde{\chi}_{\scriptscriptstyle{Z_1 Z_2}} \approx \tilde{\chi}_{\scriptscriptstyle{Z_1 I}} \tilde{\chi}_{\scriptscriptstyle{I Z_2}}$ up to small corrections from the overlap between $\mathcal{L}_1$ and $\mathcal{L}_2$ and from the disconnected part. Each of the terms $\tilde{\chi}_{\scriptscriptstyle{Z_1 I}}$ and $\tilde{\chi}_{\scriptscriptstyle{I Z_1}}$ will be $\ll 1$ if we are in a regime where error correction succeeds. Therefore, the $\tilde{\chi}_{\scriptscriptstyle Z_1 Z_2}$ logical noise component will be negligible relative to the $\tilde{\chi}_{\scriptscriptstyle Z_1 I}$ logical noise component. The same holds for the other logical noise components with a nontrivial logical operator on each side of $\tilde{\rho}$. Then, we may safely neglect the more general coherent terms and consider only the $(\tilde{L}_a \tilde{\rho})$ components.

\end{proof}

\section{Growth of Infidelity}

The expression for the average infidelity after $m$ applications of the noise channel from \cite{dugas2016efficiently} is an upper bound.
\begin{equation}
\label{eq: growth of infidelity generic bound}
	r_m \leq r m + \frac{(d-1) \Theta^2}{2d} m(m-1),
\end{equation}
where $r_m$ is the average infidelity after $m$ applications of a fixed noise channel, $r$ is the average infidelity after one application of the channel, $d$ is the dimension of the Hilbert space on which the channel acts, and $\Theta$ is the coherence angle. For anything save unitary or completely coherent channels, the upper bound has a linear component. We expect that this linear part is not only an upper bound, but that the average infidelity will grow linearly to lowest order.

Working in the Pauli transfer matrix representation, a unital noise channel is written as
\begin{equation}
	\left( \begin{array}{cccc}
		1 & 0 &  0 & \dots\\
		0 & 1-\lambda_{2} & \beta_{2,3} & \dots\\
		0 & \beta_{3,2} & 1- \lambda_3 & \\
		\vdots & \vdots & & \ddots
	\end{array} \right) .
\end{equation}
When channels are composed, we multiply the Pauli transfer matrices. After applying the same noise channel twice, we have diagonal entries
\begin{equation}
	(1-\lambda_j)^2 + \sum_{l | l \neq j} \beta_{j,l} \beta_{l,j}.
\end{equation}
After $m$ applications of the noise channel, the diagonal entries are
\begin{equation}
	(1-\lambda_j)^m + \binom{m}{1} (1-\lambda_j)^{m-1} \sum_{l | l \neq j} \beta_{j,l} \beta_{l,j}  + \cdots .
\end{equation}
Then, the infidelity after composing the channel $m$ times is proportional to
\begin{align}
    \sum_{j=1}^d 1-(1-\lambda_j)^m - \binom{m}{1} (1-\lambda_j)^{m-1} \sum_{i| i \neq j} \beta_{j,i} \beta_{i,j} - \cdots \nonumber
    \\
    = \sum_{j=1}^d m \lambda_j - \lambda_j^2 \frac{m(m-1)}{2} + \cdots -m \sum_{l |l \neq j} \beta_{j,l} \beta_{l,j}  \cdots .
\end{align}
To lowest order the infidelity grows proportional to $r$ the first term in the upper bound in equation (\ref{eq: growth of infidelity generic bound}).

\section{Diamond Distance Bound}

The diamond distance from identity can be bounded in terms of the average infidelity, $r$, and the sum of squares of the off-diagonal (coherent) components of the chi matrix. 

\begin{Lemma}
    In equation (\ref{eq:kueng-bound}), we upper bounded the diamond distance from identity for a channel by a function $f$ based on \cite{kueng2016comparing}. This function depended on the components of the Pauli transfer matrix for the channel. With a little algebra, we can show
    \begin{equation}
    \label{eq: diamond distance lemma statement}
        f^2 \leq c_1 \left( \sum_{i,j | i \neq j} \chi_{i,j}^2 \right) + c_2 r^2 .
    \end{equation}
    where the constants are given by $c_1 = d_L^2$ and $c_2 = 2 (d_L+1)^2$ and $d_L$ is the dimension of the logical space.
    \label{Lemma: diamond distance bound}
\end{Lemma}

\begin{proof}
    We start with equation (\ref{eq:kueng-bound}), and rewrite the Pauli transfer matrix in terms of chi matrix. We expand $(1-N_{i,i})^2$ and compare to $r^2$. Equation (\ref{eq: f^2 bound diamond distance}) reads
    \begin{equation}
        f^2 = \frac{1}{d_L^2-1} \left( \sum_{i,j|i \neq j} N_{i,j}^2 + \sum_l (1-N_{l,l})^2 \right),
        \label{eq: f equation}
    \end{equation}
    where $N$ is the Pauli transfer matrix representation of the noise channel. The diamond distance from identity is bounded by a constant times $f$. We can expand $f$ in terms of the chi matrix elements. Recall that we already have Lemma \ref{lemma:off-diagonal} concerning the off-diagonal elements. Also, the infidelity $r$ is related to the trace of the Pauli transfer matrix or the (0,0) element of the chi matrix.
    
    We can write the diagonal components of Pauli transfer matrix in terms of the diagonal components of the chi matrix in the following way:
    \begin{equation}
        N_{i,i} = \sum_{j \in C_i} \chi_{j,j} - \sum_{l \in A_i} \chi_{l,l} ,
    \end{equation}
    where the set $C_i$ includes all the Pauli operators $\sigma^j$ that commute with $\sigma^i$ and the set $A_i$ is all Pauli operators $\sigma^l$ that anticommute with $\sigma^i$. For example, in the case of a single-qubit channel
    \begin{equation}
        N_{1,1} = \chi_{0,0} + \chi_{1,1} - \chi_{2,2} - \chi_{3,3} .
    \end{equation}
    Then, we can sum over all the diagonal components of $N$ using the fact that the identity operator commutes with every operator.
    \begin{equation}
        \sum_{i=0}^{d_L^2} N_{i,i} = d_L^2 \chi_{0,0} ,
    \end{equation}
    where $d_L$ is the dimension of the logical space. Next, we can expand the diagonal term from equation (\ref{eq: f equation}):
    \begin{align}
        \sum_{i=0}^{d_L^2} (1-N_{i,i})^2 & = \sum_{i=0}^{d_L^2} \left( 1 - 2 N_{i,i} + N_{i,i}^2 \right) \nonumber
        \\
        & = d_L^2 - 2 d_L^2 \chi_{0,0} + d_L^2 \sum_{j} \chi_{j,j}^2 \nonumber
        \\
        & = d_L^2 \left(1 - \chi_{0,0} \right)^2 + d_L^2 \sum_{j| j \neq 0} \chi_{j,j}^2 \nonumber
        \\
        & = d_L^2 \left(\sum_{l| l \neq 0} \chi_{l,l} \right)^2 + d_L^2 \sum_{j| j \neq 0} \chi_{j,j}^2 ,
    \end{align}
    where we have used the trace preservation condition $\sum_{i} \chi_{i,i} = 1$. Because the noise channel is unitary, the diagonal components of the chi matrix are real and greater than 0. Then, we can bound
    \begin{equation}
        \sum_{i=0}^{d_L^2} (1-N_{i,i})^2 \leq 2 d_L^2 \left(\sum_{l| l \neq 0} \chi_{l,l} \right)^2 .
    \end{equation}
    When we substitute into equation (\ref{eq: f equation}) and use Lemma \ref{lemma:off-diagonal} for the off-diagonal terms, we have the following bound the diamond norm distance from identity:
    \begin{equation}
        \label{eq: Diamond norm chi matrix bound}
        D_\Diamond (N)^2 \leq \frac{d_L^2}{4} \left( \sum_{i,j|i \neq j} \chi_{i,j}^2 \right) + \frac{d_L^2}{2} \left( \sum_{l \neq 0} \chi_{l,l} \right)^2 .
    \end{equation}
    Finally, the average infidelity $r$ is given by
    \begin{equation}
        \label{eq: infidelity identity}
        r = \frac{d_L}{d_L+1}(1-\chi_{0,0}) = \frac{d_L}{d_L+1} \sum_{l \neq 0} \chi_{l,l}
    \end{equation}
    in the chi matrix representation. Equation (\ref{eq: diamond distance lemma statement}) follows.
\end{proof}

\end{document}